\documentclass{article}

\usepackage{soul}

\usepackage[affil-it]{authblk}
\usepackage[usenames,dvipsnames]{xcolor}
\usepackage{amsfonts}
\usepackage{amsmath,amsthm,amssymb,dsfont}
\usepackage{enumerate}
\usepackage[english]{babel}
\usepackage{graphicx}	
\usepackage[caption=false]{subfig}
\usepackage[margin=3cm]{geometry}   
\usepackage{url}
\usepackage{todonotes}
\usepackage{bbm}
\usepackage{booktabs}
\usepackage{subcaption}
\usepackage{tabularray}
\UseTblrLibrary{booktabs}

\usepackage[T1]{fontenc}

\usepackage{hyperref}
\usepackage{url}

\usepackage{pifont}
\newcommand{\cmark}{\ding{51}}%
\newcommand{\xmark}{\ding{55}}%

\usepackage{tikz}
\usetikzlibrary{chains}
\usetikzlibrary{fit}
\usepackage{pgflibraryarrows}		
\usepackage{pgflibrarysnakes}		

\usepackage{makecell}

\usepackage{epsfig}
\usetikzlibrary{shapes.symbols,patterns} 
\usepackage{pgfplots}
\pgfplotsset{compat=newest}

\usepackage{mathrsfs}

\usepackage{hyperref}
\hypersetup{colorlinks=true,citecolor=blue,linkcolor=blue,filecolor=blue,urlcolor=blue,breaklinks=true}

\usepackage[capitalise]{cleveref}

\usepackage{nicefrac}
\usepackage{mathtools}

\theoremstyle{plain}
\newtheorem{theorem}{Theorem}[section]
\newtheorem{lemma}[theorem]{Lemma}

\newtheorem{claim}[theorem]{Claim}

\newtheorem{corollary}[theorem]{Corollary}

\theoremstyle{definition}
\newtheorem{definition}[theorem]{Definition}

\newcommand{\mcA}{\mathcal{A}}
\newcommand{\mcB}{\mathcal{B}}
\newcommand{\mcC}{\mathcal{C}}
\newcommand{\mcD}{\mathcal{D}}
\newcommand{\mcE}{\mathcal{E}}
\newcommand{\mcF}{\mathcal{F}}

\newcommand{\mcH}{\mathcal{H}}
\newcommand{\mcI}{\mathcal{I}}

\newcommand{\mcK}{\mathcal{K}}
\newcommand{\mcL}{\mathcal{L}}
\newcommand{\mcM}{\mathcal{M}}
\newcommand{\mcN}{\mathcal{N}}

\newcommand{\mcP}{\mathcal{P}}

\newcommand{\mcR}{\mathcal{R}}
\newcommand{\mcS}{\mathcal{S}}

\newcommand{\mcZ}{\mathcal{Z}}

\newcommand{\mbE}{\mathbb{E}}

\newcommand{\mbH}{\mathbb{H}}

\newcommand{\mbR}{\mathbb{R}}

\usepackage{mathtools}
\DeclarePairedDelimiter\ceil{\lceil}{\rceil}

\newcommand{\ket}[1]{| #1 \rangle}
\newcommand{\bra}[1]{\langle #1 |}
\newcommand{\braket}[2]{\langle #1|#2\rangle}
\newcommand{\ketbra}[2]{|#1\rangle\!\langle#2|}
\newcommand{\id}{I}
\newcommand{\tr}{{\mathrm {tr}}}

\usepackage[most]{tcolorbox}
\begin{document}

\title{Chain rules for conditional entropies in quantum cryptography: limitations and improvements} 

\author[1]{Lewis Wooltorton\thanks{\href{mailto:lewis.wooltorton@ens-lyon.fr}{lewis.wooltorton@ens-lyon.fr}}}
\author[2]{Peter Brown\thanks{\href{mailto:peter.brown@telecom-paris.fr}{peter.brown@telecom-paris.fr}}}
\author[1]{Omar Fawzi\thanks{\href{mailto:omar.fawzi@ens-lyon.fr}{omar.fawzi@ens-lyon.fr}}}
\affil[1]{Inria, ENS de Lyon, UCBL, LIP, 46 Allee d’Italie, 69364 Lyon Cedex 07, France}
\affil[2]{Télécom Paris, LTCI, Institut Polytechnique de Paris, 19 Place Marguerite Perey, 91120 Palaiseau, France}

\date{May 28, 2026}

\maketitle

\begin{abstract}
    Security proofs in quantum cryptography rely on conditional entropies. In a many-round protocol, their estimation is a challenging task; one must account for the most general attacks by an eavesdropper, including those that are not independently and identically distributed (i.i.d.) across all rounds. Chain rules address this problem by relating the conditional entropy of a structured, but non-i.i.d.~process to a sum of entropy contributions from each round. They are a key ingredient in entropy accumulation theorems (EATs), which provide a versatile security proof framework for many protocols in quantum cryptography. Recently, chain rules in the setting of trusted devices have lead to tight i.i.d.~reductions at a finite number of rounds, and whether analogous results can be recovered in the device-independent (DI) setting has not been addressed. Surprisingly, we show that a natural tightening of the chain rule of Dupuis \emph{et al.} [Commun. Math. Phys. \textbf{379}, 867–913 (2020)] that would answer this question affirmatively cannot hold, highlighting a limitation of the current DI security proof approach. Nonetheless, we show that an intermediate improvement is possible by proving a new chain rule in this setting. Following the framework of Arqand \emph{et al.} [Phys. Rev. X \textbf{15}, 041013 (2025)], we use our chain rule to provide a slightly tighter version of the R\'enyi EAT in certain contexts. In addition, we provide a self-contained framework that unifies existing chain rules and compares their applications, framing our results in a broader context.    
\end{abstract}

\newpage

\tableofcontents 

\newpage

\section{Introduction}
\label{sec:intro}

Conditional entropies characterize the performance of many information processing tasks. A principle example is in quantum cryptography, where one or more honest parties wish to generate data that is secret from an adversary, Eve. Instances include the generation of private randomness or the distribution of a shared secret key~\cite{BB84,Ekert,Renner,ColbeckThesis}. Depending on the setting, Eve may access numerous features of the experiment; she might intercept and manipulate quantum signals sent along an insecure quantum channel, or, in the most severe scenario, prepare the devices used and exploit memory effects. Based on the assumptions made and the statistics collected while running the experiment, a protocol is secure if its output data is decoupled from any system that could be held by Eve. Conditional entropies provide a framework for proving such statements. Specifically, security can be expressed in terms of the entropy of the raw data generated during the protocol conditioned on all possible side information held by Eve at the end of its execution. 

The leftover hashing lemma~\cite{Renner,TSSR} is an example of this connection: it states that a secure key can be extracted from an imperfect source using two-universal hashing when the key length roughly equals the smooth conditional min-entropy of the source, a procedure known as randomness extraction~\cite{RK_05,Konig_2008, De_2012, Foreman2025cryptomite}. The amount of secure key that can be extracted is often more directly quantified by R\'enyi entropies~\cite{renyi61,ML_2013}. In particular, it was shown by Renner that bounding the sandwiched R\'enyi entropy of order $\alpha = 2$ is sufficient~\cite{Renner}, and this was extended to all $\alpha \in (1,2]$ by Dupuis~\cite{Dupuis23}. 

Thus, for many security proofs, the central challenge is bounding a type of conditional entropy between the raw data accumulated across many rounds and all relevant side information. This problem is challenging as the number of rounds is large and one must account for the most general attacks by Eve, which may not be independent and identically distributed (i.i.d.). An important step towards making such a problem tractable is the use of chain rules, which allow us to rewrite or bound the total entropy in terms of the entropies of the sub-systems. Applied to randomness extraction, the problem of characterizing the total entropy produced during the protocol can then often be reduced to characterizing the entropy produced during a single round, which can be easier to analyze. 

Many chain rules have been developed for different entropies and different applications. For example, quantum key distribution (QKD) can be described by the preparation of a state entangled between two sub-systems in a secure laboratory by one party, named Alice. Alice then transmits one of the sub-systems through an insecure quantum channel to her recipient, Bob. Eve may modify this signal before it reaches Bob, however she cannot influence the sub-system held in Alice's lab, resulting in a marginal constraint on the measured state~\cite{Curty_2004}. When considering the accumulated entropy over many such interactions, one must consider chain rules that incorporate these marginal constraints~\cite{VHB25,Fawzi2026,arqand2025}. A different situation arises in the study of device-independent (DI) protocols, where no aspect of the underlying state is characterized and the devices can hold an internal memory. These memories cannot leak arbitrary information to Eve, otherwise the protocol becomes trivially insecure. It is therefore common to consider the accumulated entropy from a sequence of quantum channels in which certain sub-systems can be updated each round, but not included in the side information~\cite{DFR,MetgerGEAT,arqand2024}.

For applications, a chain rule is often presented as an entropy accumulation theorem (EAT)~\cite{DFR}. This consists of a successive application of the rule to an $n$-round process, combined with a technical procedure called testing: in a standard cryptographic protocol, classical data may be collected each round that is used as an indicator of how well the devices are performing at the given task. For example, it may be used to estimate the quantum bit error rate between two separated devices during a QKD experiment, or the average score achieved in a nonlocal game. This testing data is crucial for inferring the presence of an eavesdropper, something that can be observed as excess noise. In this scenario the user can abort the protocol if they detect that the eavesdropper potentially has gained too much information about their secret. Otherwise, the testing procedure allows the user to certify an accumulated entropy of roughly $n$ times the worst case single round entropy compatible with the expected level of noise. This approach has been successful in establishing finite size security under general attacks for a variety of cryptographic protocols~\cite{Arnon-Friedman2018,LLR&,Metger2023,tupkary2026}.

The EAT and its variants are often referred to as an i.i.d.~reduction. They imply that the optimal attack by Eve is close to one in which she behaves independently and identically according to the expected level of noise in every round. Moerover, as the number of rounds increases, the difference between general and i.i.d.~attacks diminishes at a speed of roughly $O(1/\sqrt{n})$~\cite{DFR}. In fact, certain chain rules in the device-dependent regime are tight for a standard choice of R\'enyi entropy~\cite{VHB25,Fawzi2026,arqand2025}. That is, the worst case entropy over $n$ rounds is equal to $n$ times the worst case entropy of a single round, without the need for an $O(\sqrt{n})$ correction. This implies that an i.i.d.~attack is the optimal general attack for Eve even at a finite number of rounds.  While many advancements have been made in the DI setting~\cite{DFR,arqand2024}, it is yet to be seen if the same holds:

\vspace{0.2cm}

\noindent \textbf{Question 1:} \textit{Are there DI protocols for which i.i.d.~attacks are optimal at a finite number of rounds?}

\vspace{0.2cm}

In this manuscript, we identify a natural improvement to the chain rule of Dupuis, Fawzi and Renner~\cite{DFR} that would answer Question 1 affirmatively. Surprisingly, we show this chain rule cannot hold via an explicit counterexample, suggesting that a positive answer to Question 1 cannot follow from naturally tightening the existing DI proof framework. Nonetheless, we prove a new chain rule in terms of an intermediate conditional entropy that sits in between our no-go result and the existing chain rule of Ref.~\cite{DFR}. Equipped with this, we present a self-contained proof of the R\'enyi EAT by Arqand, Hahn and Tan~\cite{arqand2024} that is tighter in certain contexts; it is an open question which protocols benefit from this improvement, and we discuss some possibilities. As an additional contribution, we present a unified framework for comparing chain rules in quantum cryptography, capturing key examples in the literature. This overview discusses how different chain rules relate to each other, and to which settings they apply, allowing us to clearly highlight the difference between the device-dependent and device-independent setting and improve the accessibility of the topic. 

The remainder of the paper is structured as follows. In \cref{sec:setting}, we introduce our unified framework for chain rules, outlining the general quantum process and template chain rule we consider. In \cref{sec:survey}, we present a survey of the existing chain rules; a summary can be found in \cref{tab:summary}. \cref{sec:new_chain} contains a counterexample, followed by our improved chain rule. \cref{sec:REAT} presents our new variant of the entropy accumulation theorem based on this chain rule, and we conclude with a discussion in \cref{sec:disc}. 

\section{Cryptographic setting}\label{sec:setting}

\subsection{Notation and definitions}
Quantum systems are denoted in italics, e.g., $A$, and $\mcH_{A}$ is their associated Hilbert space that is assumed to be finite dimensional. The set of linear operators and positive semidefinite operators on $\mcH_{A}$ is denoted by $\mcB(\mcH_{A})$ and $\mcP(\mcH_{A})$, respectively. A normalized (sub-normalized) state $\rho$ is an element of $\mcP(\mcH_{A})$ with $\tr[\rho] = 1$ ($\tr[\rho] \leq 1$). The set of all normalized and sub-normalized states is denoted by $\mcD(\mcH_{A})$ and $\mcD_{\leq}(\mcH_{A})$, respectively. A quantum channel is a completely positive trace preserving (CPTP) map $\mcM : \mcB(\mcH_{A}) \to \mcB(\mcH_{A'})$, which we abbreviate to $\mcM:A \to A'$. We will also abbreviate sets of operators, e.g., $\mcD(\mcH_{A})$ to $\mcD(A)$, and $\mcD(\mcH_{A} \otimes \mcH_{B})$ to $\mcD(AB)$. When a channel $\mcM : A \to A'$ acts on a sub-system $A$ of a bipartite system $AB$, we will suppress its tensor product with the identity channel on $B$, $\mcI_{B}$. That is, we write $\mcM(\rho_{AB})$ for $(\mcM \otimes \mcI_{B})(\rho_{AB})$. We will also suppress the tensor product on states when it is clear from context. For example, we write $\sigma_{B}\rho_{AB}$ for $(\id_{A} \otimes \sigma_{B})\rho_{AB}$. Given multiple sub-systems $A_{1}$, $A_{2}$, ..., $A_{n}$, we denote the joint system by $A^{n} = A_{1}A_{2}\cdots A_{n}$. For a bipartite state $\rho_{AB} \in \mcD(AB)$, we define the conditional operator $\rho_{A|B} = \rho_{B}^{-\frac{1}{2}} \rho_{AB} \rho_{B}^{-\frac{1}{2}}$, where $\rho_{B} = \tr_{A}[\rho_{AB}]$. For a classical system $C$ with a finite alphabet $\mcC$, the set of probability distributions on $C$ is denoted by $\Delta(C)$. Throughout, $\log$ denotes the base 2 logarithm. 

Given a pair of operators $\rho, \, \sigma \in \mcP(A)$, the sandwiched R\'enyi divergence between $\rho$ and $\sigma$ is defined for $\alpha \in (1,\infty)$ as~\cite{ML_2013}
\begin{equation}
    D_{\alpha}(\rho \| \sigma) := \frac{1}{\alpha - 1} \log \Bigg( \frac{\tr\big[\big(\sigma^{\frac{1-\alpha}{2\alpha}} \rho \sigma^{\frac{1-\alpha}{2\alpha}} \big)^{\alpha}\big]}{\tr[\rho]} \Bigg) 
\end{equation}
if the support of $\rho$ is contained in the support of $\sigma$, and $D_{\alpha}(\rho \| \sigma) := + \infty$ otherwise. In the limit $\alpha \to 1$, the Umegaki relative entropy, $D(\rho \| \sigma) = \tr[\rho (\log \rho - \log \sigma)]$, is recovered. Using this R\'enyi divergence, the conditional (sandwiched) R\'enyi entropy of a bipartite state $\rho_{AB} \in \mcD(AB)$ can be defined in multiple ways:
\begin{equation}
    H^{\downarrow}_{\alpha}(A|B)_{\rho} := -D_{\alpha}(\rho_{AB} \| \id_{A} \otimes \rho_{B})
\end{equation}
and
\begin{equation}
    H^{\uparrow}_{\alpha}(A|B)_{\rho} := \sup_{\sigma_B \in \mcD(B)} -D_{\alpha}(\rho_{AB} \| \id_{A} \otimes \sigma_{B}).
\end{equation}
When $A$ and $B$ are classical random variables with a probability distribution $p_{AB}(a,b)$, the conditional R\'enyi entropies reduce to (see~\cite[Section 5.2.2]{Tomamichel_2016}) 
\begin{equation}
    H_{\alpha}^{\downarrow}(A|B)_{p} = \frac{1}{1-\alpha} \log \Bigg( \sum_{b} p_{B}(b)\sum_{a}p_{A|B}(a|b)^{\alpha}\Bigg) \label{eq:classical_down}
\end{equation}
and
\begin{equation}
    H_{\alpha}^{\uparrow}(A|B)_{p} = \frac{\alpha}{1-\alpha} \log \Bigg( \sum_{b} p_{B}(b) \Bigg[\sum_{a}p_{A|B}(a|b)^{\alpha}\Bigg]^{\frac{1}{\alpha}}\Bigg)\,. \label{eq:classical_up}
\end{equation}

\subsection{Problem setup} \label{sec:setup}    

Cryptographic protocols can be modeled using a sequence of quantum channels. Each channel accepts an input state on multiple sub-systems, some of which may be controlled by Eve, while others may be a trusted input made by the user. We refer to the quantum systems controlled by or accessible to Eve as side information. Each channel then outputs a system that is not accessible to Eve, and the collection of all such outputs constitutes the raw secret data generated by the protocol. Each channel also updates the side information of the previous round, and updates any internal memory held by the devices. The former might include generating new information that will become known to Eve, or modeling any operation performed by Eve on the sub-systems she controls. 

We now formalize this process. Consider a sequence of $n$ quantum channels, $\mcM_{i}$, for $i \in \{1,...,n\}$. Each channel has three inputs and three outputs, $\mcM_{i}:Y_{i-1}R_{i-1}E_{i-1} \to A_{i}R_{i}E_{i}$. The composition of two such channels is depicted in \cref{fig:chan}, and we describe the role of each sub-system below in the context of cryptographic applications.

\begin{figure}[h]
\includegraphics[width=10cm]{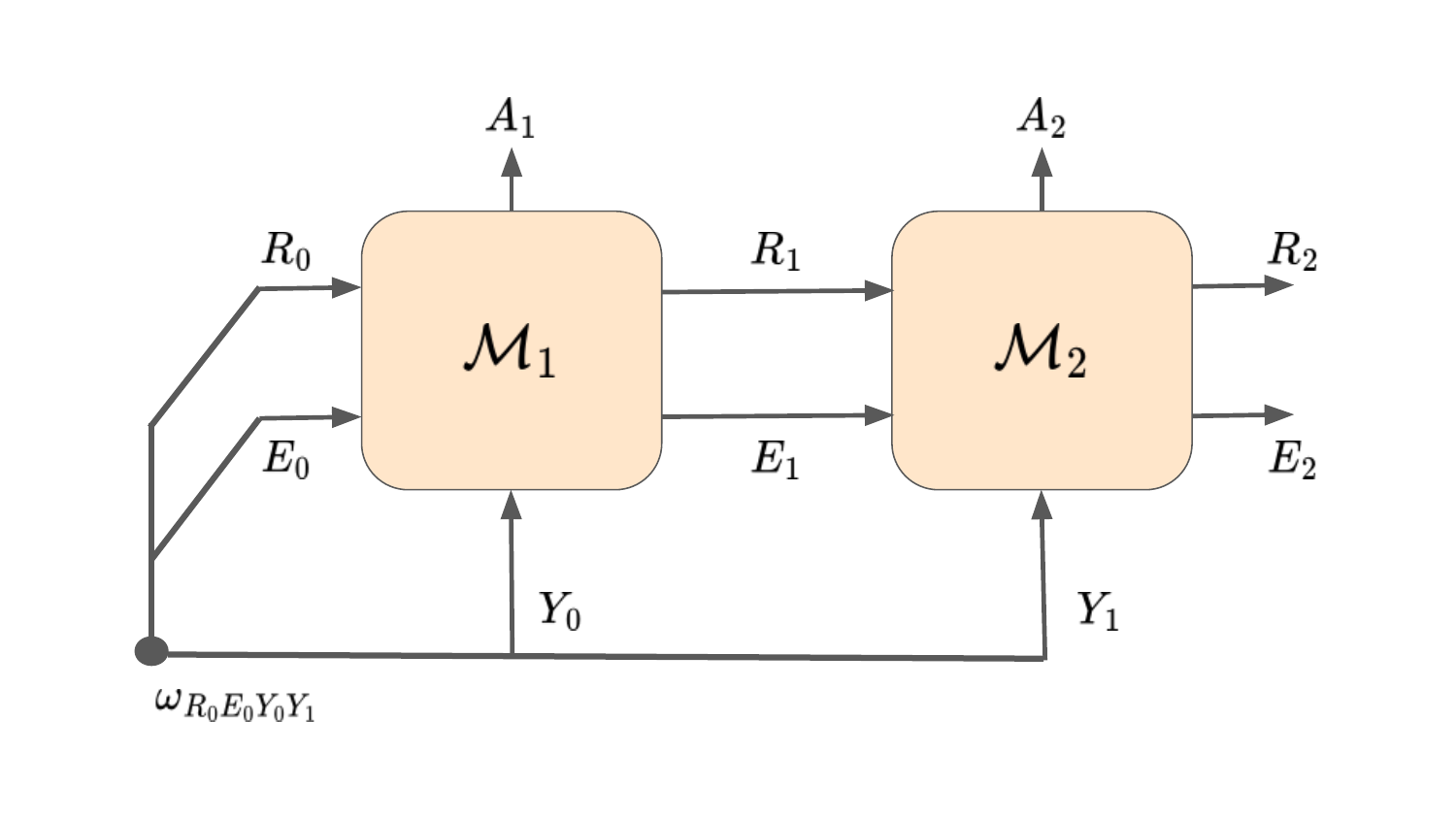}
\centering
\caption{Two steps of a general cryptographic protocol. The protocol accepts an input state $\omega_{R_{0}E_{0}Y_{0}Y_{1}}$, and performs a sequence of channels that output a piece of secret information $A_{i}$ at each step, and update the side information $E_{i}$ along with an internal memory $R_{i}$. The systems $Y_{i}$ are trusted, in the sense that the marginal $\omega_{Y_{0}Y_{1}}$ (or some linear constraints on it) is fixed and known to the user. The objective is to certify the entropy of the systems $A_{1}A_{2}$ conditioned on the side information $E_{2}$, $\mbH(A_{1}A_{2}|E_{n})$. Under certain assumptions on $\mcM_{1}$ and $\mcM_{2}$, a chain rule lower bounds this by the conditional entropy at each step; roughly speaking, the sum of $\mbH(A_{1}|E_{1})$ and $\mbH(A_{2}|A_{1}E_{2})$.} 
\label{fig:chan}
\end{figure}

\begin{itemize}
    \item The system $Y_{i-1}$ is a trusted input to $\mcM_{i}$. This means that, for any input state $\rho_{Y_{i-1}R_{i-1}E_{i-1}}$, the marginal $\rho_{Y_{i-1}} = \tr_{R_{i-1}E_{i-1}}[\rho_{Y_{i-1}R_{i-1}E_{i-1}}]$ satisfies some linear constraints. Furthermore, each sub-system $Y_{i-1}$ is consumed, and therefore does not appear as an output of $\mcM_{i}$. In the cryptographic context, this may account for a quantum system that Eve cannot manipulate. Such marginal constraints arise in prepare and measure QKD (see \cref{sec:marg1,sec:marg2}).

    \item The second type of system, $R_{i-1}$, is an untrusted input to $\mcM_{i}$ that can be updated, $R_{i-1} \to R_{i}$. This represents the internal memory of any devices used to run the protocol (see \cref{sec:markov,sec:NS}).

    \item The third type of system, $E_{i-1}$, is the side information associated to round $i-1$, and can be updated, $E_{i-1} \to E_{i}$, by $\mcM_{i}$. There are at least three scenarios this system can capture. Firstly, it may represent any side information present before the protocol commences, e.g., a system held by Eve that is entangled with the initial system held by the honest parties. Secondly, it can account for new information generated each round that becomes available to Eve. This could, for example, include measurement settings that are publicly announced, or measurement outcomes that are leaked to Eve. Thirdly, it can account for a sequential interaction between Eve and the honest parties. A key example of this is in prepare and measure QKD, where any system sent by Alice to Bob through an insecure quantum channel can be intercepted and modified by Eve (see \cref{sec:NS}). This gives her the opportunity to update her side information during the protocol's execution. 

    \item The final type of system, $A_{i}$, is the secret information generated by $\mcM_{i}$. These registers are not accessible to Eve at any point of the protocol, and the objective of the honest parties is to certify the entropy of $A_{1}A_{2}\cdots A_{n}$ conditioned on Eve's side information at the end of the sequence, $E_{n}$.     
\end{itemize}
We summarize the role of each sub-system in \cref{tab:labels}. 

\begin{table}
\centering
\begin{tabular}{||c  | c ||} 
 \hline
 System  & Meaning \\ [0.5ex] 
 \hline \hline
 $R_{i-1}$  & Untrusted input to $\mcM_{i}$  \\  
 \hline 
 $Y_{i-1}$ & Trusted input to $\mcM_{i}$ \\ \hline
 $E_{i-1}$ & Side information available to $\mcM_{i}$ \\
 \hline $A_{i}$  & Secret information generated by $\mcM_{i}$ \\[1ex] 
 \hline
\end{tabular}
\caption{Summary of the different sub-systems depicted in \cref{fig:chan}.}
\label{tab:labels}
\end{table}

\subsection{A template chain rule}

In the context of \cref{fig:chan}, the chain rules we are interested in generally take the following form.

\vspace{0.4cm}

\begin{tcolorbox}[colback=blue!5!white, colframe=blue!15!white, coltitle=blue!20!black, title=Template chain rule] 
Let $\mcS_{1}$, $\mcS_{2}$ and $\mcS_{3}$ be subsets of input states to the channels $\mcM_{1}$, $\mcM_{2}$ and $\mcM_{2}\circ \mcM_{1}$, respectively. Then our template chain rule takes the form
\begin{equation}
    \inf_{\omega_{3} \in \mcS_{3}} \mbH_{\alpha}(A_{1}A_{2}|E_{2})_{[\mcM_{2} \circ \mcM_{1}](\omega_{3})} \geq  \inf_{\omega_{1} \in \mcS_{1}} \mbH_{\alpha}(A_{1}|E_{1})_{ \mcM_{1}(\omega_{1})} +  \inf_{\omega_{2} \in \mcS_{2}} \widehat{\mbH}_{\beta}(A_{2}|E_{2})_{ \mcM_{2}(\omega_{2})}, \label{eq:temp}
\end{equation}
where $\alpha,\beta \in (1,\infty)$ and $\mbH_{\alpha}$ and $\widehat{\mbH}_{\beta}$ are quantum R\'enyi entropies (e.g., $H^{\uparrow}_{\alpha}$ or $H^{\downarrow}_{\alpha}$).
\end{tcolorbox}

\vspace{0.4cm}

\noindent More specifically, the sets $\mcS_{1}$, $\mcS_{2}$ and $\mcS_{3}$ are subsets of the form
\begin{equation*}
        \mcS_{1} \subset \mcD(Y_{0}R_{0}E_{0}),\ \ \
        \mcS_{2} \subset \mcD(Y_{1}R_{1}E_{1}), \ 
        \  \text{and} \ \ 
        \mcS_{3} \subset \mcD(Y_{0}Y_{1}R_{0}E_{0}).
\end{equation*}
\cref{eq:temp} includes an infimum over these sets because in cryptography, we are rarely interested in the entropy of a specific state. Instead, the state measured in the protocol can be manipulated, or even prepared by Eve, and we must evaluate the worst case entropy over all such states to obtain a security proof. In some cases, there are also well justified constraints to take into account, such as the linear constraints on the systems $Y_{i-1}$ discussed in \cref{sec:setup}, and $\mcS_{1}$, $\mcS_{2}$ and $\mcS_{3}$ can be chosen to reflect this. Furthermore, we may also extend these sets to include a register $\tilde{E}$ that purifies the input state $\omega_{Y_{i-1}R_{i-1}E_{i-1}}$. Considering both $\tilde{E}$ and the updated side information $E_{i}$ as conditioning registers in the R\'enyi entropy represents, in general, an overestimation of the amount of side information Eve could have before the channel in question is applied, that is,
\begin{equation*}
    \inf_{\omega \in \mcS_{i}} \mbH_{\alpha}(A_{i}|E_{i})_{\mcM_{i}(\omega)} \geq \inf_{\omega \in \tilde{\mcS}_{i}} \mbH_{\alpha}(A_{i}|E_{i}\tilde{E})_{\mcM_{i}(\omega)},
\end{equation*}
where $\tilde{S}_{i}$ is a subset of pure states on $Y_{i-1}R_{i-1}E_{i-1}\tilde{E}$.

\subsection{When does entropy accumulate?} \label{sec:when_ent}
In the survey that follows, we will look at different instances of the chain rule \eqref{eq:temp}. Before doing so, it is natural to ask why, throughout the entropy accumulation literature, so many variants of this chain rule have been considered. The answer is that entropy cannot accumulate in general for the specific setup of \cref{eq:temp}, as illustrated by the following example. Consider a fully classical instance of \cref{fig:chan}, where $E_{0}$ and $A_{2}$ are empty, and each channel $\mcM_{i}$ outputs a new piece of side information $B_{i}$ every round so that $E_{1} = B_{1}$, and $E_{2} = B_{1}B_{2}$. Without any restriction on the way the side information $B_{2}$ is generated, information about $A_{1}$ could be contained in $B_{2}$ via the internal memory $R_{1}$. This would imply
\begin{equation*}
    \mbH_{\alpha}(A_{1}|B_{1})_{\mcM_{1}(\omega)} > \mbH_{\alpha}(A_{1}|B_{1}B_{2})_{[\mcM_{1}\circ \mcM_{2}](\omega)} = \mbH_{\alpha}(A_{1}A_{2}|B_{1}B_{2})_{[\mcM_{1}\circ \mcM_{2}](\omega)},
\end{equation*}
resulting in a violation of \eqref{eq:temp}. It is therefore necessary to impose structural assumptions on the channel $\mcM_{2}$.

For example, consider again a setting where all registers are classical and suppose that we assume that the side information $B_{2}$ is generated independently of the past $A_{1}B_{1}$, as captured by the condition $p(a_1,b_1,b_2) = p(a_1,b_1)\, p(b_2)$ on the joint distribution\footnote{This is a special case of the requirement that $A_{1} \leftrightarrow B_{1} \leftrightarrow B_{2}$ forms a Markov chain, outlined in \cref{sec:markov}.}. Then a chain rule for $H_{\alpha}^{\downarrow}$ can be derived directly from \cref{eq:classical_down}:
\begin{equation*}
    \begin{aligned}
        2^{(1-\alpha)H_{\alpha}^{\downarrow}(A_{1}A_{2}|B_{1}B_{2})_{p}} &= \sum_{a_1,b_{1},a_2,b_{2}} p(b_{1})\, p(b_{2}) \, p(a_1|b_1)^{\alpha} p(a_2|a_1,b_1,b_2)^{\alpha} \\
        &= \sum_{b_{1}} p(b_{1})\, \frac{\sum_{a_1'}p(a_1'|b_1)^\alpha }{\sum_{a_1'}p(a_1'|b_1)^\alpha} \sum_{a_1,a_2,b_2}p(b_{2})\, p(a_1|b_1)^{\alpha} p(a_2|a_1,b_1,b_2)^{\alpha} \\
        &= \sum_{a_1',b_{1}} p(b_{1})\,p(a_1'|b_1)^\alpha \sum_{a_1,a_2,b_2} q(a_1|b_1) \, p(b_{2})\, p(a_2|a_1,b_1,b_2)^{\alpha} \\
        &\leq 2^{(1-\alpha)H_{\alpha}^{\downarrow}(A_1|B_1)_{p}} \cdot \max_{a_1,b_1} \sum_{a_{2},b_{2}} p(b_{2}) \, p(a_2|a_1,b_1,b_2)^{\alpha},
    \end{aligned}
\end{equation*}
where we defined the conditional distribution $q(a_1|b_1) = p(a_1|b_1)^{\alpha}/(\sum_{a_1'}p(a_1'|b_1)^{\alpha})$. Taking the logarithm and multiplying both sides by $1/(1-\alpha)$, we obtain the following chain rule:
\begin{equation*}
    H_{\alpha}^{\downarrow}(A_{1}A_{2}|B_{1}B_{2})_{p} \geq H_{\alpha}^{\downarrow}(A_1|B_1)_{p} + \min_{a_{1},b_{1}} H_{\alpha}^{\downarrow}(A_2|B_2,A_1=a_1,B_1=b_1)_{p},
\end{equation*}
where the minimum is taken over all possible values the past side information $A_1B_1$ can take. 

The assumptions made on the channel structure are specific to the cryptographic protocol being considered, and these assumptions are one of the main differences between the various chain rules. Furthermore, certain assumptions can result in tighter chain rules than others; in some cases, it is even possible for equality to hold in \eqref{eq:temp}, in which case we say the entropy in question is additive. 


\section{Survey of existing chain rules} \label{sec:survey}
We now present instances of the template chain rule \eqref{eq:temp} that have been established in the literature, and discuss their applications to cryptography. 

\subsection{Independent processes} 
We begin with the most straightforward instance of our template chain rule: when both the channels and the input state are independent. Consider the case where $R_{i}$ is empty, and $\mcM_{i}$ acts only on $Y_{i-1}$ to produce a fresh piece of side information $B_{i}$, $\mcM_{i}:Y_{i-1}\to A_{i}B_{i}$. If we further assume that the initial side information takes a product form $E_{0} = E_{0}'E_{1}'$, where $E_{i-1}'$ is entangled with $Y_{i-1}$ only, we have the substitution 
\begin{equation*}
    \underbrace{E_{0}'E_{1}'}_{E_{0}} \xrightarrow{\mcM_{1}} \underbrace{E_{0}'E_{1}'B_{1}}_{E_{1}} \xrightarrow{\mcM_{2}} \underbrace{E_{0}'E_{1}'B_{1}B_{2}}_{E_{2}}, 
\end{equation*}
and the initial state is of the form $\omega_{E_{0}'Y_{0}} \otimes \omega_{E_{1}'Y_{1}}$. Using the fact that the conditional entropy $H_{\alpha}^{\uparrow}$ is additive on tensor product states~\cite[Corollary 5.9.]{Tomamichel_2016}, we immediately have
\begin{equation}
    H_{\alpha}^{\uparrow}(A_{1}A_{2}|E_{0}'E_{1}'B_{1}B_{2})_{[\mcM_{2}\circ \mcM_{1}](\omega \otimes \omega)} = H_{\alpha}^{\uparrow}(A_{1}|E_{0}'B_{1})_{\mcM_{1}(\omega)} + H_{\alpha}^{\uparrow}(A_{2}|E_{1}'B_{2})_{\mcM_{2}(\omega)}. \label{eq:iid}
\end{equation}
This is an instance of \cref{eq:temp}. Furthermore, an independent processes will typically be a special case of a more general non-i.i.d.~scenario. For example, we could no longer restrict to input states of the form $\omega_{E_{0}'Y_{0}} \otimes \omega_{E_{1}'Y_{1}}$, or we could allow for a non-trivial memory register $R_{i}$. Consequently, \cref{eq:iid} will generally upper bound the left hand side of \cref{eq:temp}, and we can conclude that a chain rule is tight when this is matched by a lower bound. 

Additivity of the conditional R\'enyi entropy plays a key role in establishing fundamental results such as the fully quantum asymptotic equipartition property (AEP)~\cite{TCR}, which states that for many independent and identical copies of a quantum experiment, the output sequence almost always lies in the typical set. Phrased in a cryptographic context, as $n$ becomes very large, the smooth min-entropy of the outputs $A^n$ given the side information $(E'B)^n$ evaluated on the state $\mcM(\omega_{E'Y})^{\otimes n}$ is roughly equal to $n$ times the conditional von Neumann entropy $H(A|E'B)_{\mcM(\omega)}$. In particular, given a bound on the single round von Neumann entropy, the fully quantum AEP can be applied to establish security at a finite number of rounds under the i.i.d.~assumption.

\subsection{Markov channels}
\label{sec:markov}
Having shown how the fully i.i.d.~case fits into our framework, we now look at results that lift this assumption. Our first example is a chain rule derived by Dupuis, Fawzi and Renner \cite[Corollary 3.5]{DFR}. Here, the authors consider the setting in which the systems $Y_{0}$ and $Y_{1}$ are empty, and the channels $\mcM_{1}$ and $\mcM_{2}$ generate a new piece of side information $B_{1}$ and $B_{2}$, respectively. Specifically, letting $E_{0}$ represent the initial system held by Eve (that may be entangled with $R_{0}$), the side information is updated in the following way:
\begin{equation*}
    E_{0} \xrightarrow{\mcM_{1}} \underbrace{E_{0}B_{1}}_{E_{1}} \xrightarrow{\mcM_{2}} \underbrace{E_{0}B_{1}B_{2}}_{E_{2}} .
\end{equation*}
This is illustrated in \cref{fig:EAT}.

\begin{figure}[h]
\includegraphics[width=10cm]{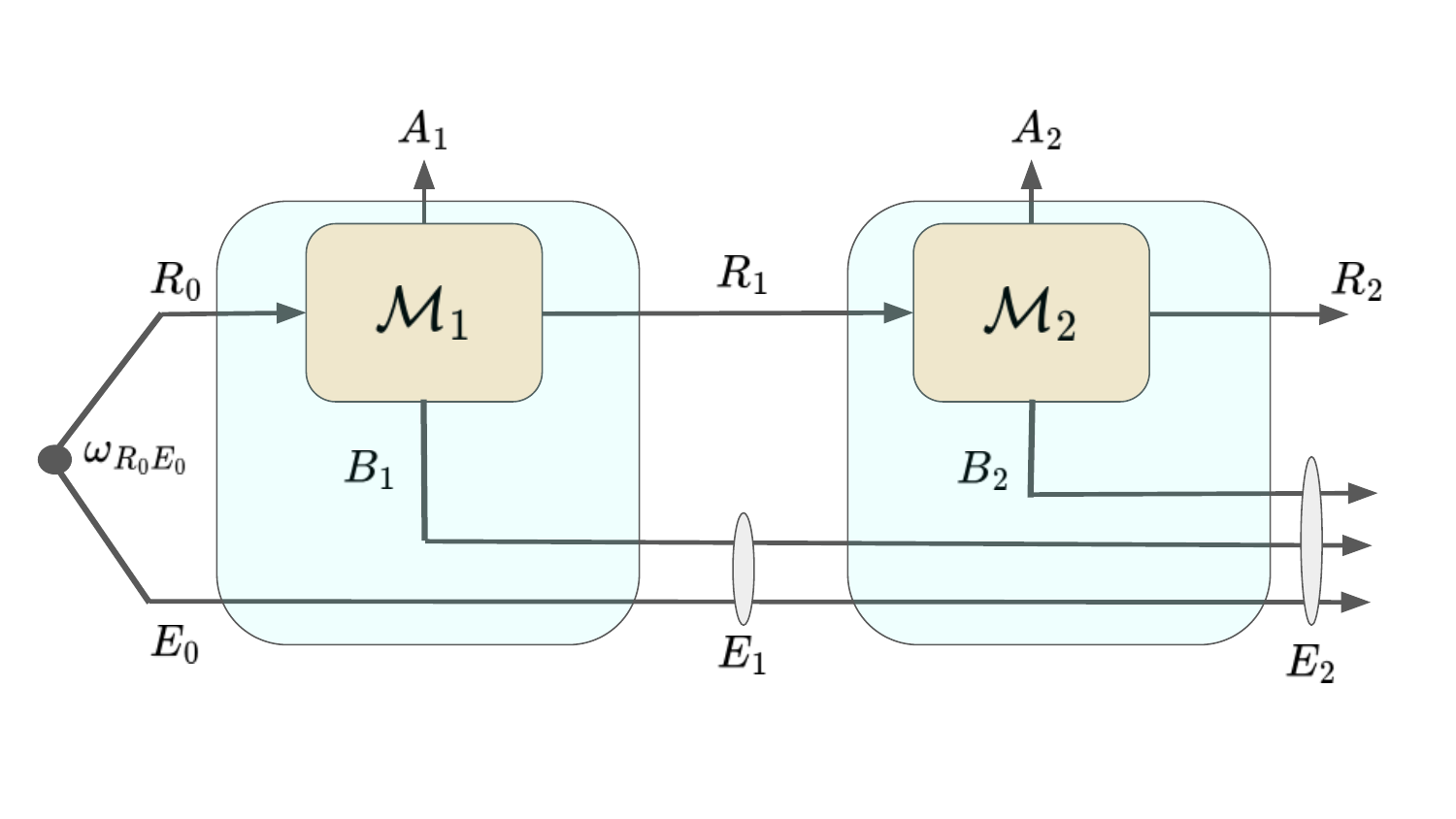}
\centering
\caption{The setting considered in~\cite{DFR}, which is a special case of~\cref{fig:chan}. The trusted inputs $Y_{0}Y_{1}$ are taken to be empty, and a new piece of side information $B_{2}$ is generated such that $A_{1}\leftrightarrow E_{0}B_{1} \leftrightarrow B_{2}$ forms a Markov chain.} 
\label{fig:EAT}
\end{figure}

In order to derive a chain rule of the form \eqref{eq:temp}, the authors require that $A_{1}$, $E_{0}B_{1}$ and $B_{2}$ form a Markov chain, $A_{1} \leftrightarrow E_{0}B_{1} \leftrightarrow B_{2}$. That is, for any initial state $\omega_{R_{0}E_{0}}$,
\begin{equation}
    I\big(A_{1} : B_{2} | E_{0}B_{1} \big)_{[\mcM_{2} \circ \mcM_{1}](\omega)} = 0\,, \label{eq:MI}
\end{equation}
where $I(A:B|C)_{\rho} := H(AC)_{\rho} + H(BC)_{\rho} - H(ABC)_{\rho} - H(C)_{\rho}$ is the conditional mutual information and $H(A)_{\rho} = -\tr[\rho \log \rho]$ is the von Neumann entropy. In words, \eqref{eq:MI} states that any new side information $B_{2}$ generated by $\mcM_{2}$ must be independent of the previous secret information $A_{1}$ when conditioned on the existing side information $E_{0}B_{1}$. Intuitively, this prevents any information about $A_{1}$ that is not already present in the side information $E_{0}B_{1}$ becoming accessible to Eve via $B_{2}$ (cf. \cref{sec:when_ent} and~\cite[Appendix C]{DFR}). 

We now state the chain rule from~\cite{DFR}.

\vspace{0.4cm}

\begin{tcolorbox}[colback=blue!5!white, colframe=blue!15!white, coltitle=blue!20!black, title=Chain rule for Markov channels~\cite{DFR}] 
Let $\omega_{R_{0}E_{0}} \in \mcD(R_{0}E_{0})$ be any initial state, $\mcM_{1}:R_{0} \to A_{1}B_{1}R_{1}$, $\mcM_{2}:R_{1} \to A_{2}B_{2}R_{2}$ and $\alpha \in (1,\infty)$. Then provided the output state $[\mcM_{2} \circ \mcM_{1}] (\omega)$ forms a Markov chain $A_{1} \leftrightarrow E_{0}B_{1} \leftrightarrow B_{2}$,  
\begin{equation}
    H_{\alpha}^{\downarrow}(A_{1}A_{2}|E_{0}B_{1}B_{2})_{[\mcM_{2} \circ \mcM_{1}] (\omega)} \geq H_{\alpha}^{\downarrow}(A_{1}|E_{0}B_{1})_{\mcM_{1} (\omega)} + \inf_{\omega' \in \mcD(R_{1}\tilde{E})} H_{\alpha}^{\downarrow}(A_{2}|B_{2}\tilde{E})_{\mcM_{2} (\omega')}, \label{eq:EAT_chain_1}
\end{equation}
where the system $\tilde{E}$ has the same dimension as $R_{1}$. 
\end{tcolorbox}

\vspace{0.4cm}

\noindent To see how \eqref{eq:EAT_chain_1} relates to the template chain rule \eqref{eq:temp}, consider the case when $Y_{0}$ and $Y_{1}$ are empty, $\mbH_{\alpha} = \widehat{\mbH}_{\beta} =H^{\downarrow}_{\alpha}$, and the following substitutions are made:
\begin{equation*}
    \begin{aligned}
        E_{1} \mapsto E_{0}B_{1}, \ \ 
        E_{2} \mapsto E_{0}B_{1}B_{2},\ \  
        \mcS_{1} = \mcS_{3} = \{ \omega_{R_{0}E_{0}} \}, \ \ \text{and} \ \
        \mcS_{2} = \mcD(R_{1}\tilde{E}).
    \end{aligned}
\end{equation*}
A modified version of \eqref{eq:EAT_chain_1} can also been obtained under the same Markov conditions~\cite[Corollary A.7]{Marwah2024}, when the R\'enyi entropies evaluated on the fixed state are replaced with their optimized versions~\cite[Corollary A.7]{Marwah2024},
\begin{equation}
    H_{\alpha}^{\uparrow}(A_{1}A_{2}|E_{0}B_{1}B_{2})_{[\mcM_{2} \circ \mcM_{1}] (\omega)} \geq H_{\alpha}^{\uparrow}(A_{1}|E_{0}B_{1})_{\mcM_{1} (\omega)} + \inf_{\omega' \in \mcD(R_{1}\tilde{E})} H_{\alpha}^{\downarrow}(A_{2}|B_{2}\tilde{E})_{\mcM_{2} (\omega')}. \label{eq:EAT_chain_mod}
\end{equation}
We emphasize that, similarly to \eqref{eq:EAT_chain_1}, the last entropy term is in terms of $H_{\alpha}^{\downarrow}$ and not $H_{\alpha}^{\uparrow}$. Whether \eqref{eq:EAT_chain_mod} can be tightened will be the subject of \cref{sec:counter}. 

\vspace{0.4cm}

\begin{tcolorbox}[colback=blue!5!white, colframe=blue!15!white, coltitle=blue!20!black, title=From chain rule to entropy accumulation,breakable] 
As discussed in the introduction, a chain rule can be presented as a type of entropy accumulation theorem~\cite{DFR,EAT2,MetgerGEAT,marwah2024b,VHB25,arqand2024,arqand2025}, which is often more useful for applications. Generally speaking, an EAT is the product of the following ingredients: 
\begin{enumerate}[(i)]
    \item A channel model for the process the user wishes to describe (e.g., an instance of \cref{fig:chan}). 
    \item An $n$-round conditional entropy that the user wishes to certify (e.g., the smooth min or max entropy, or the optimized R\'enyi entropy).  
    \item A chain rule for this or a related entropy that fits the assumptions of the process. 
    \item A means to handle testing (i.e., a way to translate observations made during the protocols execution to constraints on the final optimization problem).
\end{enumerate}
In the following, we give a rough sketch of how this procedure goes for the chain rule discussed in this subsection using DI cryptography as an example. We will omit the testing step for simplicity, and we note that all other chain rules in this survey can be lifted to an entropy accumulation theorem using a similar procedure. 

\vspace{0.2cm}

In a typical DI protocol, the honest parties provide random classical inputs to their devices each round $i \in \{1,...,n\}$, and obtain classical outputs. These can be modeled by the systems $B_{i}$ and $A_{i}$, respectively. Eve may prepare the initial state $\omega_{R_{0}E_{0}}$, and $E_{0}$ (the system held by Eve) is not modified at any stage of the protocol. Nothing further is assumed about the state or measurements used to run the protocol, and a non-trivial bound on the conditional entropy is established by witnessing nonlocal correlations among the devices~\cite{CK2,PABGMS,PAMBMMOHLMM}. The channels that describe this process are of the form $\mcM_{i}:R_{i-1} \to A_{i}B_{i}R_{i}$, and since every input $B_{i}$ is generated uniformly at random, the state of the $i^{\text{th}}$ round, $\rho_{A^{i}E_{0}B^{i}} = [\mcM_{i}  \circ \cdots \circ \mcM_{1}](\omega_{R_{0}E_{0}})$, satisfies
\begin{equation*}
    \tr_{A_{i}}[\rho_{A^{i}E_{0}B^{i}}] = \rho_{A^{i-1}E_{0}B^{i-1}} \otimes \rho_{B_{i}}.
\end{equation*}
It thus follows that $A^{i-1} \leftrightarrow E_{0}B^{i-1} \leftrightarrow B_{i}$ forms a Markov chain, and a recursive application of \eqref{eq:EAT_chain_1} with the substitution $A_{1}\mapsto A^{i-1}$, $B_{1} \mapsto B^{i-1}$, $A_{2} \mapsto A_{i}$ and $B_{2} \mapsto B_{i}$ yields
\begin{equation*}
    H_{\alpha}^{\downarrow}(A^{n}|B^{n}E_{0})_{[\mcM_{n} \circ \cdots \circ \mcM_{1}](\omega)} \geq \sum_{i=1}^{n} \inf_{\omega' \in \mcD(R_{i-1}\tilde{E})} H_{\alpha}^{\downarrow}(A_{i}|B_{i}\tilde{E})_{\mcM_{i}(\omega')}.
\end{equation*}
By relating the R\'enyi entropy $H_{\alpha}^{\downarrow}$ to the smooth min-entropy $H_{\text{min}}^{\epsilon}$ on the left hand side~\cite[Lemma B.10]{DFR}, and to the von Neumann entropy $H$ for every term in the summation~\cite[Lemma B.9]{DFR} (see also~\cite[Corollary IV.2]{EAT2}), one obtains (following a suitable choice of $\alpha = 1 + O(1/\sqrt{n})$) , 
\begin{equation}
    H_{\text{min}}^{\epsilon}(A^{n}|B^{n}E_{0})_{[\mcM_{n} \circ \cdots \circ \mcM_{1}](\omega)} \geq \sum_{i=1}^{n} \inf_{\omega' \in \mcD(R_{i-1}\tilde{E})} H(A_{i}|B_{i}\tilde{E})_{\mcM_{i}(\omega')} - O(\sqrt{n}). \label{eq:EAT_eq}
\end{equation}
This is the statement of the entropy accumulation theorem: the min-entropy of a Markov process is no greater than the sum of worst case von Neumann entropies arising from each round. Furthermore, for any $\epsilon < 1$, in the limit $n\to \infty$ the asymptotic i.i.d.~rate is recovered~\cite{TCR}, in which all the channels $\mcM_{i}$ are identical and independently generate the variables $A_{i}B_{i}$ at each step. 

\vspace{0.2cm}

To apply this theorem in practice, one needs a refined version that takes into account testing. In the DI case, this corresponds to restricting each infimum to input states that are compatible with the observed nonlocal statistics when the protocol does not abort (see \cref{sec:REAT,app:REAT} for an example). Furthermore, a method for lower bounding this constrained infimum is also required~\cite{TanDI,BrownDeviceIndependent,BrownDeviceIndependent2,KS25}.
\end{tcolorbox} 

\vspace{0.4cm}

Note that in the above example, all entanglement between the devices and Eve is distributed in the initial state $\omega_{E_{0}R_{0}}$, since Eve's system $E_{0}$ cannot be updated once the protocol commences. This may appear to be incompatible with common experimental practice, in which a new entangled state is generated and measured each round. However, Eve could equivalently prepare an initial state $\omega_{E_{0}'R_{0}'E_{1}'R_{1}'\cdots E_{n-1}'R_{n-1}'}$, and instruct each channel $\mcM_{i}$ to act only on $R_{i-1}'$. By writing $E_{0} = E_{0}'\cdots E_{n-1}'$ and $R_{i} = R_{i}'\cdots R_{i-1}'$, this is an instance of the setting illustrated in \cref{fig:EAT}, and Eve is still able to implement such strategies from a security analysis perspective. The distribution of all entanglement before the protocol starts is therefore not required in an experimental implementation.     

\subsection{Non-signaling channels}
\label{sec:NS}

While the Markov condition used to prove the EAT chain rule \eqref{eq:EAT_chain_1} suffices to prevent the secret information of past rounds being leaked by the side information generated in future rounds, it may be overly restrictive in some contexts. 

\begin{figure}
\centering
\begin{minipage}{.5\textwidth}
  \centering
  \includegraphics[width=\linewidth]{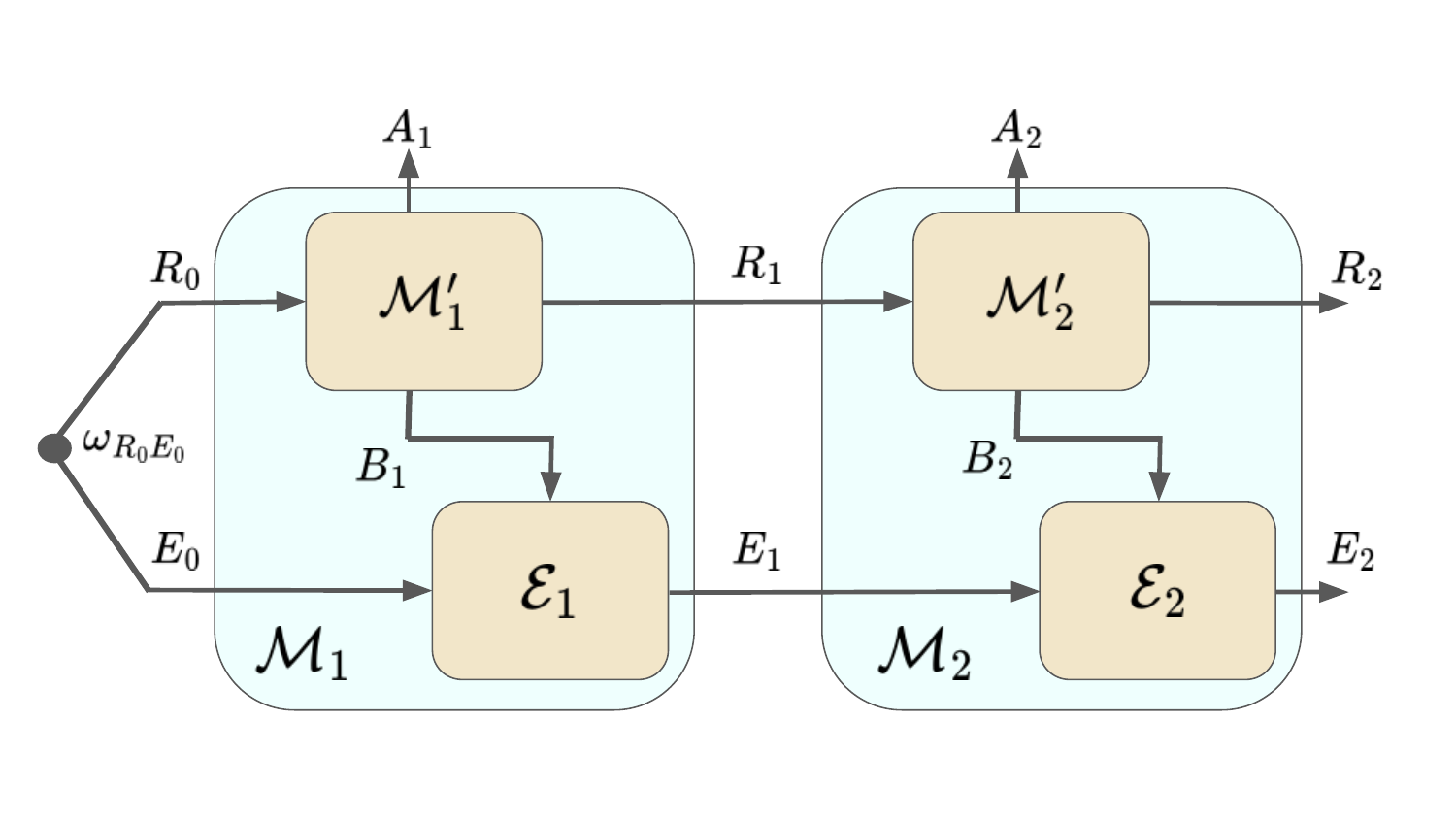}
  \caption*{(a)}
\end{minipage}%
\begin{minipage}{.5\textwidth}
  \centering
  \includegraphics[width=\linewidth]{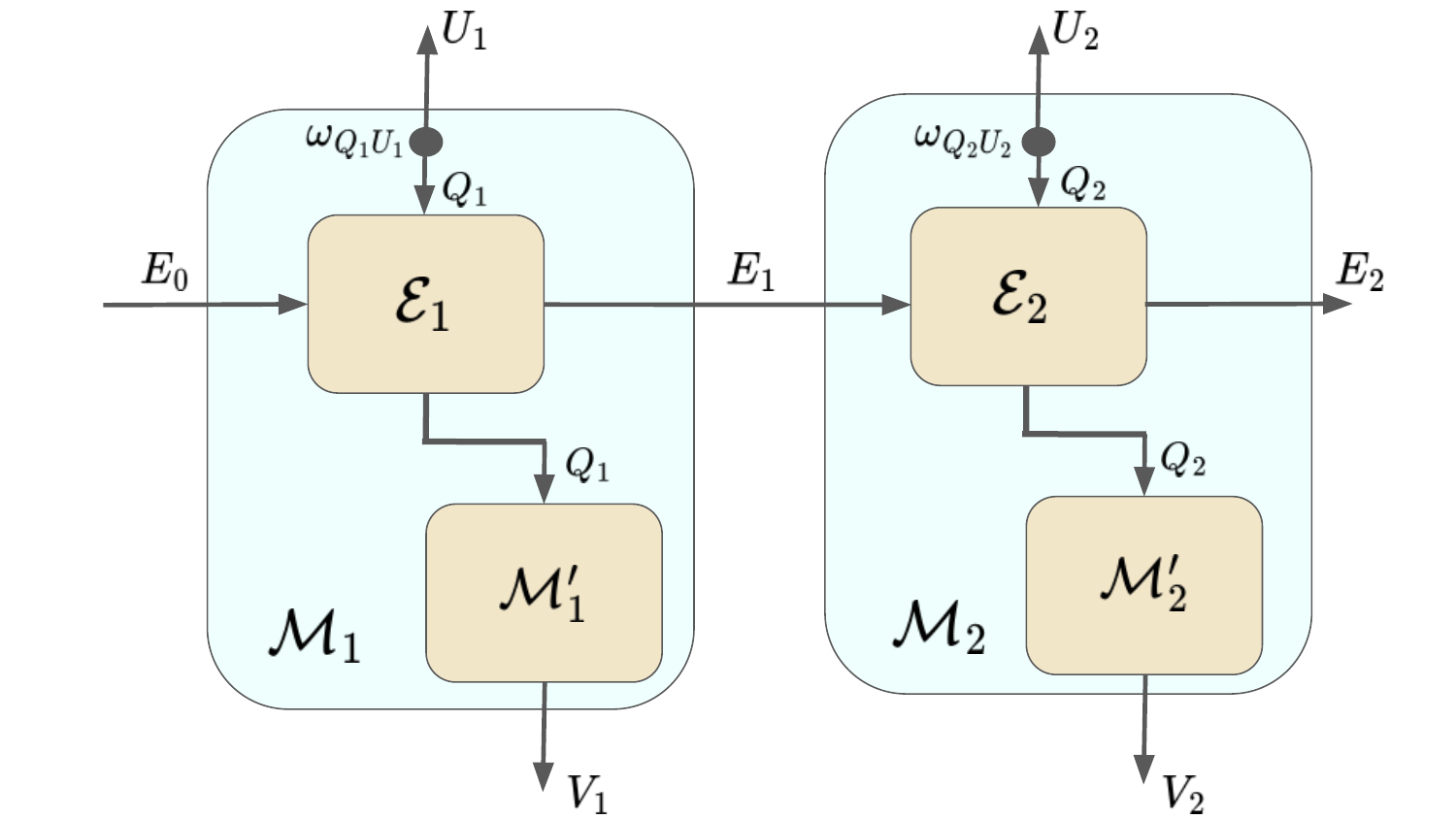}
  \caption*{(b)}
\end{minipage}
\caption{Two instances of a pair of non-signaling channels. Neither can be modeled by Markov channels of the form shown in \cref{fig:EAT}, illustrating how the non-signaling condition of the generalized entropy accumulation theorem~\cite{MetgerGEAT} is strictly weaker than the Markov condition~\cite{DFR}. (a) New side information $B_{i}$ is leaked each round to Eve, allowing her to perform an operation on her system during the runtime of the protocol. (b) The setup for prepare and measure QKD described in~\cite{Metger2023}. Roughly speaking, each round a cq-state $\omega_{Q_{i}U_{i}}$ is generated, where $U_{i}$ is Alice's raw key. Eve can then interact with the system $Q_{i}$ freely, which is then measured by Bob to obtain his raw key $V_{i}$. 
}
\label{fig:GEAT}
\end{figure}

Consider the scenario depicted in \cref{fig:GEAT} (a). Eve possesses an initial system $E_{0}$, and every round a channel $\mcM_{i}':R_{i-1}\to A_{i}B_{i}R_{i}$ generates a new piece of side information $B_{i}$ sampled randomly, and a new piece of secret information $A_{i}$. Now however, the system $B_{1}$ is immediately accessible to Eve after it has been generated. This gives Eve the opportunity to update her side information before the next round, which can be modeled by a quantum channel $\mcE_{1}:E_{0}B_{1} \to E_{1}$. Then, in the next round, Eve receives $B_{2}$, and can again update her information via a quantum channel $\mcE_{2}:E_{1}B_{2} \to E_{2}$. The overall channel $\mcM_{i}:E_{i-1}R_{i-1} \to A_{i}E_{i}R_{i}$ denotes the joint action of the channels $\mcM_{i}'$ and $\mcE_{i}$. Importantly, the system $E_{2}$ still contains no more information about $A_{1}$ than $E_{1}$. To see this, note that any input state $\omega_{A_{1}R_{1}E_{1}}$ to $\mcM_{2}$ satisfies the following equation:
\begin{equation}
    \rho_{A_{1}E_{2}} = [\tr_{A_{2}R_{2}} \circ \mcM_{2}](\omega_{A_{1}R_{1}E_{1}}) = [\mcE'_{2} \circ \tr_{R_{1}}](\omega_{A_{1}R_{1}E_{1}}), \label{eq:NSexample}
\end{equation}
where $\mcE'_{2}:E_{1} \to E_{2}$ is a quantum channel that samples $B_{2}$ randomly, and then applies $\mcE_{2}$ to $E_{1}B_{2}$. In words, the system $E_{2}$ can be generated without access to $R_{1}$. This implies
\begin{equation*}
    \mbH_{\alpha}(A_{1}|E_{2})_{\rho} = \mbH_{\alpha}(A_{1}|E_{2})_{\mcE'_{2}(\omega_{A_{1}E_{1}})} \geq \mbH_{\alpha}(A_{1}|E_{1})_{\omega_{A_{1}E_{1}}},
\end{equation*}
where we applied the data processing inequality. This indicates that a chain rule of the form \eqref{eq:temp} might hold. However, due to the way the registers $A_{1}$, $E_{1}$ and $E_{2}$ are generated, they do not form a Markov chain, and the chain rule \eqref{eq:EAT_chain_1} does not apply.

In~\cite[Lemma 3.6]{MetgerGEAT}, Metger et al. show that a chain rule of the form \eqref{eq:temp} holds when the Markov condition is relaxed to a non-signaling condition. This encompasses the example illustrated in \cref{fig:GEAT} (a). The systems $Y_{0}$ and $Y_{1}$ are again taken to be empty, and the channels can take a general form $\mcM_{1}:R_{0}E_{0}\to A_{1}R_{1}E_{1}$ and $\mcM_{2}:R_{1}E_{1}\to A_{2}R_{2}E_{2}$. The channel $\mcM_{2}$ is then said to be non-signaling between $R_{1}$ and $E_{2}$ if there exists a channel $\mcE_{2}' : E_{1} \to E_{2}$ such that for all input states $\omega_{R_{1}E_{1}}$,
\begin{equation}
    [\tr_{A_{2}R_{2}} \circ \mcM_{2} ](\omega_{R_{1}E_{1}}) = [\mcE_{2}' \circ \tr_{R_{1}}](\omega_{R_{1}E_{1}}).
\end{equation}
The above equation implies that the marginal state $\rho_{E_{2}} = [\tr_{A_{2}R_{2}} \circ \mcM_{2} ](\omega_{R_{1}E_{1}})$ is independent of $R_{1}$. This condition prevents $R_{1}$ from passing on information about previous outcomes (i.e., the system $A_{1}$) to $E_{2}$, and is directly satisfied by our example channel in \cref{fig:GEAT} (a), as shown by \cref{eq:NSexample}. The authors of~\cite{MetgerGEAT} show that this is a sufficient condition for entropy to accumulate.

\vspace{0.4cm}

\begin{tcolorbox}[colback=blue!5!white, colframe=blue!15!white, coltitle=blue!20!black, title=Chain rule for non-signaling channels~\cite{MetgerGEAT}] 
Let $\omega_{R_{0}E_{0}} \in \mcD(R_{0}E_{0})$ be any initial state, $\mcM_{1}:R_{0}E_{0} \to A_{1}R_{1}E_{1}$, $\mcM_{2}:R_{1}E_{1} \to A_{2}R_{2}E_{2}$ and $\alpha \in (1,2)$. Then provided the channel $\mcM_{2}$ is non-signaling between $R_{1}$ and $E_{2}$,  
\begin{equation}
    H_{\alpha}^{\downarrow}(A_{1}A_{2}|E_{2})_{[\mcM_{2} \circ \mcM_{1}] (\omega)} \geq H_{\alpha}^{\downarrow}(A_{1}|E_{1})_{\mcM_{1} (\omega)} + \inf_{\omega' \in \mcD(R_{1}E_{1}\tilde{E})} H_{\frac{1}{2-\alpha}}^{\downarrow}(A_{2}|E_{2}\tilde{E})_{\mcM_{2} (\omega')}, \label{eq:GEAT_chain_1}
\end{equation}
where the system $\tilde{E}$ has the same dimension as $R_{1}E_{1}$. 
\end{tcolorbox}

\vspace{0.4cm}

\noindent We observe that \eqref{eq:GEAT_chain_1} is another instance of \eqref{eq:temp}, where $Y_{0}$, and $Y_{1}$ are empty, $\mbH_{\alpha} =  H^{\downarrow}_{\alpha}$, $\widehat{\mbH}_{\beta} = H^{\downarrow}_{\frac{1}{2-\alpha}}$ and the following substitution is made:
\begin{equation*}
    \begin{aligned}
        \mcS_{1} = \mcS_{3} = \{ \omega_{R_{0}E_{0}} \}, \ \ \text{and} \ \
        \mcS_{2} = \mcD(R_{1}E_{1}\tilde{E}).
    \end{aligned}
\end{equation*}
Similar to the Markov chain rule in \cref{eq:EAT_chain_1}, a modified version of \eqref{eq:GEAT_chain_1} can be obtained with  optimized R\'enyi entropies on all but the right most term~\cite[Lemma 5.2]{arqand2024_mutual}: 
\begin{equation*}
    H_{\alpha}^{\uparrow}(A_{1}A_{2}|E_{2})_{[\mcM_{2} \circ \mcM_{1}] (\omega)} \geq H_{\alpha}^{\uparrow}(A_{1}|E_{1})_{\mcM_{1} (\omega)} + \inf_{\omega' \in \mcD(R_{1}E_{1}\tilde{E})} H_{\frac{1}{2-\alpha}}^{\downarrow}(A_{2}|E_{2}\tilde{E})_{\mcM_{2} (\omega')}.
\end{equation*}

As we saw with \eqref{eq:EAT_chain_1}, the chain rule \eqref{eq:GEAT_chain_1} can be used to derive a version of the entropy accumulation theorem with weaker assumptions on the channels $\mcM_{i}$. Specifically, the generalized EAT (GEAT) requires each channel $\mcM_{i}:R_{i-1}E_{i-1} \to A_{i}R_{i}E_{i}$ to be non-signaling between $R_{i-1}$ and $E_{i}$, i.e., there exists a channel $\mcE_{i}'$ such that 
\begin{equation*}
    \tr_{A_{i}R_{i}} \circ \mcM_{i}  = \mcE_{i}' \circ \tr_{R_{i-1}}.
\end{equation*}
Then a statement analogous to \cref{eq:EAT_eq} holds. Testing can also be incorporated, and the GEAT has found applications in quantum cryptography that go beyond those that can be handled using the EAT, including blind randomness expansion~\cite{MillerBlind} (see also~\cite{MetgerGEAT}) and prepare and measure QKD \cite{Metger2023}. 

For example, the security of prepare and measure QKD was established using the GEAT in~\cite{Metger2023}. Roughly speaking, this setting includes a set of quantum systems $Q_{1}\cdots Q_{n}$ prepared by Alice, along with an initial system $E_{0}$ held by Eve that is unentangled with $Q^{n}$. On round $i \in \{1,...,n\}$, Alice prepares a state $\omega_{U_{i}Q_{i}}$, where $U_{i}$ is a classical system holding her raw key bit for that round. The system $Q_{i}$ is then transmitted from Alice to Bob across an insecure channel held by Eve, allowing her to perform a quantum channel $\mcE_{i}:E_{i-1}Q_{i} \to E_{i}Q_{i}$. After this, Bob receives the system $Q_{i}$ from Eve, and measures it to obtain his raw key $V_{i}$, described by a channel $\mcM_{i}':Q_{i} \to V_{i}$. This overall process can be described by a single channel $\mcM_{i}:E_{i-1} \to U_{i}V_{i}E_{i}$ that we illustrate in \cref{fig:GEAT} (b). In an actual QKD protocol, the channels $\mcM_{i}$ also include some public announcements and additional processing of the registers $U_{i}$ and $V_{i}$. This would form Alice's sifted key, $A_{i}$, and release additional information to Eve, $B_{i}$. We refer the reader to~\cite{Metger2023} for full details. Note that in this setting, each system $R_{i}$ is empty, and therefore the non-signaling condition between $R_{i-1}$ and $E_{i}$ holds trivially. Furthermore, Eve cannot act on multiple $Q_{i}$ systems at once. Specifically, each channel $\mcE_{i}$ is assumed to act on $Q_{i}$ only, whereas a more general update channel that maps $E_{i-1}'Q^{n}$ to $E_{i}'Q^{n}$ is not permitted. Alice must therefore only send the system $Q_{i+1}$ once Bob has measured the system $Q_{i}$, which can be enforced by Bob announcing when his measurement has taken place across an authenticated communication channel. As a result, the repetition rate of the protocol is limited, since Alice must always wait until Bob has made this announcement before sending the next signal. 

\subsection{Independent and identical channels with linear constraints}
\label{sec:marg1}

In the previous subsections, we discussed chain rules used to prove the EAT and GEAT. Both fit into the setting described in \cref{fig:chan}, where the trusted input systems $Y_{i}$ are empty, while the untrusted input systems $R_{i}$ can be nontrivial. Such cases are particularly relevant to DI protocols, where $R_{i}$ represents the internal system of an untrusted device after round $i$. This may encompass, for example, entangled sub-systems that are measured to obtain the classical outcomes of future rounds, an internal memory, or even pre-programmed instructions by Eve that force the device to behave deterministically. We now explore chain rules in a very different setting, where there are trusted inputs and the untrusted inputs are empty. 

To illustrate the significance of this setting, it is instructive to consider again the security of prepare and measure QKD. It is not immediate that this problem can be tackled using the EAT, since the prepare and measure setting requires a model of side information that can be updated each round. The GEAT circumvents this issue by allowing for a more general type of side information at the cost of a limited repetition rate of the protocol. An alterative security approach is to employ the so called ``source replacement scheme''~\cite{BBM92,Scarani09}. Rather than preparing a classical quantum state $\omega_{UQ}$ each round, Alice could equivalently prepare an entangled state $\omega_{Q_AQ_B}$, and later measure the quantum system $Q_B$ to obtain her raw key $U$. We then give Eve control of the entangled state $\omega_{Q_AQ_B}$, reducing the security analysis to that of an entanglement based protocol. However, if Eve has full control of $\omega_{Q_AQ_B}$ in the source replacement scheme, no secret key can be certified; the marginal state $\omega_{Q_A}$ of $\omega_{Q_AQ_B}$ must be constrained to equal a known state $\sigma_{Q_A}$ to obtain non-trivial key rates. This captures the fact that in a prepare and measure protocol, Eve can only access the system $Q_B$. Marginal constraints such as these are not incorporated into the EAT or GEAT directly. Bäuml \textit{et al.}~\cite{Bauml2024} overcame this challenge by adding an additional tomography step for Alice in randomly chosen rounds. The data she collects from these rounds can be used to verify that she has the correct marginal state, or abort the protocol if not. However, the final key rate incurs a penalty arising from the statistical uncertainty of this test. A security framework that handles marginal constraints more directly via a chain rule with trusted inputs would circumvent these issues. 

Brown and Van Himbeeck~\cite{VHB25} provide the first result of this type. The systems $R_{i}$ are empty, there are linear constraints on the inputs $Y_{i}$, and each channel is independent and identical. Furthermore, under these conditions an exact additivity result can be established.

\vspace{0.4cm}

\begin{tcolorbox}[colback=blue!5!white, colframe=blue!15!white, coltitle=blue!20!black, title=Chain rule for independent and identical CP maps with linear constraints \cite{VHB25}]
Let $\mcM:Y \to AB$ be a quantum to classical CP map, $\mcN : Y \to Q_{A}$ be a quantum channel, $\sigma_{Q_{A}} \in \mcD(Q_{A})$ be a state and $\alpha \in (1,\infty)$. Then for any positive integer $n$,
\begin{equation}
    \inf_{\omega \in \mcS_{n}(\mcN,\sigma)} H_{\alpha}^{\uparrow}(A^{n}|B^nE^n)_{\mcM^{\otimes n}(\omega)} = n \inf_{\omega \in \mcS_{1}(\mcN,\sigma)} H_{\alpha}^{\uparrow}(A|BE)_{\mcM(\omega)}, \label{eq:tensorChain}
\end{equation}
where
\begin{equation*}
    \mcS_{n}(\mcN,\sigma) := \big\{ \omega_{E^nY^n} \in \mcD(E^nY^n) \ : \ \mcN^{\otimes n}(\omega_{Y^{n}}) = \sigma_{Q_{A}}^{\otimes n} \big\}.
\end{equation*}
\end{tcolorbox}

\vspace{0.4cm}

\noindent To see that the chain rule \eqref{eq:tensorChain} is another instance of the template \eqref{eq:temp}, consider the case $n = 2$ and the following substitution: $R_{0}, \ R_{1}$ and $R_{2}$ are empty, $\mbH_{\alpha} = \widehat{\mbH}_{\beta} = H_{\alpha}^{\uparrow}$, and\footnote{Strictly speaking, to ensure $\mcS_{1}$ and $\mcS_{2}$ match the form given below \cref{eq:temp}, we should take the tensor product with an arbitrary state on $E_{1}$ and $B_{1}E_{0}$, respectively. We have not included this for simplicity.} 
\begin{equation*}
    \begin{gathered}
        E_{0} \mapsto E_{0}E_{1}, \ \ E_{1} \mapsto B_{1}E_{0}E_{1}, \ \ 
        E_{2} \mapsto B_{1}B_{2}E_{0}E_{1}, \\
        \mcS_{1} \mapsto \{ \omega_{E_0Y_0} \ : \ \omega_{E_0Y_0} \in \mcS_{1}(\mcN,\sigma) \}, \\  \mcS_{2} \mapsto \{ \omega_{E_1Y_1}  \ : \ \omega_{E_1Y_1} \in \mcS_{1}(\mcN,\sigma) \} \ \text{and} \\
        \mcS_{3} \mapsto \mcS_{2}(\mcN,\sigma).
    \end{gathered}
\end{equation*}
We then view the channel $\mcM_{1} : E_{0}Y_{0} \to A_{1}E_{1}$ as $\mcM$ tensored with the identity channel on $E_{0}E_{1}$, and the channel $\mcM_{2} : E_{1}Y_{1} \to A_{2}E_{2}$ as $\mcM$ tensored with the identity channel on $B_{1}E_{0}E_{2}$. It can then be verified that $\mcM_{2} \circ \mcM_{1} = \mcM^{\otimes 2}$. This is illustrated in \cref{fig:tensor1}, and under this substitution the template chain rule \eqref{eq:temp} reduces to \eqref{eq:tensorChain}. Furthermore, because \cref{eq:tensorChain} allows for CP maps, it implies a chain rule for the more general family of $f$-weighted R\'enyi entropies used in the security analysis of QKD~\cite{VHB25} (see also \cref{sec:REAT} for a brief discussion).       

As a concrete example, consider the case where the system $Y$ decomposes into two sub-systems $Q_{A}$ and $Q_{B}$, and $\mcN = \tr_{Q_{B}}$ is the partial trace. Then any state $\omega_{E^nQ_{A}^nQ_{B}^n} \in \mcS_{n}(\mcN,\sigma)$ satisfies $\tr_{E^nQ_{B}^n}[\omega_{E^nQ_{A}^nQ_{B}^n}] = \sigma_{Q_{A}}$, i.e., the marginal state on the system $Q_{A}$ is fixed. Thus, the additivity result \eqref{eq:tensorChain} can accommodate marginal constraints directly, enabling a security proof of prepare and measure QKD via the source replacement scheme and without the need for tomographic tests on the systems $Q_{A}^{n}$. Furthermore, the security analysis does not restrict the repetition rate of the protocol. Specifically, in the setting of \cref{eq:tensorChain}, Eve prepares an initial entangled state $\omega_{E^nQ_{A}^nQ_{B}^n}$ subject to a marginal constraint on $Q_{A}^{n}$. Alice and Bob then measure $Q_{A}^{n}$ and $Q^{n}_{B}$ to obtain their raw data $U^{n}$ and $V^{n}$, along with any announcement data $B^{n}$. The system $A^{n}$ represents the raw key of Alice, and is generated by processing $(UVB)^{n}$. This is equivalent to a prepare and measure setting in which Alice prepares an initial cq-state $\omega_{U^nQ_{B}^n}$ and sends $Q^{n}_{B}$ to Eve, who applies a channel $\mcE^{\otimes n} : Q^{n}_{B} \to E^nQ_{B}^n$. Bob then receives and measures $Q^{n}_{B}$ to obtain $V^nB^n$, after which $A^{n}$ can be generated. Crucially, it is not necessary that Alice waits for Bob to measure the $i^{\text{th}}$ $Q_B$ system before sending system $i+1$, as was the case for the GEAT analysis.    

Another difference to the EAT and GEAT chain rules is that \eqref{eq:tensorChain} is a statement for $H_{\alpha}^{\uparrow}$, rather than $H_{\alpha}^{\downarrow}$, and no loss in the parameter $\alpha$ is incurred. Dupuis~\cite{Dupuis23} showed that $H_{\alpha}^{\uparrow}(A^{n}|E)_{\rho}$ for $\alpha \in (1,2]$ can control the soundness error when distilling a secret key $S^{l}$ of length $l$ from a raw key $A^{n}$ against quantum side information $E$ using two-universal hashing~\cite{Carter79,Stinson92}:
\begin{equation}
    \mbE_{f \in \mcF} \big\| \rho_{f(A^{n})E} - \frac{\id_{S^l}}{2^{l}} \otimes \rho_{E} \big\|_{1} \leq 2^{\frac{2}{\alpha} - 1}2^{\frac{\alpha - 1}{\alpha}( l - H_{\alpha}^{\uparrow}(A^{n}|E)_{\rho})}, \label{eq:PA}
\end{equation}
where $\mcF$ is a family of two-universal hash functions\footnote{$\mcF$ is a family of two-universal hash functions from $\{0,1\}^{n} \to \{0,1\}^{l}$ if, for any two functions $f, \ g\in \mcF$ sampled uniformly at random, $\text{Pr}[f(x) = g(x)] \leq 2^{-l}$ for all $x \in \{0,1\}^{n}$.} from $\{0,1\}^{n} \to \{0,1\}^{l}$. The special case when $\alpha = 2$ was previously established by Renner~\cite{Renner}, ultimately leading to a result in terms of $H_{\text{min}}^{\epsilon}$\footnote{This can be obtained from $H_{2}^{\uparrow}$, and then handled using, e.g., the EAT or GEAT}. The advantage of using the general statement \eqref{eq:PA} is that by working with the R\'enyi entropy directly, no penalty is incurred by reducing to the smooth min-entropy\footnote{Specifically, the first step of the EAT proof~\cite{DFR} is to lower bound $H_{\text{min}}^{\epsilon}$ by $H_{\alpha}^{\uparrow}$ for $\alpha \in (1,2]$, incurring a penalty in the second order term (in $n$).}. Chain rules such as \cref{eq:tensorChain} that satisfy $\mbH_{\alpha} = \widehat{\mbH}_{\beta} = H_{\alpha}^{\uparrow}$ (i.e., do not use the bound $H_{\alpha}^{\uparrow} \geq H_{\alpha}^{\downarrow}$ or have a loss in $\alpha$) can provide a tight finite size analysis with respect to the randomness extraction result in \eqref{eq:PA}~\cite{VHB25}. This is because, unlike $H_{\text{min}}^{\epsilon}$, $H_{\alpha}^{\uparrow}$ is additive when Eve's attack is i.i.d., matching the lower bound given by applying the chain rule. 

\begin{figure}
\centering
\begin{minipage}{.5\textwidth}
  \centering
  \includegraphics[width=\linewidth]{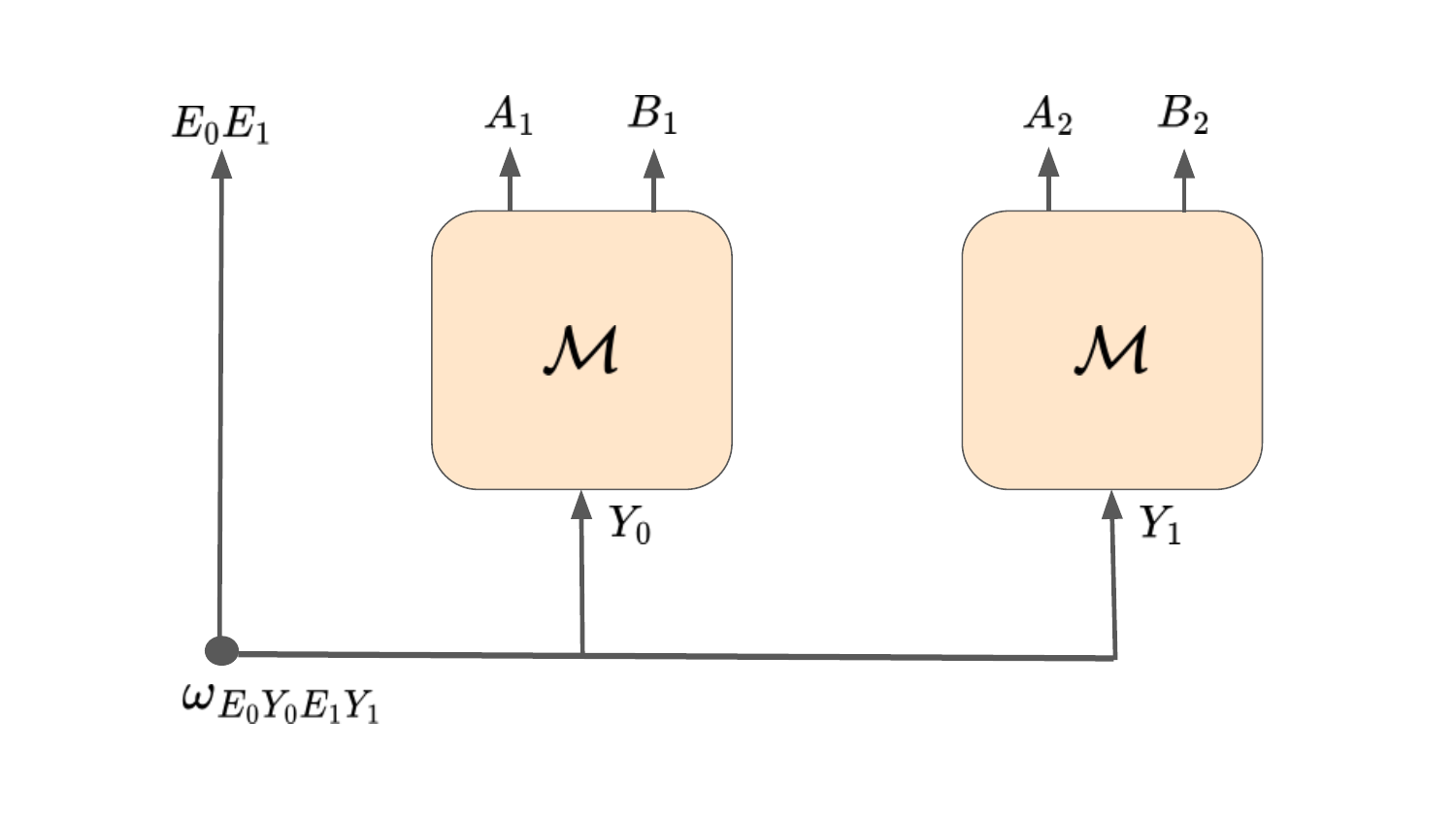}
  \caption*{(a)}
\end{minipage}%
\begin{minipage}{.5\textwidth}
  \centering
  \includegraphics[width=\linewidth]{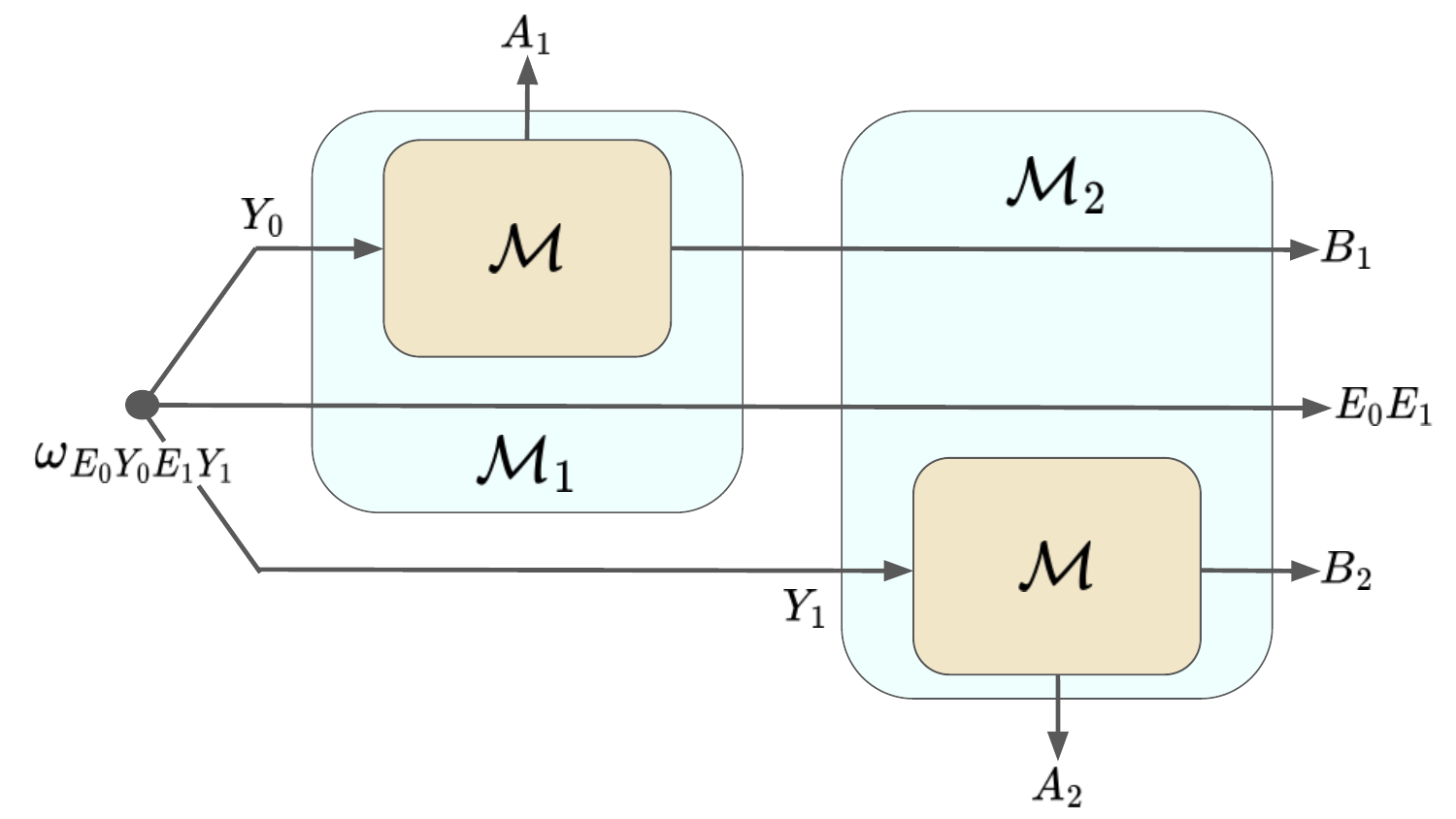}
  \caption*{(b)}
\end{minipage}
\caption{(a) The i.i.d.~channel structure considered in~\cite{VHB25}, where each system $Y$ of the state $\omega_{E_{0}Y_{0}E_{1}Y_{1}}$ is independently measured by the same channel $\mcM$ to obtain a pair of classical outcomes $A$ and $B$, where $A$ corresponds to secret information, and $B$ side information. The state $\omega_{E_{0}Y_{0}E_{1}Y_{1}}$ can however be arbitrarily entangled across the different sub-systems, subject to linear constraints on each system $Y$. (b) The same scenario illustrated in (a), now viewed as the sequential application of the channel $\mcM_{1}$ followed by $\mcM_{2}$.}
\label{fig:tensor1}
\end{figure}

\subsection{Non-identical channels with linear constraints}
\label{sec:marg2}
In the work of Arqand and Tan~\cite{arqand2025} and Fawzi \textit{et al.}~\cite{Fawzi2026}, similar results to~\cite{VHB25} are obtained under weaker conditions. Firstly, both works obtain an additivity result when the channels admit a tensor product structure, but are not necessarily identical each round. Secondly, the trusted input to each channel can have different constraints; for~\cite{arqand2025}, a different marginal state, and for~\cite{Fawzi2026}, a different linear constraint.    

We begin by outlining the chain rule derived by Arqand and Tan~\cite[Theorem 3.1]{arqand2025}. 

\vspace{0.4cm}

\begin{tcolorbox}[colback=blue!5!white, colframe=blue!15!white, coltitle=blue!20!black, title=Chain rule for sequential channels with marginal constraints~\cite{arqand2025}]
Let $\mcM_{1}:E_{0}Y_{0} \to A_{1}E_{1}$ and $\mcM_{2}:E_{1}Y_{1} \to A_{2}E_{2}$ be quantum channels, $\sigma_{Y_{0}} \in \mcD(Y_{0})$ and $\tau_{Y_{1}} \in \mcD(Y_{1})$ be states, and $\alpha \in (1,\infty)$. Then
\begin{multline}
    \inf_{\substack{\omega \in \mcD(E_{0}Y_{0}Y_{1}\tilde{E}) \\ \omega_{Y_{0}Y_{1}} = \sigma_{Y_{0}} \otimes \tau_{Y_{1}} }} H_{\alpha}^{\uparrow}(A_{1}A_{2}|E_{2}\tilde{E})_{[\mcM_{2} \circ \mcM_{1}](\omega)} \geq \inf_{\substack{\omega \in \mcD(E_{0}Y_{0}\tilde{E}) \\ \omega_{Y_{0}} = \sigma_{Y_{0}}}} H_{\alpha}^{\uparrow}(A_{1}|E_{1}\tilde{E})_{\mcM_{1}(\omega)} \\ + \inf_{\substack{\omega \in \mcD(E_{1}Y_{1}\tilde{E}) \\ \omega_{Y_{1}} = \tau_{Y_{1}}}} H_{\alpha}^{\uparrow}(A_{2}|E_{2}\tilde{E})_{\mcM_{2}(\omega)}.  \label{eq:meatChain}
\end{multline}
\end{tcolorbox}

\vspace{0.4cm}

\noindent The chain rule \eqref{eq:meatChain} is an instance of the template chain rule \eqref{eq:temp} when the untrusted inputs $R_{i}$ are empty, the sets $\mcS_{1}, \ \mcS_{2}$ and $\mcS_{3}$ are chosen to be the set of input states that satisfy an appropriate marginal constraint, and $\mbH_{\alpha} = \widehat{\mbH}_{\beta} = H_{\alpha}^{\uparrow}$. Furthermore, the purifying registers $\tilde{E}$ are included on all three entropy terms, as opposed to the last term only as encountered in \cref{eq:EAT_chain_1,eq:GEAT_chain_1}. We illustrate this setup in \cref{fig:tensor2} (a).

A special case of \eqref{eq:meatChain} is given in~\cite[Corollary 3.3]{arqand2025} for the pair of channels $\mcM_{1}:Q_{A,0}Q_{B,0}\to A_{1}B_{1}$ and $\mcM_{2}:Q_{A,1}Q_{B,1}\to A_{2}B_{2}$. Specifically, letting $Y_{i} = Q_{A,i}Q_{B,i}$,
\begin{multline}
    \inf_{\substack{\omega \in \mcD(Y_{0}Y_{1}\tilde{E}) \\ \omega_{Q_{A,0}Q_{A,1}} = \sigma_{Q_{A,0}} \otimes \tau_{Q_{A,1}} }} H_{\alpha}^{\uparrow}(A_{1}A_{2}|B_{1}B_{2}\tilde{E})_{[\mcM_{2} \otimes \mcM_{1}](\omega)} = \inf_{\substack{\omega \in \mcD(Y_{0}\tilde{E}) \\ \omega_{Q_{A,0}} = \sigma_{Q_{A,0}}}} H_{\alpha}^{\uparrow}(A_{1}|B_{1}\tilde{E})_{\mcM_{1}(\omega)} \\ + \inf_{\substack{\omega \in \mcD(Y_{1}\tilde{E}) \\ \omega_{Q_{A,1}} = \tau_{Q_{A,1}}}} H_{\alpha}^{\uparrow}(A_{2}|B_{2}\tilde{E})_{\mcM_{2}(\omega)}.  \label{eq:meatChain_v2}
\end{multline}
This can be directly compared to the additivity result of Van Himbeeck and Brown \eqref{eq:tensorChain}: \cref{eq:meatChain_v2} is more general in that the channels $\mcM_{1}$ and $\mcM_{2}$ can differ, as can the marginal states $\sigma$ and $\tau$. Furthermore, the parent chain rule \eqref{eq:meatChain} also allows for sequential channels rather than restricting to a tensor product structure. This can be helpful for certain QKD protocols where some announcements take place before all measurements are completed~\cite{arqand2025,tupkary2026}. However, \cref{eq:meatChain,eq:meatChain_v2} are not strictly more general, since arbitrary linear constraints on the trusted inputs are not permitted, only marginal ones.

The chain rule derived by Fawzi \textit{et al.}~\cite[Corollary 5.1]{Fawzi2026} is a strict generalization of~\cite{VHB25}. 

\vspace{0.4cm}

\begin{tcolorbox}[colback=blue!5!white, colframe=blue!15!white, coltitle=blue!20!black, title=Chain rule for independent CP maps with linear constraints~\cite{Fawzi2026}]
Let $n$ be a positive integer, and for all $i \in \{1,...,n\}$ let $\mcM_{i}:Y_{i-1} \to A_{i}B_{i}$ be a CP map, $\mcN_{i-1} : Y_{i-1} \to Y_{i-1}'$ be a quantum channel and $\sigma_{Y_{i-1}'}^{(i-1)} \in \mcD(Y_{i-1}')$ be a state. For any $\alpha \in (1,\infty)$, 
\begin{equation}
    \inf_{\omega_{Y^{n-1}E} \in \mcD(Y^{n-1}E)} H_{\alpha}^{\uparrow}(A^{n}|B^{n}E)_{\bigotimes_{i=1}^{n}\mcM_{i}(\omega)} = \sum_{i=1}^{n}\inf_{\omega \in \mcS(\mcN_{i-1},\sigma^{(i-1)})} H_{\alpha}^{\uparrow}(A_{i}|B_{i}\tilde{E}_{i})_{\mcM_{i}(\omega)}, \label{eq:schChain}
\end{equation}
where
\begin{equation*}
    \mcS(\mcN_{i-1},\sigma^{(i-1)}) := \big\{ \omega_{Y_{i-1}\tilde{E}_{i}} \in \mcD(Y_{i-1}\tilde{E}_{i}) \ : \ \mcN_{i-1}(\omega_{Y_{i-1}}) = \sigma^{(i-1)}_{Y_{i-1}'} \big\}.
\end{equation*}
\end{tcolorbox}

\vspace{0.4cm}

\noindent This fits out template chain rule under a similar substitution to \eqref{eq:tensorChain}. An illustration can be found in \cref{fig:tensor2} (b). A key advantage of both chain rules discussed in this subsection is that they can provide a security analysis for protocols that vary in time, i.e., when the marginal of Alice's system and the channel changes at each time step. We refer the reader to the discussions in~\cite[Section 5]{arqand2025} and~\cite[Section 5]{Fawzi2026} for details.  

\begin{figure}
\centering
\begin{minipage}{.5\textwidth}
  \centering
  \includegraphics[width=\linewidth]{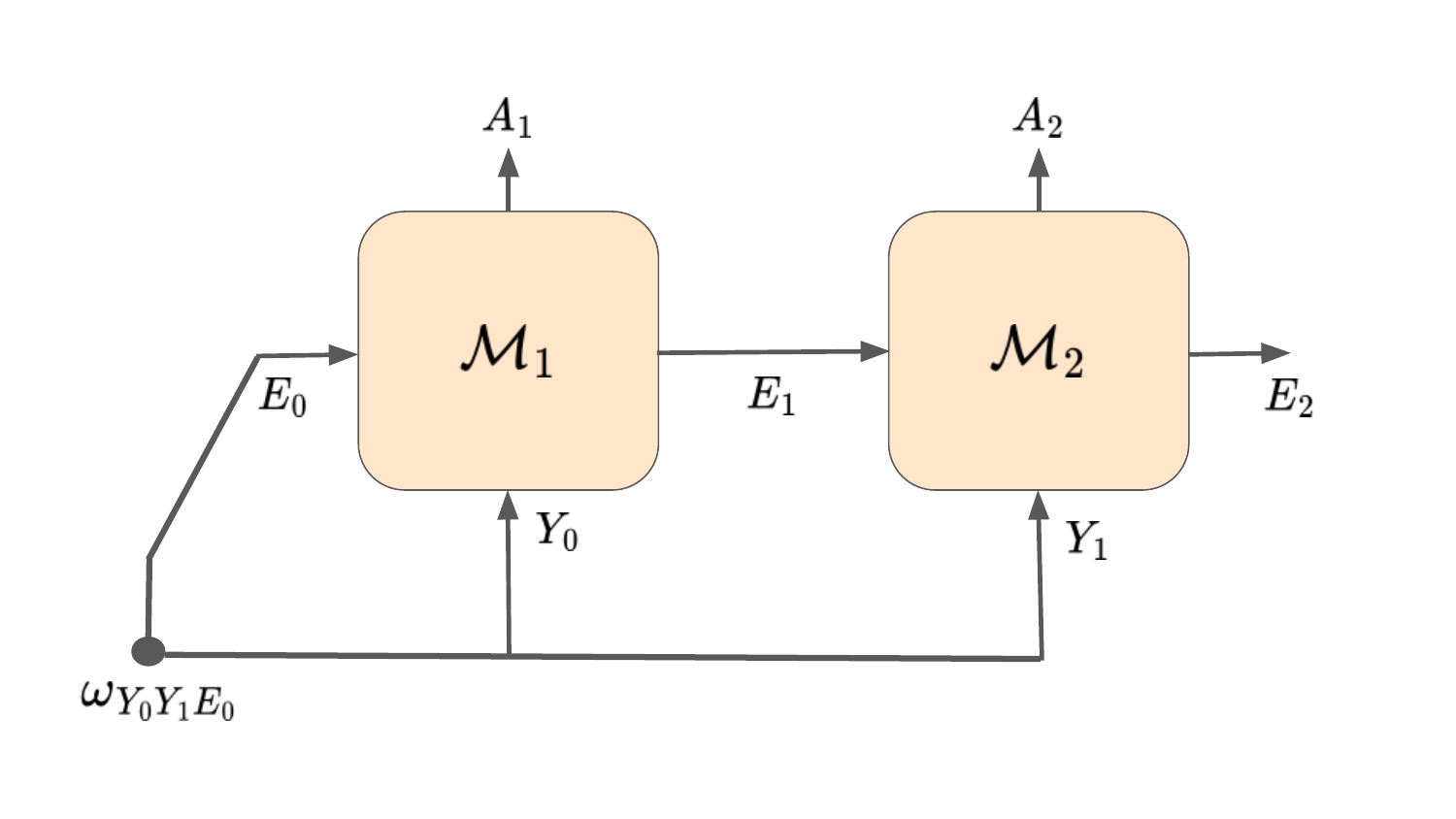}
  \caption*{(a)}
\end{minipage}%
\begin{minipage}{.5\textwidth}
  \centering
  \includegraphics[width=\linewidth]{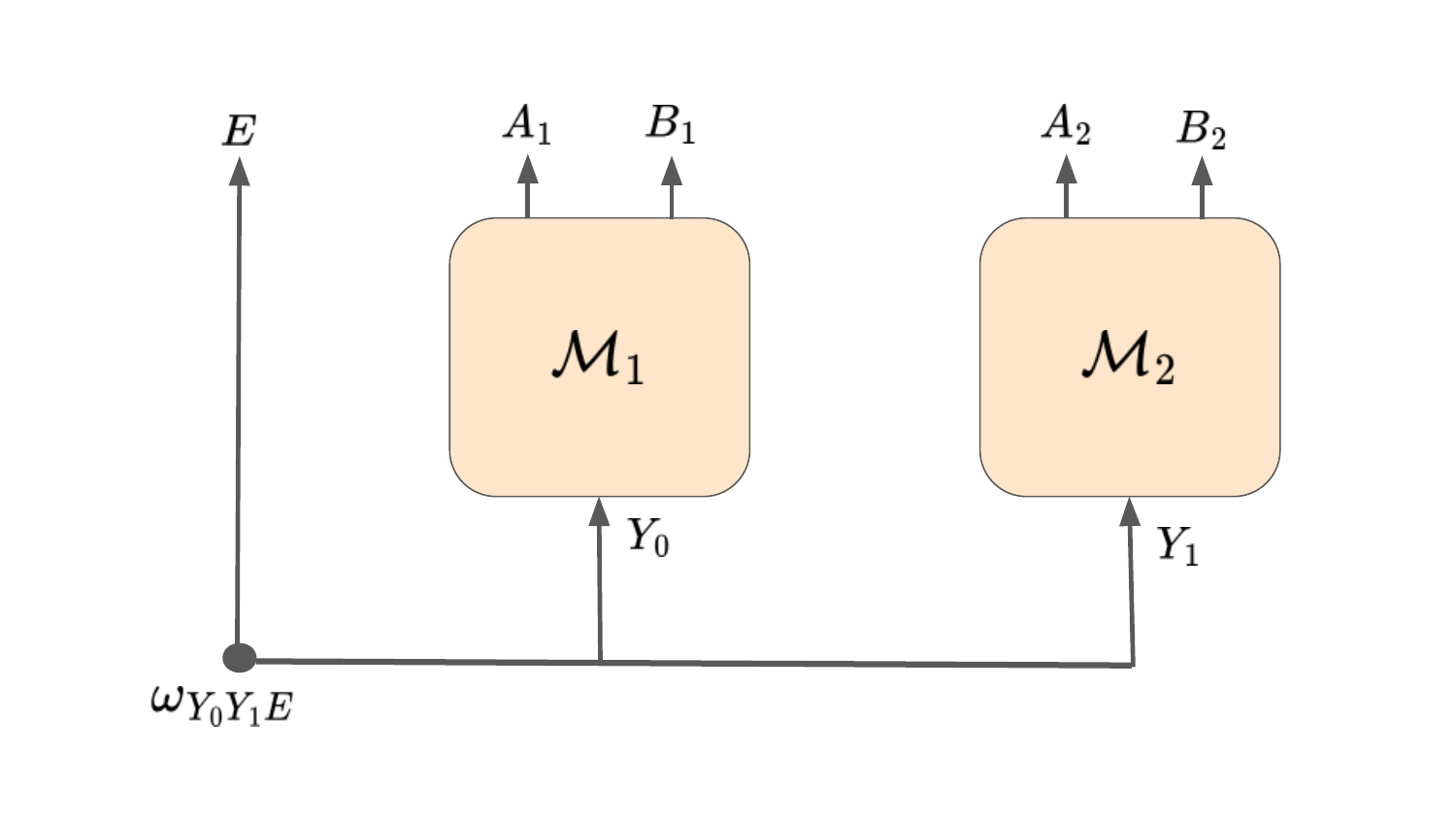}
  \caption*{(b)}
\end{minipage}
\caption{(a) The channel structure considered in the marginal constrained entropy accumulation theorem~\cite{arqand2025}. An initial state $\omega_{Y_{0}Y_{1}E_{0}}$ is sequentially measured by two (potentially different) channels $\mcM_{1}$ and $\mcM_{2}$. Each channel has the ability to update the adversaries side information $E_{i}$, and the initial state can be arbitrarily entangled across all sub-systems subject to a (potentially different) marginal constraint on $Y_{0}$ and $Y_{1}$. (b) A generalization of the channel structure depicted in \cref{fig:tensor1} (a), considered in~\cite{Fawzi2026}. Here, the channels $\mcM_{i}$ are allowed to differ, as are the linear constraints on each $Y_{i}$.}
\label{fig:tensor2}
\end{figure}

\subsection{Further examples}
In this final subsection, we collect additional examples of chain rules that fit (or are related to) our template. A summary of all the chain rules discussed in this survey can be found in \cref{tab:summary}.

There are two more chain rules derived by Fawzi \emph{et al.} that fit our template~\cite[Corollaries 4.2 and 4.8]{Fawzi2026}. A graphical illustration can be found in \cref{fig:Schatten}.  

\vspace{0.4cm}

\begin{tcolorbox}[colback=blue!5!white, colframe=blue!15!white, coltitle=blue!20!black, title=Chain rule without trusted or untrusted inputs~\cite{Fawzi2026}]
Let $\mcM_{1}:E_{0} \to A_{1}E_{1}$ and $\mcM_{2}:E_{1} \to A_{2}E_{2}$ be CP maps, $\omega_{E_{0}} \in \mcD(E_{0})$ be a state and $\alpha \in (1,\infty)$. Then 
\begin{equation}
    H_{\alpha}^{\uparrow}(A_{1}A_{2}|E_{2})_{[\mcM_{2} \circ \mcM_{1}](\omega)} \geq H_{\alpha}^{\uparrow}(A_{1}|E_{1})_{\mcM_{1}(\omega)} + \inf_{\substack{\omega' \in \mcD(E_{1})}} H_{\alpha}^{\uparrow}(A_{2}|E_{2})_{\mcM_{2}(\omega')}.  \label{eq:SchChain2}
\end{equation}
\end{tcolorbox}

\vspace{0.4cm}

\noindent The chain rule \eqref{eq:SchChain2} can be directly compared to the GEAT chain rule (see~\cite[Remark 4.3]{Fawzi2026}) when the untrusted inputs $R_{i}$ are empty and the non-signaling condition holds trivially. The key difference is that \eqref{eq:SchChain2} is in terms of $H_{\alpha}^{\uparrow}$ only, i.e., $\mbH_{\alpha} = \widehat{\mbH}_{\beta} = H_{\alpha}^{\uparrow}$. Note the similarities between this channel structure in \cref{fig:Schatten} (a) and that of the GEAT when applied to prepare and measure QKD in \cref{fig:GEAT} (b). By using \eqref{eq:SchChain2} instead of the GEAT, one could perform a similar analysis, except without incurring a loss in the parameter $\alpha$ or needing a purifying register in the entropy term $\widehat{\mbH}_{\beta}$. This would still be subject to the same repetition rate constraints, however. 

\vspace{0.4cm}

\begin{tcolorbox}[colback=blue!5!white, colframe=blue!15!white, coltitle=blue!20!black, title=Chain rule with untrusted inputs and product CP maps~\cite{Fawzi2026}]
Let $\mcM_{1}:E_{0}R_{0} \to A_{1}E_{1}$ and $\mcM_{2}:E_{1}R_{1} \to A_{2}E_{2}$ be CP maps, $\omega_{R_{0}R_{1}E_{0}} \in \mcD(R_{0}R_{1}E_{0})$ be a state and $\alpha \in (1,\infty)$. Then provided the channel $\mcM_{2}$ admits a tensor product structure $\mcM_{2} = \mcM_{2}' \otimes \mcE_{2}$, where $\mcM_{2}':R_{1} \to A_{2}$ and $\mcE_{2}:E_{1} \to E_{2}$,
\begin{equation}
    H_{\alpha}^{\uparrow}(A_{1}A_{2}|E_{2})_{[\mcM_{2} \circ \mcM_{1}](\omega)} \geq H_{\alpha}^{\uparrow}(A_{1}|E_{1})_{\mcM_{1}(\omega)} + \inf_{\substack{\omega' \in \mcD(R_{1}E_{1}\tilde{E})}} H_{\alpha}^{\uparrow}(A_{2}|E_{2}\tilde{E})_{\mcM_{2}(\omega')},  \label{eq:SchChain3}
\end{equation}
where $\tilde{E}$ has the same dimension as $R_{1}E_{1}$.
\end{tcolorbox}

\vspace{0.4cm}

\noindent This can also be compared to the GEAT chain rule when the untrusted inputs $R_{i}$ are not updated and the channels admit a tensor product structure across $R_{i}$ and $E_{i}$ (see~\cite[Remark 4.9]{Fawzi2026}). This structure guarantees that there cannot be any signaling between $R_{0}$ and $E_{1}$, or between $R_{1}$ and $E_{2}$. As with the previous example, the benefit is that $H_{\alpha}^{\uparrow}$ appears for all terms.

\begin{figure}
\centering
\begin{minipage}{.5\textwidth}
  \centering
  \includegraphics[width=\linewidth]{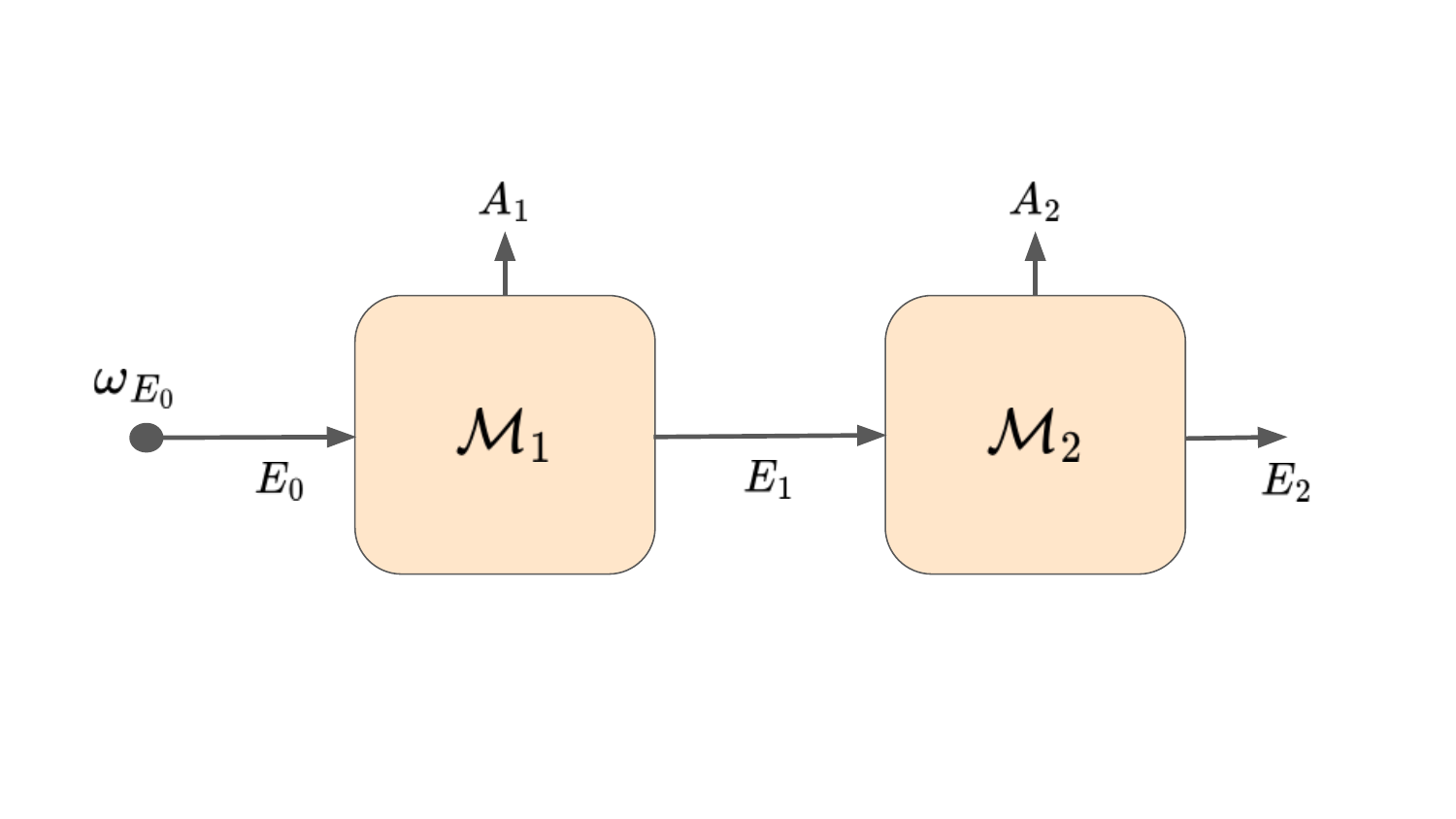}
  \caption*{(a)}
\end{minipage}%
\begin{minipage}{.5\textwidth}
  \centering
  \includegraphics[width=\linewidth]{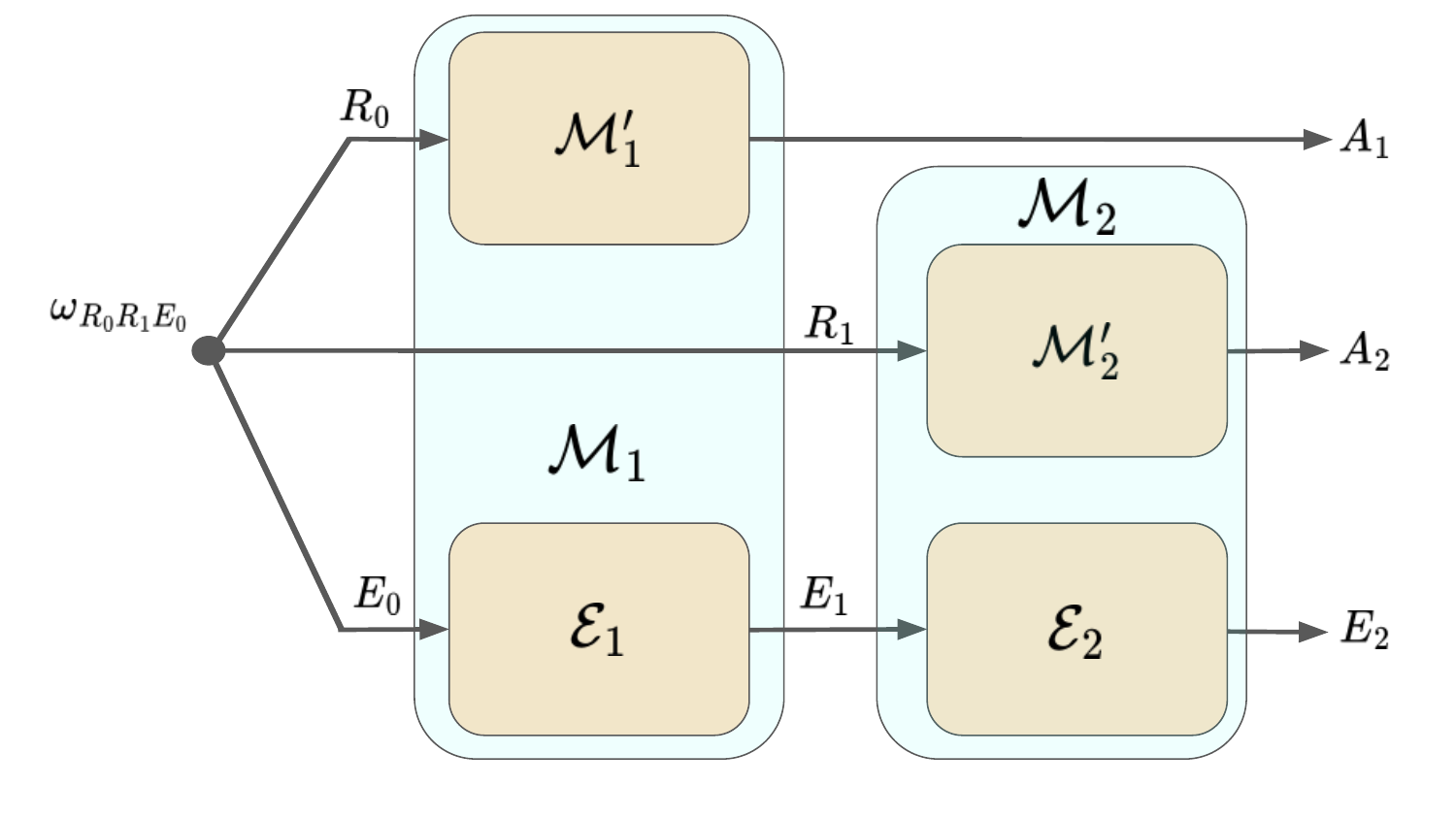}
  \caption*{(b)}
\end{minipage}
\caption{(a) The non-signaling channel structure considered in~\cite[Corollary 4.2]{Fawzi2026}. This can be viewed as a special case of the non-signaling channels used in the generalized entropy accumulation theorem, where the memory registers $R_{i}$ are empty. (b) Another non-signaling channel structure considered in~\cite[Corollary 4.8]{Fawzi2026} where both channels are independent.}
\label{fig:Schatten}
\end{figure}

An extension of the chain rule for non-signaling channels, used to prove the GEAT, was proven by Arqand, Metger and Tan~\cite[Theorem 5.1.]{arqand2024_mutual}. The same setting as that described in \cref{sec:NS} is considered, where the trusted inputs $Y_{i}$ are empty and the untrusted inputs $R_{i}$ can be updated each round, as can the side information $E_{i}$. The condition that $\mcM_{i}$ is non-signaling between $R_{i-1}$ and $E_{i}$ is then relaxed. Specifically, a leakage channel $\mcL_{i}$ between $R_{i-1}$ and $E_{i}$ is permitted, and a penalty is subtracted that depends on the maximum correlation permitted by $\mcL_{i}$ over all possible input states.   

\vspace{0.4cm}

\begin{tcolorbox}[colback=blue!5!white, colframe=blue!15!white, coltitle=blue!20!black, title=Chain rule for non-signaling channels with leakage~\cite{arqand2024_mutual},breakable]
Let $\mcL_{1}:R_{0} \to R_{0}'L_{0}$, $\mcM_{1}':R_{0}'L_{0}E_{0} \to A_{1}R_{1}E_{1}$, $\mcL_{2}:R_{1} \to R_{1}'L_{1}$ and $\mcM_{2}':R_{1}'L_{1}E_{1} \to A_{2}R_{2}E_{2}$ be quantum channels, and define $\mcM_{1} = \mcM_{1}' \circ \mcL_{1}$ and $\mcM_{2} = \mcM_{2}' \circ \mcL_{2}$. Let $\omega_{R_{0}E_{0}} \in \mcD(R_{0}E_{0})$ be a state and $\alpha \in (1,2)$. Then provided the channel $\mcM_{2}$ is non-signaling between $R_{1}'$ and $E_{2}$, 
\begin{multline}
    H_{\alpha}^{\uparrow}(A_{1}A_{2}|E_{2})_{[\mcM_{2} \circ \mcM_{1}](\omega)} \geq H_{\alpha}^{\uparrow}(A_{1}|E_{1})_{\mcM_{1}(\omega)}  + \inf_{\substack{\omega' \in \mcD(R_{1}'E_{1}\tilde{E})}} H_{\frac{1}{2-\alpha}}^{\downarrow}(A_{2}|E_{2}\tilde{E})_{\mcM_{2}(\omega')} \\ - \sup_{\omega'' \in \mcD(R_{1}\tilde{R})} I_{\alpha}^{\downarrow}(\tilde{R};L_{1})_{\mcL_{2}(\omega'')},  \label{eq:leakyGEAT}
\end{multline}
where $\tilde{E}$ has the same dimension as $R_{1}'L_{1}E_{1}$, $\tilde{R}$ has the same dimension as $R_{1}$ and for a bipartite state $\rho_{AB} \in \mcD(AB)$, $I_{\alpha}^{\downarrow}(A;B)_{\rho} = \inf_{\sigma \in \mcD_{\leq}(B)}D_{\alpha}(\rho_{AB}\|\rho_{A} \otimes \sigma_{B})$ is the R\'enyi mutual information, where $\rho_{A} = \tr_{B}[\rho_{AB}]$.
\end{tcolorbox}

\vspace{0.4cm}

\noindent An illustration can be found in \cref{fig:leak}. The first two terms on the right hand side of \eqref{eq:leakyGEAT} directly correspond to the original GEAT chain rule \eqref{eq:GEAT_chain_1} when $\mbH_{\alpha} = H_{\alpha}^{\uparrow}$, since than channel $\mcM_{2}'$ is non-signaling between $R_{1}'$ and $E_{2}$. However, the global channel $\mcM_{2} = \mcM_{2}' \circ \mcL_{2}$ can pass signals from $R_{1}$ to $E_{2}$, since some information stored in $R_{1}$ could be passed to $L_{1}$ via $\mcL_{1}$, and subsequently leaked to $E_{2}$ via $\mcM_{2}'$ (recall $\mcM_{2}'$ is allowed to signal from $L_{1}$ to $E_{2}$). The extent to which this can take place is bounded by the mutual information term subtracted at the end of \eqref{eq:leakyGEAT}. It captures the maximum amount of correlation that can be shared between the leakage register $L_{1}$ and a register $\tilde{R}$ that purifies the input $R_{1}$. If minimal leakage takes place, i.e., $L_{1}$ is independent of $R$, $L_{1}$ must also be independent of $\tilde{R}$ and the penalty is zero. On the other hand, in the event of maximal leakage, i.e., $\mcL_{1}$ outputs the contents of $R_{1}$ in $L_{1}$, a highly entangled input state between $R_{1}$ and $\tilde{R}$ will result in a large penalty.   

\begin{figure}[h]
\includegraphics[width=10cm]{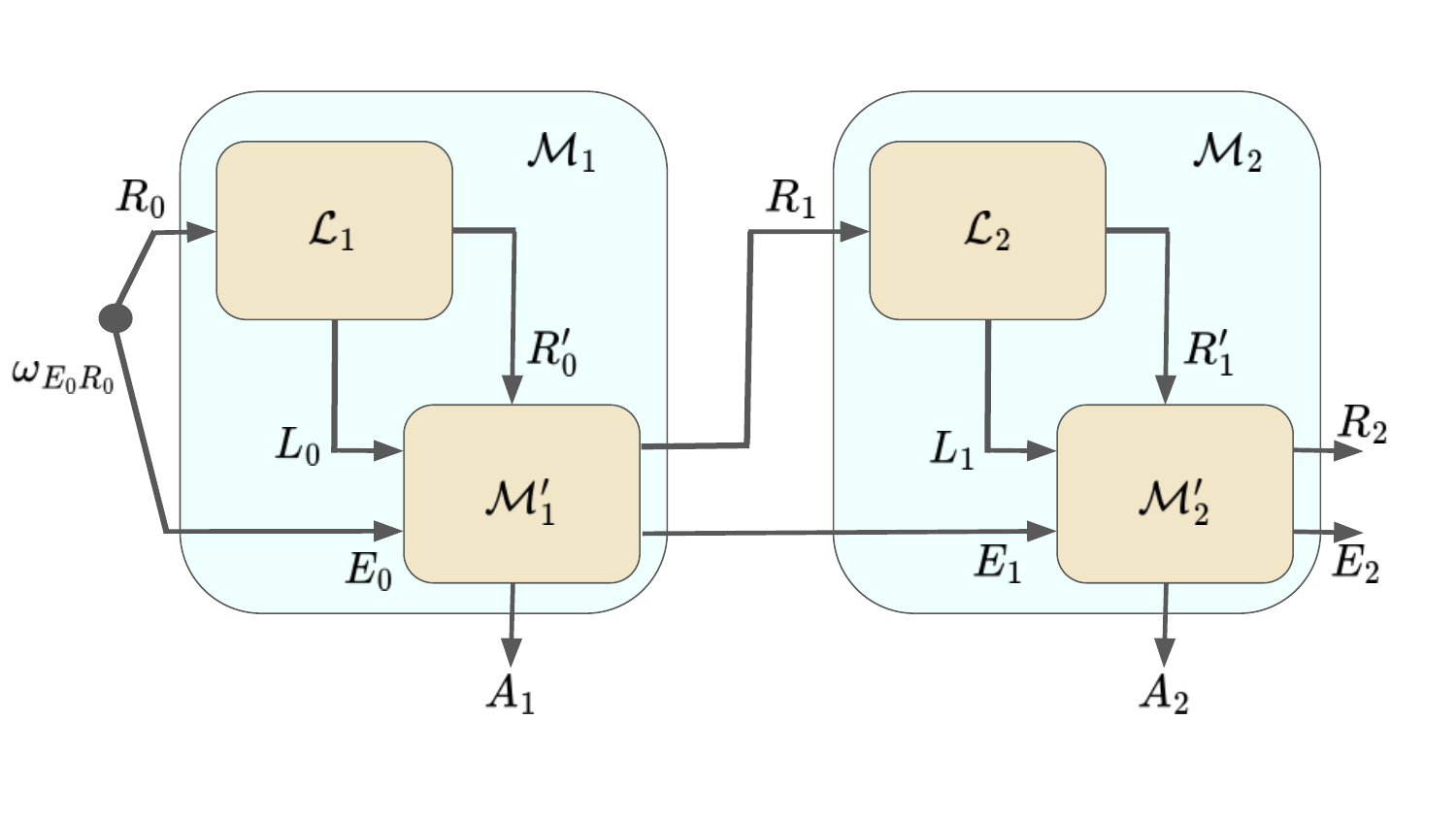}
\centering
\caption{Relaxation of the non-signaling channel structure considered in~\cite{arqand2024_mutual}. The channels $\mcM_{i}$ are non-signaling from $R_{i-1}'$ to $E_{i}$. However, depending on the leakage channel $\mcL_{i}$, $R_{i-1}$ might signal to $E_{i}$ via the leakage register $L_{i-1}$. If the $L_{i}$ systems are always empty, we recover the non-signaling channels considered in~\cite{MetgerGEAT}. If $L_{i}$ contains a copy of $R_{i}$, we have perfect signaling between the devices and the adversary, and no entropy can be accumulated.} 
\label{fig:leak}
\end{figure}

\begin{table}
  \centering
  \begin{tblr}{
      colspec={l|l|l|l|l|l},
      row{1}={font=},
      column{1}={font=\itshape},
      row{even}={bg=blue!10},
    }
      & Channel constraints  & Input constraints   & $\widehat{\mbH}_{\beta}$ & \makecell{Apply to \\ DD?} & \makecell{Apply to \\ DI?} \\
    \toprule
    \makecell{\cite{DFR}} & Markov, CPTP & None  &$H_{\alpha}^{\downarrow}$ & \cmark$^{\text{(b)}}$ & \cmark \\ 
    \cite{MetgerGEAT} & \makecell{Non-signaling, \\ CPTP} & None & \makecell{$H_{\frac{1}{2-\alpha}}^{\downarrow}$, \\ $\alpha \in (1,2)$} & \cmark$^{\text{(c)}}$ & \cmark \\ 
    \makecell{\cite{VHB25}} & \makecell{Independent \\ and identical, CP} & \makecell{Memoryless, \\ linear constraints} & $H_{\alpha}^{\uparrow}$ & \cmark & \xmark \\ 
    \makecell{\cite{arqand2025}} & CPTP & \makecell{Memoryless, \\ marginal constraints} & $H_{\alpha}^{\uparrow}$ & \cmark & \xmark\\ 
    \cite{Fawzi2026} & Independent, CP & \makecell{Memoryless, \\ linear constraints} & $H_{\alpha}^{\uparrow}$ & \cmark & \xmark \\ 
    \cite{Fawzi2026} & CP & Memoryless & $H_{\alpha}^{\uparrow}$ & \cmark$^{\text{(c)}}$ & \xmark \\
    \cite{Fawzi2026} & \makecell{Non-signaling, CP} & Memoryless & $H_{\alpha}^{\uparrow}$ &\cmark$^{\text{(c)}}$ & \xmark \\ 
    \makecell{\cite{arqand2024_mutual}} & \makecell{Non-signaling \\ with leakage, CPTP} & None & \makecell{$H_{\frac{1}{2-\alpha}}^{\downarrow}$, \\ $\alpha \in (1,2)$} & \cmark$^{\text{(b)}}$ &\cmark \\
    This work & \makecell{Independent and \\ classical new \\ side information, \\ CPTP} & None & $H_{\alpha}^{\uparrow \downarrow}$ $^{\text{(a)}}$ &  \cmark$^{\text{(b)}}$ & \cmark \\
    \bottomrule
  \end{tblr}
  \caption{A summary of the different chain rules captured by our framework. Each row corresponds to a chain rule for the optimized sandwiched R\'enyi entropy. That is, for a state $\rho_{A_{1}A_{2}E}$ generated by the sequential application of two channels that generate raw data $A_{1}$ and $A_{2}$, respectively, we consider a lower bound on $H_{\alpha}^{\uparrow}(A_{1}A_{2}|E)_{\rho}$ by the sum of a term similar to $H_{\alpha}^{\uparrow}(A_{1}|E)_{\rho}$, and a term similar to $\widehat{\mbH}_{\beta}(A_{2}|E)_{\rho}$, for $1 < \alpha \leq \beta$. ``Channel constraints'' refer to the assumptions made on the channels. ``Input constraints'' refer to constraints on the input states. $\widehat{\mbH}_{\beta}$ denotes the R\'enyi entropy of the second term in the summation, along with the R\'enyi parameter and any necessary restrictions. ``Apply to DD'' and ``Apply to DI'' indicate the applicability of the chain rule to device-dependent and device-independent protocol, respectively. (a) $H_{\alpha}^{\uparrow \downarrow}$ refers to the partially optimized R\'enyi entropy defined in \cref{def:partial_ent}. (b) To apply to prepare and measure QKD, extra ``tomography rounds'' are required~\cite{Bauml2024} (see \cref{sec:marg1}). (c) To apply to prepare and measure QKD, a constraint on the protocol's repetition rate is required~\cite{Metger2023} (see \cref{sec:NS}).}
  \label{tab:summary}
\end{table}

\subsection{Unstructured processes}

In this survey, we have restricted our focus to chain rules that apply to structured processes, i.e., ones that hold for states generated by a sequence of channels obeying certain properties. Such cases arise naturally in many applications, particularly in cryptography. There are however examples of chain rules for conditional R\'enyi entropies that hold for unstructured processes. While beyond the scope of this manuscript, we briefly mention some examples here and refer the reader to the references for further information. 

A chain rule for $H_{\alpha}^{\uparrow}$ that holds for all tripartite states $\rho_{ABC}$ was given by Dupuis~\cite[Theorem 1]{Dupuis_2015},
\begin{equation}
    H_{\alpha}^{\uparrow}(AB|C) \geq H_{\beta}^{\uparrow}(A|BC)_{\rho} + H_{\gamma}^{\uparrow}(B|C)_{\rho} \label{eq:unstr_chain}
\end{equation}
for an appropriately chosen $(\alpha,\beta,\gamma)$. A similar inequality holds in the opposite direction for a complementary regime of R\'enyi parameters. This can be viewed as a R\'enyi entropy analogue of the chain rule for the von Neumann entropy $H(AB|C)_{\rho} = H(A|BC)_{\rho} + H(B|C)_{\rho}$. Note that the change of R\'enyi parameters results in a weaker chain rule than the examples considered in previous subsections. In return however, this chain rule applies to all quantum states rather than only states arising from the application of, e.g., non-signaling channels. Chain rules for the smooth min and max-entropy that hold for all states and take a similar form to \eqref{eq:unstr_chain} were provided earlier in~\cite{Vitanov_2013}. Notably, instead of a loss in R\'enyi parameter, a loss in the smoothing parameter is incurred after each application of the rule.  

The un-optimized conditional min entropy $H_{\infty}^{\downarrow}$\footnote{We have frequently refered to the min-entropy, given by $H_{\text{min}} = \lim_{\alpha \to \infty} H_{\alpha}^{\uparrow}$. Its smoothed version $H_{\text{min}}^{\epsilon}$ is obtained by maximizing over states within an $\epsilon$-ball, and is an operationally relevant quantity that arises in, e.g., randomness extraction~\cite{TSSR}. The un-optimized version is given by $H_{\infty}^{\downarrow} = \lim_{\alpha \to \infty} H_{\alpha}^{\downarrow}$, and it is known that the smoothed version $H_{\infty}^{\downarrow,\epsilon}$ is equal to the smooth min-entropy up to a constant~\cite[Lemma 20]{TSSR}.} is known to satisfy a chain rule of the form \eqref{eq:unstr_chain}~\cite[Proposition 5.12]{Tomamichel_2016}, and a smoothed version of this was more recently proven by Marwah and Dupuis~\cite{marwah2024a,marwah2024b}, referred to as a universal chain rule. In the same work, the authors further prove an unstructured version of the entropy accumulation theorem. Broadly, this is the statement that the smooth min-entropy can be lower bounded by roughly $n$ times the conditional von Neumann entropy when the state can be approximated, in a certain sense, by the output of a sequence of quantum channels fulfilling the Markov conditions of the original EAT (cf. \cref{sec:markov}). 

\section{Counterexample and a modified chain rule} \label{sec:new_chain}
In this section, we explore potential modifications of the chain rule used to prove the EAT that is particularly relevant to the security of device-independent protocols (see \cref{fig:EAT} for an illustration). As discussed in \cref{sec:markov}, a key step in the EAT is reducing the $n$-round optimized R\'enyi entropy $H_{\alpha}^{\uparrow}(A^{n}|B^{n}E_{0})_{\rho}$ to a sum of entropy contributions from each round, roughly of the form $H_{\alpha}^{\downarrow}(A_{i}|B_{i}\tilde{E})_{\mcM_{i}(\omega)}$. Note that the initial optimized entropy $H_{\alpha}^{\uparrow}$ has been replaced by a sum of non-optimized entropies $H_{\alpha}^{\downarrow}$, contrasting the $n$-round entropy of an i.i.d.~process, equal to a sum of terms $H_{\alpha}^{\uparrow}(A_{i}|B_{i}\tilde{E})_{\mcM_{i}(\omega)}$. This difference does not impact the asymptotic rate $n \to \infty$ since both $H_{\alpha}^{\uparrow}$ and $H_{\alpha}^{\downarrow}$ converge to the von Neumann entropy as $\alpha \to 1$. However, it does suggest room for improvement at a finite $n$. This is well motivated by the device-dependent context where such improvements have been made. As discussed in \cref{sec:marg1}, the lower bound on the accumulated entropy exactly matches the upper bound given by an i.i.d.~process. By strengthening the EAT chain rule \eqref{eq:EAT_chain_1}, it may be possible to recover such additivity results in the device-independent setting.  

Towards this end, we are interested in improving the chain rule~\cite{DFR},
\begin{equation}
    H_{\alpha}^{\uparrow}(A_{1}A_{2}|B_{1}B_{2}E_{0})_{[\mcM_{2} \circ \mcM_{1}] (\omega)} \geq H_{\alpha}^{\uparrow}(A_{1}|B_{1}E_{0})_{\mcM_{1} (\omega)} + \inf_{\omega' \in \mcD(R_{1}\tilde{E})} H_{\alpha}^{\downarrow}(A_{2}|B_{2}\tilde{E})_{\mcM_{2} (\omega')}, \label{eq:EAT_chain_2}
\end{equation}
where $\omega_{R_{0}E_{0}}$ is an initial state and $\mcM_{i}:R_{i-1} \to A_{i}B_{i}R_{i}$ are quantum channels that satisfy the Markov condition $A_{1} \leftrightarrow E_{0}B_{1} \leftrightarrow B_{2}$. Restricting our focus to DI applications, $A_{i}$ stores the raw key generated at round $i$, $E_{0}$ is the initial system held by the adversary and $B_{i}$ stores the inputs of each device at round $i$ that are assumed to be chosen independently and randomly each round. As a consequence, for any input state $\omega_{R_{1}A_{1}B_{1}E_{0}}$ to the second channel, the output must satisfy 
\begin{equation}
    [\tr_{A_{2}R_{2}} \circ \mcM_{2}](\omega_{R_{1}A_{1}B_{1}E_{0}}) = \rho_{A_{1}B_{1}E_{0}} \otimes \rho_{B_{2}}. \label{eq:DIch}
\end{equation}
This implies that $A_{1} \leftrightarrow E_{0}B_{1} \leftrightarrow B_{2}$ forms a Markov chain. Thus, for DI applications, it will be sufficient to restrict our attention to channels $\mcM_{2}$ that satisfy \eqref{eq:DIch} rather than the more general Markov condition. We refer to such channels as DI channels.

\subsection{Classical counterexample} \label{sec:counter}
We first ask if the following holds when $\mcM_{2}$ satisfies \eqref{eq:DIch}, each $A_{i}$ and $B_{i}$ are classical and $\alpha \in (1,\infty)$:  
\begin{equation}
    H_{\alpha}^{\uparrow}(A_{1}A_{2}|E_{0}B_{1}B_{2})_{[\mcM_{2} \circ \mcM_{1}] (\omega)} \stackrel{\text{?}}{\geq} H_{\alpha}^{\uparrow}(A_{1}|E_{0}B_{1})_{\mcM_{1} (\omega)} + \inf_{\omega' \in \mcD(R_{1}\tilde{E})} H_{\alpha}^{\uparrow}(A_{2}|B_{2}\tilde{E})_{\mcM_{2} (\omega')}. \label{eq:EAT_wrong}
\end{equation}
Note the key difference between \eqref{eq:EAT_wrong} and the existing chain rule \eqref{eq:EAT_chain_2} is that every term in \eqref{eq:EAT_wrong} is an optimized R\'enyi entropy $H_{\alpha}^{\uparrow}$. Unfortunately, we provide a counterexample to \eqref{eq:EAT_wrong} where the state and channels that violate the inequality are fully classical.

Consider the case where all variables are bits. We take the channel $\mcM_{1}$ to prepare a fixed state on $A_{1}B_{1}$ (ignoring the input state, meaning we can omit $E_{0}$ throughout this example) described by the probability distribution $p_{A_{1}B_{1}}(a_{1},b_{1}) = p_{B_{1}}(b_{1}) \, p_{A_{1}|B_{1}}(a_{1}|b_{1})$ where
\begin{equation*}
    p_{B_{1}}(b_{1}) = \frac{1}{2}, \ \forall b_{1} \in \{0,1\} \ \text{and} \ 
        p_{A_{1}|B_{1}}(a_{1}|b_{1}) = \begin{cases}
            \frac{3}{4}\delta_{a_{1},0} + \frac{1}{4} \delta_{a_{1},1} \ \text{if} \ b_{1}=0,\\
            \delta_{a_{1},1} \ \text{if} \ b_{1}=1.
        \end{cases}
\end{equation*}
The value of $R_{1}$ is then chosen to be independent of $B_{1}$, and set to equal $A_{1}$, i.e.,
\begin{equation}
    p_{R_{1}|A_{1}B_{1}}(r_{1}|a_{1},b_{1}) = \delta_{r_{1},a_{1}}. \label{eq:p_r1} 
\end{equation}
For $\mcM_{2}$, we consider channels that generate $B_2$ ignoring the input state and then generate $A_2$ depending on $R_1B_2$. Then given any input distribution $p_{R_{1}R'}$ for any classical system $R'$, the output distribution after applying the map $\mcM_{2}:R_{1}\to A_{2}B_{2}$ is of the form
\begin{equation}
    p_{A_{2}B_{2}R'}(a_{2},b_{2},r') = \sum_{r_{1}} p_{R_{1}R'}(r_{1},r')\,p_{B_{2}}(b_{2}) \, p_{A_{2}|R_{1}B_{2}}(a_{2}|r_{1},b_{2}). \label{eq:p_r2}
\end{equation}
Note that by summing over $a_{2}$, we obtain $p_{B_{2}R'}(b_{2},r') = p_{B_{2}}(b_{2}) \, p_{R'}(r')$, i.e., any channel $\mcM_{2}$ of this form satisfies the condition \eqref{eq:DIch}. We then consider the following instance of this channel:
\begin{equation}
    p_{B_{2}}(b_{2}) = \frac{1}{2} \ \forall b_{2} \in \{0,1\}, \ \text{and} \ p_{A_{2}|R_{1}B_{2}}(a_{2}|r_{1},b_{2}) = \begin{cases}
            \frac{1}{2} \ \text{if} \ r_{1} \oplus b_{2} = 0\\
            \delta_{a_{2},0} \ \text{if} \ r_{1} \oplus b_{2} = 1.
        \end{cases} \label{eq:classical_ch1}
\end{equation}
By combining \cref{eq:p_r1,eq:p_r2}, the distribution of all variables $A_{1}B_{1}A_{2}B_{2}$ after applying $\mcM_{2} \circ \mcM_{1}$ to any input state $p_{R_0}$ is given by
\begin{equation*}
\begin{aligned}
    p_{R_{0}}(r_{0}) &\xrightarrow{\mcM_{1}} p_{R_{0}}(r_{0}) \, p_{B_1}(b_1)\, p_{A_1|B_1}(a_1|b_1) \, \delta_{r_1,a_1} \\
    &\xrightarrow{\tr_{R_{0}}\circ \mcM_{2}}  \sum_{r_1} p_{B_1}(b_1)\, p_{A_1|B_1}(a_1|b_1) \, \delta_{r_1,a_1} \,p_{B_{2}}(b_{2}) \, p_{A_{2}|R_{1}B_{2}}(a_{2}|r_{1},b_{2})\\
    &=  p_{B_{1}}(b_{1}) \,  p_{A_{1}|B_{1}}(a_{1}|b_{1}) \, p_{B_{2}}(b_{2}) \, p_{A_{2}|R_{1}B_{2}}(a_{2}|a_{1},b_{2}),
\end{aligned}
\end{equation*}
where in the second line we applied \cref{eq:p_r2} with $R' = A_1B_1$.

Using the expression \eqref{eq:classical_up} for the classical optimized R\'enyi entropy, we can compute the following quantities for $\alpha = 1.5$ (rounding to 5 d.p.) and any input state $p_{R_{0}}$:
\begin{equation*}
\begin{aligned}
        H_{\alpha}^{\uparrow}(A_{1}A_{2}|B_{1}B_{2})_{[\mcM_{2}\circ \mcM_{1}](p)} &= \frac{\alpha}{1-\alpha} \log \Bigg( \frac{1}{4}\sum_{b_{1},b_{2}} \Bigg( \sum_{a_{1},a_{2}} p_{A_{1}|B_{1}}(a_{1}|b_{1})^{\alpha} p_{A_{2}|R_{1}B_{2}}(a_{2}|a_{1},b_{2})^{\alpha}\Bigg)^{\frac{1}{\alpha}} \Bigg) \approx 0.82057,\\
        H_{\alpha}^{\uparrow}(A_{1}|B_{1})_{\mcM_{1}(p)} &= \frac{\alpha}{1-\alpha} \log \Bigg(\frac{1}{2} \Big[ \Big(\frac{3}{4}\Big)^{\alpha} + \Big(\frac{1}{4}\Big)^{\alpha} \Big]^{\frac{1}{\alpha}} + \frac{1}{2}  \Bigg)\approx 0.35295.
\end{aligned}
\end{equation*}
In Appendix \ref{app:counter} we show
\begin{equation}
    \inf_{\omega' \in \mcD(R_{1}\tilde{E})} H_{\alpha}^{\uparrow}(A_{2}|B_{2}\tilde{E})_{\mcM_{2}(\omega')} = \frac{\alpha}{1-\alpha}\log\Bigg(\frac{1}{2}\Bigg( \Big(\frac{1}{2}\Big)^{\alpha} + \Big(\frac{1}{2}\Big)^{\alpha} \Bigg)^{\frac{1}{\alpha}} + \frac{1}{2}\Bigg) \approx 0.47118, \label{eq:app_claim}
\end{equation}
where the infimum is taken over all cq-states $\omega_{R_{1}\tilde{E}}'$. For this example, the right hand side of \eqref{eq:EAT_wrong} is equal to 0.82413 which exceeds the left hand side value of 0.82057, i.e., the hypothesized chain rule is violated,
\begin{equation*}
    H_{\alpha}^{\uparrow}(A_{1}A_{2}|B_{1}B_{2})_{p} < H_{\alpha}^{\uparrow}(A_{1}|B_{1})_{\mcM_{1}(p)} +  \inf_{\omega' \in \mcD(R_{1}\tilde{E})} H_{\alpha}^{\uparrow}(A_{2}|B_{2}\tilde{E})_{\mcM_{2}(\omega')}. 
\end{equation*}
Furthermore, note that the channel $\mcM_{1}$ in this example is entirely independent of its input. We therefore see that
\begin{equation*}
    H_{\alpha}^{\uparrow}(A_{1}|B_{1})_{\mcM_{1} (\omega)} = \inf_{\omega'\in \mcD(R_{0}\tilde{E})} H_{\alpha}^{\uparrow}(A_{1}|B_{1}\tilde{E})_{\mcM_{1}(\omega')},
\end{equation*}
where $\tilde{E}$ is any purifying register. As a result, this example also rules out a variant of the chain rule \eqref{eq:EAT_wrong} where both terms on the right hand side are minimized over input states, and therefore the case where all three terms are minimized over input states cannot hold\footnote{This follows from the fact that $H_{\alpha}^{\uparrow}(A_{1}A_{2}|E_{0}B_{1}B_{2})_{[\mcM_{2} \circ \mcM_{1}] (\omega)} \geq \inf_{\omega} H_{\alpha}^{\uparrow}(A_{1}A_{2}|E_{0}B_{1}B_{2})_{[\mcM_{2} \circ \mcM_{1}] (\omega)}$, hence a chain rule in which all three terms are minimized would imply \eqref{eq:EAT_wrong} and be a contradiction.}. Furthermore, for this example we numerically observe that the chain rule of Dupuis, Fawzi and Renner~\cite{DFR} in \cref{eq:EAT_chain_1} is saturated, i.e.,
\begin{equation*}
    H_{\alpha}^{\downarrow}(A_{1}A_{2}|B_{1}B_{2})_{p} = H_{\alpha}^{\downarrow}(A_{1}|B_{1})_{\mcM_{1}(p)} +  \inf_{\omega' \in \mcD(R_{1}\tilde{E})} H_{\alpha}^{\downarrow}(A_{2}|B_{2}\tilde{E})_{\mcM_{2}(\omega')}.
\end{equation*}

A modified chain rule for non-signaling channels was considered in~\cite[Appendix A]{arqand2025}, where a counterexample was also shown. The key difference to the chain rule we consider is that~\cite{arqand2025} imposes marginal constraints on each entropy term (as described in \cref{sec:marg2}), something we do not enforce. Nonetheless, both examples highlight the difference between chain rules that allow for a memory register, as necessary for device-independent protocols, and chain rules applicable to memoryless settings such as device-dependent protocols.     

\subsection{Tightened chain rule} \label{sec:new_chain_cor}
Under the assumption that the side information is generated independently each round, the preceding counterexample shows that a quantum process described by \cref{fig:EAT} does not admit a chain rule for the optimized R\'enyi entropy $H_{\alpha}^{\uparrow}$ only. A natural follow up question is if \eqref{eq:EAT_chain_2} can be improved at all. We now show that this is the case. To do so, we introduce a new notion of conditional R\'enyi entropy for cq-states when there are two conditioning registers, one of which is classical. 

\begin{definition}[Partially optimized conditional R\'enyi entropy for cq-states]
    Let $\rho_{ABC} = \sum_{b} p_{B}(b) \, \ketbra{b}{b}_{B} \otimes \rho_{AC}^{|b} \in \mcD(ABC)$ be a cq-state classical on $B$, and let $\alpha \in (1,\infty)$. The partially optimized conditional R\'enyi entropy of order $\alpha$ between $A$ and $BC$ is defined as
    \begin{equation}
        H_{\alpha}(A|B^{\uparrow}C^{\downarrow})_{\rho} := \frac{\alpha}{1-\alpha} \log\Bigg( \sum_{b} p_{B}(b) 2^{\frac{1-\alpha}{\alpha}H_{\alpha}^{\downarrow}(A|C)_{\rho^{|b}}}\Bigg). \label{eq:partial_ent}
    \end{equation}
    \label{def:partial_ent}
\end{definition}
We call $H_{\alpha}(A|B^{\uparrow}C^{\downarrow})_{\rho}$ ``partially optimized'' in contrast to the ``fully'' optimized entropy $H_{\alpha}^{\uparrow}(A|BC)_{\rho}$. To see the difference, when conditioned on a classical register the optimized R\'enyi entropy takes the form (see~\cite[Proposition 5.4]{Tomamichel_2016}) 
\begin{equation*}
        H_{\alpha}^{\uparrow}(A|BC)_{\rho} = \frac{\alpha}{1-\alpha} \log\Bigg( \sum_{b} p_{B}(b) 2^{\frac{1-\alpha}{\alpha}H_{\alpha}^{\uparrow}(A|C)_{\rho^{|b}}}\Bigg). 
\end{equation*}
It also contrasts the ``un-optimized'' entropy $H_{\alpha}^{\downarrow}(A|BC)_{\rho}$, given by
\begin{equation*}
        H_{\alpha}^{\downarrow}(A|BC)_{\rho} = \frac{1}{1-\alpha} \log\Bigg( \sum_{b} p_{B}(b) 2^{(1-\alpha)H_{\alpha}^{\downarrow}(A|C)_{\rho^{|b}}}\Bigg). 
\end{equation*}
In particular, we have the chain of inequalities
\begin{equation*}
    H_{\alpha}^{\downarrow}(A|BC)_{\rho} \leq H_{\alpha}(A|B^{\uparrow}C^{\downarrow})_{\rho} \leq H_{\alpha}^{\uparrow}(A|BC)_{\rho}.
\end{equation*}
The upper bound follows directly from the fact that $H_{\alpha}^{\downarrow}(A|C)_{\rho^{|b}} \leq H_{\alpha}^{\uparrow}(A|C)_{\rho^{|b}}$, and the lower bound can be seen from the convexity of the function $x \mapsto x^{\alpha}$ for $x \geq 0$ and $\alpha \in (1,\infty)$. The following variational characterization of $H_{\alpha}(A|B^{\uparrow}C^{\downarrow})_{\rho}$ also illustrates this series of inequalities.
\begin{lemma}[Variational formula]
    Let $\rho_{ABC} = \sum_{b} p_{B}(b) \, \ketbra{b}{b}_{B} \otimes \rho_{AC}^{|b} \in \mcD(ABC)$ be a cq-state classical on $B$, and let $\alpha \in (1,\infty)$. Then the partially optimized conditional R\'enyi entropy of order $\alpha$ according to \cref{def:partial_ent} is equal to the variational expression
    \begin{equation}
        H_{\alpha}(A|B^{\uparrow}C^{\downarrow})_{\rho} = \sup_{q_{B} \in \Delta(B)} - D_{\alpha}\big( \rho_{ABC} \| \id_{A} \otimes \sigma_{BC}\big), \label{eq:var_form}
    \end{equation}
    where $\sigma_{BC} = \sum_{b} q_{B}(b) \, \ketbra{b}{b} \otimes \rho_{C}^{|b}$, $\rho_{C}^{|b} = \tr_{A}[\rho_{AC}^{|b}]$ and the supremum is taken over all probability distributions on the classical register $B$. \label{lem:var_ent}
\end{lemma}
\noindent Proof can be found in \cref{sec:ent_props}. The entropy $H_{\alpha}^{\downarrow}(A|BC)_{\rho}$ is equal to the right hand side of \eqref{eq:var_form} for the feasible point $q_{B} = p_{B}$, i.e., $\sigma_{BC} = \rho_{BC}$, and is therefore a lower bound. On the other hand, $H_{\alpha}^{\uparrow}(A|BC)_{\rho}$ is equal to the right hand side of \eqref{eq:var_form} when the supremum is taken over all quantum states $\sigma_{BC}$, and is therefore an upper bound. We list some further properties that follow directly from the underlying R\'enyi divergence, also proven in \cref{sec:ent_props}.   

\begin{lemma}[Additional properties]
     Let $\rho_{ABCD} = \sum_{b,d} p_{BD}(b,d) \, \ketbra{b}{b}_{B} \otimes \ketbra{d}{d}_{D} \otimes \rho_{AC}^{|b,d} \in \mcD(ABCD)$ be a cq-state classical on $BD$, and let $\alpha \in (1,\infty)$. Then the partially optimized conditional R\'enyi entropy of order $\alpha$ according to \cref{def:partial_ent} has the following properties:
    \begin{enumerate}[(i)]
        \item When $\rho_{ABC} = \rho_{AB} \otimes \rho_{C}$, $H_{\alpha}(A|B^{\uparrow}C^{\downarrow})_{\rho} = H_{\alpha}^{\uparrow}(A|B)_{\rho}$. Similarly, when $\rho_{ABC} = \rho_{AC} \otimes \rho_{B}$,  $H_{\alpha}(A|B^{\uparrow}C^{\downarrow})_{\rho} = H_{\alpha}^{\downarrow}(A|C)_{\rho}$ (consistency with $H_{\alpha}^{\uparrow}$ and $H_{\alpha}^{\downarrow}$).
        \item Let $\mcE: C \to \hat{C}$ be a CP map. Then $H_{\alpha}(A|B^{\uparrow}C^{\downarrow})_{\rho} \leq H_{\alpha}(A|B^{\uparrow}\hat{C}^{\downarrow})_{\mcE(\rho)}$, with equality if $\mcE$ is an isometry (data processing and isometric invariance for $C$).
        \item $H_{\alpha}(AD|B^{\uparrow}C^{\downarrow})_{\rho} \geq H_{\alpha}(A|B^{\uparrow}C^{\downarrow})_{\rho}  \geq H_{\alpha}(A|B^{\uparrow}(CD)^{\downarrow})_{\rho}$ (classical registers cannot decrease entropy and subadditivity). 
    \end{enumerate} \label{lem:ent_prop}
\end{lemma}
\noindent Note that when $C$ is also classical, $\rho$ is equivalent to a probability distribution $p_{ABC}(a,b,c) = p_{B}(b)\,p_{B|C}(b|c)\,p_{A|BC}(a|b,c)$, and we have
\begin{equation*}
     H_{\alpha}(A|B^{\uparrow}C^{\downarrow})_{\rho} = \frac{\alpha}{1-\alpha} \log\Bigg( \sum_{b} p_{B}(b) \Bigg[ \sum_{a,c} p_{C|B}(c|b) \, p_{A|BC}(a|b,c)^{\alpha}\Bigg]^{\frac{1}{\alpha}}\Bigg).
\end{equation*}
The lower and upper bounds for the classical expressions \eqref{eq:classical_down} and \eqref{eq:classical_up} of $H_{\alpha}^{\downarrow}$ and $H_{\alpha}^{\uparrow}$, respectively, can then be recovered using the concavity of the function $x \mapsto x^{\frac{1}{\alpha}}$ when $\alpha \in (1,\infty)$.  

Equipped with \cref{def:partial_ent} and the subsequent properties, we can tighten the chain rule \eqref{eq:EAT_chain_2}.

\begin{corollary}[Chain rule]
    Let $\mcM : R \to A_{2}B_{2}$ be a quantum channel such that $B_{2}$ is classical and for all input states $\omega_{RR'} \in \mcD(RR')$ where $R'$ is an arbitrary register, the output state $\rho_{A_{2}B_{2}R'} = \mcM(\omega_{RR'})$ satisfies $\rho_{B_{2}R'} = \rho_{B_{2}} \otimes \omega_{R'}$. Let $\omega_{A_{1}B_{1}R} \in \mcD(A_{1}B_{1}R)$ be a quantum state. Then for any $\alpha \in (1,\infty)$,
    \begin{equation}
        H_{\alpha}^{\uparrow}(A_{1}A_{2}|B_{1}B_{2})_{\mcM(\omega)} \geq H_{\alpha}^{\uparrow}(A_{1}|B_{1})_{\omega} + \inf_{\omega'\in \mcD(R\tilde{E})} H_{\alpha}(A_{2}|B_{2}^{\uparrow}\tilde{E}^{\downarrow})_{\mcM(\omega')},
    \end{equation}
    where the infimum is taken over all pure states and the system $\tilde{E}$ has the same dimension as $R$. \label{cor:new_chain_1}
\end{corollary}
\noindent See \cref{sec:chain_prood} for proof.

\section{Application to entropy accumulation} \label{sec:REAT}
As described in \cref{sec:markov}, chain rules for the conditional R\'enyi entropy are often converted into an entropy accumulation theorem. This is achieved by successively applying the rule to the output state of a sequence of channels $\mcM_{1},...,\mcM_{n}$, where each channel is a particular case of \cref{fig:chan}. This sequence can be chosen to model a cryptographic protocol, and the EAT allows one to quantify security in terms of an entropy contribution from each round. 

Furthermore, the EAT goes beyond the $n$-fold application of a chain rule by accommodating what is refered to as ``testing''. Here, the user can collect an additional piece of classical information $C_{i}$ at each round $i \in \{1,...,n\}$, and based on the observed frequency distribution on $C^n$, decide if they want to abort the protocol. The values stored in these testing registers should indicate whether the devices are behaving as expected. For example, $C_{i}$ could encode whether a nonlocal game was won or lost on round $i$, and if too few rounds are won, the user might suspect that the devices are behaving classically and abort the protocol. If the protocol does not abort, the user can certify roughly $n$ times the entropy of a single round, where the worst case is taken over a constrained set of states defined by the non-abort event. This constrained minimization, as opposed to the unconstrained minimization that appears in the chain rules we have considered, is essential for obtaining non-trivial bounds on the single round entropy. 

A proof technique handling testing based on ``min-tradeoff functions'' was first introduced by Dupuis, Fawzi and Renner~\cite{DFR}, and the resulting second order term in the EAT was improved in follow up works~\cite{EAT2,LLR&}. More recently, Van Himbeeck and Brown~\cite{VHB25} introduced a new method for testing based on the framework of ``$f$-weighted R\'enyi entropies'', a technique that builds on the quantum probability estimation framework~\cite{Zhang18_PEF,Zhang20,knill2023} by directly encoding the trade-off function into the optimized R\'enyi entropy. This serves as a precursor to randomness extraction, resulting in a complete security proof. Subsequent works have generalized this framework~\cite{arqand2024,arqand2025,Fawzi2026}, and it has been shown to improve the rates of numerous cryptographic protocols~\cite{hahn2024,hahn2025,kamin2025,chung25,jee2025,Lu26}.  

\subsection{R\'enyi entropy accumulation}

In Appendix \ref{app:REAT}, we prove a version of the EAT derived from the tightened chain rule in \cref{cor:new_chain_1}. Recall, the setting to which our chain rule applies is a special case of the Markov conditions illustrated in~\cref{fig:EAT}. More specifically, we require that new side information generated at round $i$, denoted $B_{i}$, satisfies two conditions. Firstly, it should be classical, and secondly, it should be independent of all previous sub-systems. This was formalized in \cref{eq:DIch}. As discussed in \cref{sec:new_chain}, these conditions are typically satisfied in device-independent protocols, and our version of the EAT should be viewed with this application in mind. 

While the original EAT is concerned with lower bounding the smooth min-entropy~\cite{DFR,MetgerGEAT} (cf. \cref{sec:markov}), we follow the approach of more recent works~\cite{VHB25,arqand2024} that work directly with the R\'enyi entropy $H_{\alpha}^{\uparrow}$. The resulting bound can then be combined with a R\'enyi-based randomness extractor~\cite{Dupuis23}, or later converted into a bound on the smooth min-entropy~\cite{DFR}. Furthermore, by using the techniques of~\cite{VHB25,arqand2024} to handle testing, we obtain a single round quantity that is related to the conditional R\'enyi entropy for $\alpha > 1$. This can be compared to the original EAT proof, where one adds an additional step (via a continuity bound) to arrive at the conditional von Neumann entropy~\cite{DFR}. Our approach follows the proof of~\cite{arqand2024} closely, generalizing certain steps to accommodate our new chain rule. We provide a detailed comparison to~\cite{arqand2024} and~\cite{DFR} throughout the proof in \cref{app:REAT}.    

The most general form of our EAT is stated and discussed in \cref{thm:big_eat} of \cref{sec:big_EAT}. For simplicity, we present a special case of the theorem that is suitable for DI applications in the following section.

\subsection{Example with infrequent sampling channels}

The sequence of channels we are interested in take the following form.

\begin{definition}[Sequence of DI channels]
    Let $n$ be a positive integer, and for every $i \in \{1,...,n\}$ let $\mcM_{i} : R_{i-1} \to A_{i}C_{i}B_{i}R_{i}$ be a quantum channel, where the systems $C_{i}$ and $B_{i}$ are classical and every $C_{i}$ has the same dimension, and is isomorphic to a single register $C$. Then $\{\mcM_{i}\}_{i=1}^{n}$ is a sequence of DI channels if for any input state $\omega_{R_{0}E}$ and every $i \in \{1,...,n\}$, the state $\rho_{A^{i}C^{i}B^{i}E} = [\tr_{R_{i}} \circ \mcM_{i} \circ \cdots \circ \mcM_{1}](\omega_{R_{0}E})$ satisfies $\rho_{A^{i-1}C^{i-1}B^{i}E} = \rho_{B_{i}} \otimes \rho_{A^{i-1}C^{i-1}E}$. \label{def:DIchan}
\end{definition}
\noindent Our most general theorem (\cref{thm:big_eat}) is stated for channels of this form. In the DI context, $A_{i}$ contains the raw randomness output by the devices at round $i$, while the system $B_{i}$ contains the measurement settings at round $i$ that are sampled independently\footnote{Here, the channel $\mcM_{i}$ includes sampling the measurement settings which are therefore regarded as outputs of the channel. They can still however be understood as the inputs to a Bell test, as is commonly the case in DI protocols.}. The registers $C_{i}$ are the testing registers. Note that we treat $C_{i}$ as an independent output of $\mcM_{i}$, rather than imposing it to be a deterministic function of $A_{i}B_{i}$, as done in the original EAT~\cite{DFR}. We will then consider the joint entropy of $A^n C^n$ conditioned on the side information $B^nE$, where $E$ is an initial system held by Eve. Consider the following special case of \cref{def:DIchan}, refered to as infrequent sampling channels~\cite{EAT2,LLR&}. 

\begin{definition}[Infrequent sampling channels]
    A sequence of DI channels $\{\mcM_{i}\}_{i=1}^{n}$ (according to \cref{def:DIchan}) is infrequent sampling if the following conditions hold:
    \begin{enumerate}
        \item Each $C_{i}$ takes values in $\{0,1\}^{d} \cup \{\perp\}$ for a positive integer $d$.
        \item Each $B_{i}$ is of the form $T_{i}B_{i}$ where $T_{i}$ is a classical bit and each $B_{i}$ is isomorphic to a single register $B$ that takes values in a finite alphabet $\mcB$.
        \item Each $A_{i}$ is classical and isomorphic to a single register $A$ that takes values in a finite alphabet $\mcA$.
        \item There exists a $\gamma \in [0,1]$, distributions $p_{B}^{\text{gen}}, \, p_{B}^{\text{test}} \in \Delta(B)$, a deterministic function $f : \mcA \times \mcB \to \{0,1\}^{d}$ and for every $i \in \{1,...,n\}$, a collection of CP maps $\mcM^{a|b}_{i}:R_{i-1} \to R_{i}$ satisfying $\sum_{a}\tr[\mcM_{i}^{a|b}(\omega)] = 1$ for all input states $\omega$ and all $b\in \mcB$, such that
        \begin{equation}
        \mcM_{i} = (1-\gamma) \ketbra{0}{0}_{T_{i}} \otimes \mcM^{\text{gen}}_{i}  
        + \gamma \, \ketbra{1}{1}_{T_{i}} \otimes \mcM^{\text{test}}_{i}, \label{eq:sample_ch_main}
        \end{equation}
        where
        \begin{equation*}
            \begin{aligned}
                \mcM^{\text{gen}}_{i} = \sum_{a\in \mcA,b\in \mcB} p_{B}^{\text{gen}}(b)\ketbra{a}{a}_{A_{i}} \otimes \ketbra{\perp}{\perp}_{C_{i}} \otimes \ketbra{b}{b}_{B_{i}} \otimes \mcM_{i}^{a|b} \ \ \text{and} \\
                \mcM^{\text{test}}_{i} = \sum_{a\in \mcA,b\in \mcB} p_{B}^{\text{test}}(b)\ketbra{a}{a}_{A_{i}} \otimes \ketbra{f(a,b)}{f(a,b)}_{C_{i}} \otimes \ketbra{b}{b}_{B_{i}} \otimes \mcM_{i}^{a|b}.
            \end{aligned}
        \end{equation*}
    \end{enumerate} 
    \label{def:sample}
\end{definition}
\noindent Infrequent sampling channels can be used to describe DI protocols based on spot-checking~\cite{MS2,Arnon-Friedman2018,LLR&,Bhavsar2023}. For example, in DI randomness expansion~\cite{ColbeckThesis,CK2,PAMBMMOHLMM,MS1,MS2}, we wish to minimize the amount of randomness consumed by sampling each variable $B_{i}$. However, it is necessary that some randomness is consumed to test if the devices are behaving non-classically, i.e., to perform a Bell test~\cite{Brunner_review}. A spot-checking protocol will randomly label a small subset of rounds as ``test rounds'', in which $B_{i}$ is sampled randomly according to the desired Bell test being performed. Specifically, a round is labeled ``test'' with probability $\gamma$ (indicated by the flag $T_{i} = 1$), and a deterministic function of $A_{i}$ and $B_{i}$ is stored in the register $C_{i}$ that might, for example, indicate whether a nonlocal game was won or lost. The remaining rounds are labeled ``generate rounds'', where typically a fixed value of $B_{i}$ is chosen for the purpose of generating raw randomness. This occurs with probability $1-\gamma$ for each round (indicated by the flag $T_{i} = 0$), and the value of $C_{i}$ is set to a fixed symbol $\perp$. The probability $\gamma$ should be chosen carefully such that there are a sufficient number of test rounds to provide statistical confidence in the Bell test, but the amount of randomness consumed is much smaller than the final amount of randomness generated, i.e., there is randomness expansion. For a detailed discussion of spot-checking DI protocols, we refer to reader to~\cite[Sections IV.D and XIV.B]{pirandola2020advances} and~\cite[Section 2]{Bhavsar2023}. Note that in \cref{def:sample}, we allow for arbitrary distributions over $B_{i}$ on test and generate rounds, given by $p_{B}^{\text{test}}$ and $p_{B}^{\text{gen}}$, respectively (i.e., the value of $B_{i}$ in generate rounds need not be deterministic, encompassing a broader class of protocols~\cite{Schwonnek2021,Bhavsar2023}). 

We now present a corollary of \cref{thm:big_eat}.

\begin{corollary}[EAT with infrequent sampling]
    Let $\omega_{R_{0}E} \in \mcD(R_{0}E)$ be a state, $\{\mcM_{i}\}_{i=1}^{n}$ be a sequence of DI channels that are infrequent sampling according to \cref{def:sample} and $\rho_{A^{n}C^{n}B^{n}E} = [\tr_{R_{n}} \circ \mcM_{n} \circ \cdots \mcM_{1}](\omega_{R_{0}E})$. Let $\Omega \subset \mcC^{n}$ be an event on $C^{n}$, $\rho_{A^{n}C^{n}B^{n}E}^{|\Omega}$ be the output state conditioned on $\Omega$ and $\Delta_{\Omega} \subset \Delta(C)$ be any compact convex subset such that $\mathsf{freq}(c^{n}) \in \Delta_{\Omega}$ for all $c^{n} \in \Omega$. Then for any $\alpha \in (1,\infty)$, 
    \begin{equation*}
        H_{\alpha}^{\uparrow}(A^{n}C^{n}|B^{n}E)_{\rho^{|\Omega}} \geq n \, h_{\alpha}(\Omega) - \frac{\alpha}{\alpha - 1} \log \Big(\frac{1}{p_{\Omega}}\Big)
    \end{equation*}
    where $p_{\Omega}$ is the probability of observing $\Omega$ in $\rho_{C^{n}}$ and
    \begin{equation*}
        h_{\alpha}(\Omega) \geq \inf_{(\mcM,R,\tilde{E})} \inf_{\omega \in \mcD(R\tilde{E})}\inf_{v_{C} \in \Delta_{\Omega}} \Bigg(\frac{1}{\alpha - 1}D\big(v_{C}\|p_{C}\big)
        + v_{C}(\perp)\,H_{\alpha}(A|B^{\uparrow}\tilde{E}^{\downarrow})_{\mcM^{\mathrm{gen}}(\omega)}\Bigg),
    \end{equation*}
    where the outer infimum is over all quantum systems $R$ and $\tilde{E}$ and all quantum channels $\mcM:R\to ACB$ of the form \eqref{eq:sample_ch}, $[\tr_{AB\tilde{E}} \circ \mcM](\omega) = \sum_{c}p_{C}(c) \ketbra{c}{c}$ and $p_{C} = [p_{C}(c)]_{c\in \mcC}$. \label{cor:REAT}
\end{corollary}
\noindent See \cref{app:infreq_sample} for proof. We refer the reader to \cref{app:notation} for the definition of a frequency distribution $\mathsf{freq}(c^n)$ and the conditional state $\rho^{|\Omega}$. The function $D(v_{C}\|p_{C}) = \sum_{c} v_{C}(c) \log\big( v_{C}(c) / p_{C}(c) \big)$ if $\text{supp}(p_{C}) \subseteq \text{supp}(v_{C})$ and $+\infty$ otherwise is the Kullback–Leibler (KL) divergence~\cite{KL51}.

To gain some intuition about the single round quantity in \cref{cor:REAT}, it is instructive to make the following observation, as described in~\cite{VHB25,arqand2024}. As $n \to \infty$, $\alpha \to 1$ and $\gamma \to 0$, the pre-factor $1/(\alpha - 1)$ diverges. Thus, to minimize the objective function, the KL divergence $D\big(v_{C}\|p_{C}\big)$ must go to zero, which in turn forces $p_{C} = v_{C} \in \Delta_{\Omega}$, i.e., the distribution over $C$ of the channel output $\mcM(\omega)$ are forced to lie in the set of distributions that do not cause the protocol to abort with high probability. We also have $v_{C}(\perp) \to 1$ as $\gamma \to 0$, and $H_{\alpha}$ tends to the von Neumann entropy. Therefore, we recover an optimization problem of the form
\begin{equation*}
    \begin{aligned}
        \inf_{\mcM,R,\tilde{E}} \inf_{\omega \in \mcD(R\tilde{E})} \ & \ H(A|B\tilde{E})_{\mcM^{\mathrm{gen}}(\omega)} \\
        \text{s.t.} \ & \ p_{C} \in \Delta_{\Omega},
    \end{aligned}
\end{equation*}
which is exactly equal to the asymptotic rate of the protocol~\cite{DW,TCR}. We also refer the reader to~\cite[Section 5.2]{arqand2024} for a more detailed discussion of this single round quantity.

\cref{cor:REAT} should be directly compared to~\cite[Lemma 5.1]{arqand2024}, and the subsequent lower bound for $h_{\alpha}$ tailored to infrequent sampling channels~\cite[Eq. (93)]{arqand2024} in which the entropy contributions from the test rounds are dropped\footnote{\cref{thm:big_eat} can also be used to obtain analogous lower bounds on $h_{\alpha}$ to~\cite[Lemmas 5.1 and 5.2]{arqand2024}. We choose to present \cref{cor:REAT} due to its relative simplicity, and straightforward comparison with~\cite[Eq. (93)]{arqand2024}. Furthermore,~\cite[Eq. (93)]{arqand2024} has outperformed alternative bounds presented in~\cite{arqand2024} in the finite security analysis of certain DI protocols~\cite{hahn2024,hahn2025}}. The objective function in~\cite[Eq. (93)]{arqand2024} is of the form
\begin{equation*}
    \frac{1}{\alpha - 1}D\big(v_{C}\|p_{C}\big)
        + v_{C}(\perp)\,H_{\alpha}^{\downarrow}(A|B\tilde{E})_{\mcM^{\mathrm{gen}}(\omega)}.
\end{equation*}
Here, we can see the tightening offered by \cref{cor:REAT} as a result of our new chain rule (\cref{cor:new_chain_1}):
\begin{equation}
    H_{\alpha}(A|B^{\uparrow}\tilde{E}^{\downarrow})_{\mcM^{\mathrm{gen}}(\omega)} \geq H_{\alpha}^{\downarrow}(A|B\tilde{E})_{\mcM^{\mathrm{gen}}(\omega)}. \label{eq:ent_ineq}
\end{equation}
The significance of this difference for applications is, however, less clear. Firstly, note that if $p_{B}^{\text{gen}}$ is deterministic, i.e., $p_{B}^{\text{gen}}(b) = \delta_{b,b^*}$ where $b^* \in \mcB$ is some fixed value, we have equality in \cref{eq:ent_ineq}. This can be seen by expanding the R\'enyi entropies:
\begin{multline*}
    H_{\alpha}^{\downarrow}(A|BE)_{\rho} = \frac{1}{1-\alpha} \log \Bigg( \sum_{b}p_{B}(b)2^{(1-\alpha)H_{\alpha}(A|E)_{\rho^{|b}}}\Bigg) \\ \text{and} \ \ H_{\alpha}(A|B^{\uparrow}E^{\downarrow})_{\rho} = \frac{\alpha}{1-\alpha} \log \Bigg( \sum_{b}p_{B}(b)2^{\frac{1-\alpha}{\alpha}H_{\alpha}(A|E)_{\rho^{|b}}}\Bigg).
\end{multline*}
Typical spot-checking protocols, such those outlined in~\cite[Sections IV.D and XIV.B]{pirandola2020advances}, have this property. We thus hope to find improvements when $p_{B}^{\text{gen}}$ is close to uniform, that is, the measurement settings in generation rounds are chosen randomly. Examples of such protocols include DIQKD with a random key basis~\cite{Schwonnek2021} and DI randomness expansion with recycled input randomness~\cite{Bhavsar2023,ramanathan2025}. However, regardless of the distribution $p_{B}^{\text{gen}}$, equality holds in \cref{eq:ent_ineq} when the conditional entropies $H_{\alpha}(A|E)_{\rho^{|b}}$ are the same for different $B=b$. We expect this to be the case for DIQKD with random key basis for the following reasons. Intuitively, the advantage of a random key basis in QKD is that Eve has to guess between two randomly chosen incompatible measurements for key generation~\cite{Schwonnek2021}. Incompatibility ensures that she cannot perform an optimal attack on both bases simultaneously, resulting in a higher key rate than single basis protocols. If one basis choice (corresponding to a particular value of $B=b$) has a significantly larger entropy than another, better rates could be found from a spot-checking protocol using that basis in generation rounds. Furthermore, for symmetric protocols, it may be optimal for Eve to attack both bases evenly, rather than prefer one over the other, resulting in equal values of $H_{\alpha}(A|E)_{\rho^{|b}}$. 

Similar arguments could apply to the randomness expansion protocol with recycled input randomness considered by Bhavsar, Ragy and Colbeck~\cite{Bhavsar2023}, which is based on the Clauser-Horne-Shimony-Holt (CHSH) inequality~\cite{CHSH}. For that protocol, preliminary investigations suggest that for two-qubit strategies, the minimum values of $H_{\alpha}^{\downarrow}(A|BE)_{\rho}$ and $H_{\alpha}(A|B^{\uparrow}E^{\downarrow})_{\rho}$ compatible with a given CHSH violation differ in at most the third of fourth decimal place. This is also a consequence of each $H_{\alpha}(A|E)_{\rho^{|b}}$ being approximately equal under the best attack found by our search.

Based on this, we conclude that any improvement would be found in scenarios beyond the CHSH inequality. Good candidates are Bell inequalities whose maximum quantum violation is uniquely achieved by, i.e., self-tests~\cite{mayers2004self,SupicSelfTest}, a strategy that satisfies $H_{\alpha}(A|E)_{\rho^{|b}} \neq H_{\alpha}(A|E)_{\rho^{|b'}}$ for some $b\neq b'$. Since the self-tested state must be pure, witnessing maximum violation implies $H_{\alpha}(A|B^{\uparrow}E^{\downarrow})_{\rho} = H_{\alpha}^{\uparrow}(A|B)_{\rho}$, and the difference in the values $H_{\alpha}(A|E)_{\rho^{|b}}$ could imply that this is noticeably larger than $H_{\alpha}^{\downarrow}(A|B)_{\rho}$. For example, maximally violating the correlator part of the $I_{3322}$ Bell inequality~\cite{Froissart1981,DC_2004}, where there are three inputs and two outputs per party, self-tests a two-qubit system with this property~\cite{K20}. Performing a heuristic minimization over two-qubit strategies that violate $I_{3322}$, we found an improvement of $H_{\alpha}(A|B^{\uparrow}E^{\downarrow})_{\rho}$ over $H_{\alpha}^{\downarrow}(A|BE)_{\rho}$ in the second decimal place when $\alpha = 2$. Whether further improvements can be found with Bell inequalities tailored to higher dimensional systems with more inputs, such as the $I_{4422}$ inequality~\cite{DC_2004,VBP10} whose randomness certification properties were studied in~\cite{ZL25,ramanathan2025}, remains to be seen. 

\section{Discussion} \label{sec:disc}
The finite-size security of many protocols in quantum cryptography can be achieved by estimating the worst case conditional entropy of the output state after many rounds. Chain rules reduce this task to a round-by-round analysis, where it is sufficient to quantify the worst case output entropy of the channel associated to each round. This reduction is tight in some cases. We presented such chain rules in the literature under a unified framework, highlighting their differences and the contexts to which they apply. We then identified a potential improvement to the chain rule of~\cite{DFR} that is relevant to device-independent protocols, and is the analogue of a known result in the device-dependent setting. We provided a counterexample to this, but went on to show that some improvement is still possible, proving a new chain rule that can be used to derive a tighter version of the entropy accumulation theorem in certain contexts.  

Our counterexample provides a contrast between chain rules in the device-dependent and device-independent setting, an observation also reported in~\cite{arqand2025}. The device-dependent setting is characterized by the absence of a memory register (labeled $R_{i}$ in \cref{sec:setup}). Here, the chain rules for $H_{\alpha}^{\uparrow}$ are tight, in the sense that they match the worst case i.i.d.~value~\cite{VHB25,arqand2025,Fawzi2026}. Taking $H_{\alpha}^{\uparrow}(A^{n}|B^nE)_{\rho}$ as the figure of merit for a security proof\footnote{Quantifying security in terms of $H_{\alpha}^{\uparrow}(A^{n}|B^nE)_{\rho}$ follows directly from~\cite{Dupuis23}. For the purposes of this discussion, we do not consider the tightness of this theorem.}, this implies that the optimal attack by Eve on the $n$-round protocol is an i.i.d.~attack, i.e., the optimal output state for Eve is of the form $\rho^{\otimes n}$ where $\rho$ corresponds to the optimal attack on a single round of the protocol\footnote{Or more generally, an independent state of the form $\bigotimes_{i}\rho^{i}$ where $\rho^i$ is optimal for the channel associated to round $i$~\cite{arqand2025,Fawzi2026}.}. In the DI setting, we showed that the presence of a memory implies the existence of a state $\rho$ and non-signaling channels $\mcM_{1}$ and $\mcM_{2}$ such that $H_{\alpha}^{\uparrow}(A_{1}A_{2}|B_{1}B_{2}E)_{\rho} < \inf_{\omega}H_{\alpha}^{\uparrow}(A_{1}|B_{1}\tilde{E})_{\mcM_{1}(\omega)} + \inf_{\omega} H_{\alpha}^{\uparrow}(A_{2}|B_{2}\tilde{E})_{\mcM_{2}(\omega)}$. This prompts the following question: does there exist a sequence of channels and an initial state such that $H_{\alpha}^{\uparrow}(A^n|B^nE)_{\rho}$ is significantly smaller than $\sum_{i} \inf_{\omega} H_{\alpha}^{\uparrow}(A_i|B_i\tilde{E})_{\mcM_{i}(\omega)}$? Moreover, can this attack be designed such that the protocol aborts with a small probability? An affirmative answer would imply that for a finite $n$ and $\alpha > 1$, memory based attacks are strictly stronger than i.i.d.~attacks. This is in a similar spirit to~\cite{sandfuchs2023}, who showed this to be the case for a variant of DIQKD that uses random post-selection. Note that for $\alpha = 1 + O(1/\sqrt{n})$, the lower bound proven in this and previous works in terms of $H_{\alpha}^{\downarrow}$~\cite{DFR,MetgerGEAT,arqand2024} matches the i.i.d.~rate as $n \to \infty$. This follows from the fact that all R\'enyi entropies converge the the von Neumann entropy, recovering the asymptotic rate $H(A|BE)$~\cite{DW,TCR}. The question we ask here is if there exists a counterexample at a finite $n$. If this is not the case, the finite size rates of DI protocols could still be tightened. However, our counterexample suggests that the existing proof techniques based on~\cite{DFR} might be limited.   

We also introduced a new version of the conditional R\'enyi entropy, denoted by $H_{\alpha}(A|B^{\uparrow}E^{\downarrow})_{\rho}$. This quantity lies in between the fully optimized $H_{\alpha}^{\uparrow}$ and un-optimized $H_{\alpha}^{\downarrow}$, and in \cref{lem:var_ent} we found a variational expression in terms of the R\'enyi divergence. Specifically, rather than optimizing over all marginals in the second argument (as one would to define $H_{\alpha}^{\uparrow}$), we optimize over a subset of marginals $\sigma_{BE}$ such that for every $b$, the state of $E$ conditioned on $B=b$ is equal to that of $\rho_{BE}$, but the classical distribution over $B$ can vary. Our definition is restricted to a classical system $B$, and it would be interesting to find a fully quantum definition that is consistent with the properties in \cref{lem:ent_prop} and recovers the classical quantum definition as a special case. Furthermore, the applications of this quantity, or indeed other definitions of the conditional R\'enyi entropy that arise from a constrained marginal, could be an interesting future direction. 

Finally, it would be interesting to understand if there exists a protocol that can benefit form the new version of the entropy accumulation theorem derived in this manuscript. As discussed in \cref{sec:REAT}, the advantage can be seen in the ``stronger averaging'' of the quantity $H_{\alpha}(A|B^{\uparrow}E^{\downarrow})_{\rho}$ versus $H_{\alpha}^{\downarrow}(A|BE)_{\rho}$. Thus, protocols that generate the side information $B$ randomly on all (or at least most) rounds are good candidates. Furthermore, since both quantities converge to the von Neumann entropy $H(A|BE)$ as $\alpha \to 1$, we expect a more visible improvement when $\alpha$ deviates further from 1. We also observed in \cref{sec:REAT} that asymmetry between the values of $H_{\alpha}(A|E)_{\rho^{|b}}$ for different $B=b$ is necessary to see an improvement. In addition, applying numerical tools such as those in~\cite{hahn2024} to compute DI lower bounds on the partially optimized entropy is a related direction to consider.     

\section*{Acknowledgments}
LW thanks Rutvij Bhavsar and Shashank Kumar Ranu for fruitful discussions. We are grateful for funding support by the European
Union Horizon Europe research and innovation program
under the project “Quantum Secure Networks Partnership” (QSNP, Grant Agreement No. 101114043), by ChistEra-2023/05/Y/ST2/00005 under the project Modern Device Independent Cryptography (MoDIC) and by a government grant managed by the Agence Nationale de la Recherche under the Plan France 2030 with the reference ANR-22-PETQ-0009. OF acknowledges funding by the European Research Council (ERC Grant AlgoQIP, Agreement No. 851716). PB is supported by the Agence Nationale de la Recherche (ANR) through the JCJC programme under grant number ANR-25-CE47-3449.


\newcommand{\etalchar}[1]{$^{#1}$}

\appendix

\section{Additional notation} \label{app:notation}
Recall that for a classical system $C$ with a finite alphabet $\mcC$, the set of probability distributions on $C$ is denoted by $\Delta(C)$. The $n$-fold cartesian product of $\mcC$ is denoted $\mcC^{n}$, and we define the frequency distribution $\mathsf{freq}(c^{n}) \in \Delta(C)$ induced by a string $c^{n} \in \mcC^{n}$ by
\begin{equation*}
    \mathsf{freq}(c^{n})(c) := \frac{\big | \big\{ k \, : \, c_{k} = c \big\} \big| }{n}, \ \ \ \forall c \in \mcC.
\end{equation*}
Vectors in $\mbR^{m}$ for a positive integer $m$ are denoted $f = [f_{k}]_{k=1}^{m}\in \mbR^{m}$. 

We also formalize some notation for cq-states. A state $\rho_{CQ}$ is a cq-state, where $C$ is a classical register that takes values from a finite alphabet $\mcC$, if
\begin{equation*}
    \rho_{CQ} = \sum_{c \in \mcC} p_{C}(c) \, \ketbra{c}{c}_{C} \otimes \rho_{Q}^{|c},
\end{equation*}
where $p_{C} \in \Delta(C)$ and $\{\rho_{Q}^{|c}\}_{c \in \mcC} \subseteq \mcD(Q)$. For any event $\Gamma \subset \mcC$ on $C$, $\rho_{CQ}^{|\Gamma}$ denotes the normalized state conditioned on $\Gamma$, 
\begin{equation*}
    \rho_{CQ}^{|\Gamma} = \frac{1}{p_{\Gamma}} \sum_{c \in \Gamma} p_{C}(c) \, \ketbra{c}{c}_{C} \otimes \rho_{Q}^{|c},
\end{equation*}
where $p_{\Gamma} = \sum_{c \in \Gamma} p_{C}(c)$. We also denote the sub-normalized state by $\rho^{\Gamma}_{CQ} = p_{\Gamma} \, \rho_{CQ}^{|\Gamma}$. Recall the conditional operator of a bipartite state $\rho_{AB} \in \mcD(AB)$ is given by $\rho_{A|B} = \rho_{B}^{-\frac{1}{2}} \rho_{AB} \rho_{B}^{-\frac{1}{2}}$, and for a state $\sigma_{B} \in \mcD(B)$, $\Pi(\sigma_{B})$ denotes the projector onto its support. 

Let $\rho_{CQ}$ be a cq-state, and $\sigma_{CQ} = \sum_{c} q_{C}(c) \, \ketbra{c}{c}_{C} \otimes \sigma_{Q}^{|c}$ be a positive operator such that $\text{supp}(\rho_{CQ}) \subseteq \text{supp}(\sigma_{CQ})$. Then we recall that the sandwiched R\'enyi divergence admits the following decomposition (see, e.g.,~\cite[Eq. (5.32)]{Tomamichel_2016})
\begin{equation}
    D_{\alpha}(\rho_{CQ} \| \sigma_{CQ}) = \frac{1}{\alpha-1} \log \Bigg( \sum_{c} p_{C}(c)^{\alpha} q_{C}(c)^{1-\alpha} 2^{(\alpha-1) D_{\alpha}(\rho_{Q}^{|c} \| \sigma_{Q}^{|c})}\Bigg). \label{eq:cq_div}
\end{equation}
Furthermore, for a state $\rho_{CAQ} = \sum_{c} p_{C}(c)\ketbra{c}{c} \otimes \rho_{AQ}^{|c}$ classical on $C$, we have~\cite[Proposition 5.4]{Tomamichel_2016}
\begin{equation}
    \begin{aligned}
        H_{\alpha}^{\downarrow}(A|CQ)_{\rho} &= \frac{1}{1-\alpha} \log \Bigg( \sum_{c}p_{C}(c) 2^{(1-\alpha)H_{\alpha}^{\downarrow}(A|Q)_{\rho^{|c}}} \Bigg)\ \ \ \text{and} \\
        H_{\alpha}^{\uparrow}(A|CQ)_{\rho} &= \frac{\alpha}{1-\alpha} \log \Bigg( \sum_{c}p_{C}(c) 2^{\frac{1-\alpha}{\alpha}H_{\alpha}^{\uparrow}(A|Q)_{\rho^{|c}}}\Bigg).
    \end{aligned} \label{eq:c_ent}
\end{equation}

\section{Proofs for \cref{sec:new_chain}: chain rules} \label{app:new_chain_proofs}

\subsection{Solving the minimization in the counterexample} \label{app:counter}
In this subsection, we wish to evaluate the third term in the hypothesized chain rule \eqref{eq:EAT_wrong} for the classical state and channels described in~\cref{sec:counter}. Specifically, we need to evaluate
\begin{equation}
    \inf_{\omega' \in \mcD(R_{1}\tilde{E})} H_{\alpha}^{\uparrow}(A_{2}|B_{2}\tilde{E})_{\mcM_{2}(\omega')}, \label{eq:app_opt}
\end{equation}
where the infimum is taken over all cq-states $\omega_{R_{1}\tilde{E}}' = \sum_{r_{1}} q(r_{1})\ketbra{r_{1}}{r_{1}}_{R_{1}} \otimes \rho_{\tilde{E}}^{|r_{1}}$ with $\rho_{\tilde{E}}^{|r_{1}} \in \mcD(\tilde{E})$. For ease of reading, recall the channel $\mcM_2 : R_1 \to A_2 B_2$ is a fully classical channel that first generates $B_2$ as a uniformly random bit and then generates $A_2$ according to the conditional distribution
\begin{equation}
    p_{A_2|R_1B_2}(a_2|r_1,b_2) = \begin{cases}
        1/2 & \quad \text{if } r_1 \oplus b_2 = 0 \\
        1 & \quad \text{if } r_1 \oplus b_2 = 0  \,\wedge\, a_2=0\\
        0 & \quad \text{if } r_1 \oplus b_2 = 0  \,\wedge\, a_2=1\,.
    \end{cases}
\end{equation}We first show it is sufficient to consider cq-states in which $\{\rho_{\tilde{E}}^{|r_{1}}\}_{r_{1}}$ is a set of orthogonal pure states.

\begin{claim}
    The optimization \eqref{eq:app_opt} is achieved by a cq-state of the form $\omega_{R_{1}\tilde{E}} = \sum_{r_{1}} q(r_{1}) \ketbra{r_{1}}{r_{1}}_{R_{1}} \otimes \ketbra{\phi_{r_{1}}}{\phi_{r_{1}}}_{\tilde{E}}$, where each $\ket{\phi_{r_{1}}} \in \mcH_{\tilde{E}}$ for $i \in \{0,1\}$ is a pure state and $\braket{\phi_{0}}{\phi_{1}} = 0$. \label{claim:counter}
\end{claim}
\begin{proof}
    Let $\omega'_{R_{1}\tilde{E}}= \sum_{r_{1}} q(r_{1}) \ketbra{r_{1}}{r_{1}}_{R_{1}} \otimes \rho_{\tilde{E}}^{|r_{1}}$ be an arbitrary cq-state that is input to $\mcM_{2}$. Let us denote the spectral decomposition of each $\rho_{\tilde{E}}^{|r_{1}}$ for $r_{1}\in \{0,1\}$ by
    \begin{equation*}
        \rho_{\tilde{E}}^{|r_{1}} = \sum_{k} q(k|r_{1}) \ketbra{\phi^{k}_{r_{1}}}{\phi^{k}_{r_{1}}}_{\tilde{E}} 
    \end{equation*}
    where $q(k|r_{1}) \geq 0$, $\sum_{k}q(k|r_{1}) = 1$ and $\{\ket{\phi_{r_{1}}^{k}}\}_{k}$ forms an orthonormal basis for $\mcH_{\tilde{E}}$. By defining $q(k) = \sum_{r_{1}} q(r_{1})\, q(k|r_{1})$ and $q(r_{1}|k) = q(r_{1})q(k|r_{1})/q(k)$, we can write
    \begin{equation*}
        \omega'_{R_{1}\tilde{E}} = \sum_{k}q(k) \sum_{r_1} q(r_1|k) \ketbra{r_{1}}{r_{1}}_{R_{1}} \otimes \ketbra{\phi_{r_1}^k}{\phi_{r_1}^k}_{\tilde{E}}.
    \end{equation*}
    Then by the linearity of $\mcM_{2}$,
    \begin{equation*}
        \mcM_{2}(\omega') = \sum_{k}q(k) \rho_{A_{2}B_{2}\tilde{E}}^{|k},
    \end{equation*}
    where we defined the cq-state
    \begin{multline*}
        \rho_{A_{2}B_{2}\tilde{E}}^{|k} := \sum_{b_{2}} p_{B_{2}}(b_{2}) \ketbra{b_{2}}{b_{2}}_{B_{2}} \otimes \sum_{a_{2}} \ketbra{a_{2}}{a_{2}}_{A_{2}} \otimes \sum_{r_{1}}q(r_{1}|k) \, p_{A_{2}|B_{2}R_{1}}(a_{2}|b_{2},r_{1}) \ketbra{\phi^{k}_{r_{1}}}{\phi^{k}_{r_{1}}}_{\tilde{E}} \\= \mcM_{2}\Bigg(\sum_{r_{1}}q(r_{1}|k)\ketbra{r_{1}}{r_{1}}_{R_{1}} \otimes \ketbra{\phi_{r_{1}}^{k}}{\phi_{r_{1}}^{k}}_{\tilde{E}}\Bigg).
    \end{multline*}
    Using the quasi-concavity of the optimized sandwich R\'enyi entropy~\cite[Section 5.2.3]{Tomamichel_2016},
    \begin{equation*}
        H_{\alpha}^{\uparrow}(A_{2}|B_{2}\tilde{E})_{\mcM_{2}(\omega')} = H_{\alpha}^{\uparrow}(A_{2}|B_{2}\tilde{E})_{\sum_{k}q(k) \rho^{|k}} \geq \min_{k} H_{\alpha}^{\uparrow}(A_{2}|B_{2}\tilde{E})_{\rho^{|k}}.
    \end{equation*}
    Therefore, the infimum must be achieved by a cq-state where each $\rho_{\tilde{E}}^{|r_{1}} = \ketbra{\phi_{r_{1}}}{\phi_{r_{1}}}$ is a pure state. 

    We now show that choosing $\ket{\phi_{r_{1}}}$ such that they satisfy $\braket{\phi_{0}}{\phi_{1}}$ is optimal. Let 
    $\{\ket{\psi_{l}}\}_{l}$ be any orthonormal basis for $\mcH_{\tilde{E}}$ and define the following operators:
    \begin{equation*}
        K_{0} := \ket{\phi_{0}} \otimes \bra{\psi_{0}}, \ K_{1} := \ket{\phi_{1}} \otimes \bra{\psi_{1}} , \ \text{and} \ K_{l} := \ket{\phi_{0}} \otimes \bra{\psi_{l}}, \ \forall l > 1.   
    \end{equation*}
    Note that $\sum_{l}K_{l}^{\dagger}K_{l} = \sum_{l} \ketbra{\psi_{l}}{\psi_{l}} = \id_{\tilde{E}}$, hence the set $\{K_{l}\}_{l}$ is a set of Kraus operators. Let $\mcK$ denote the corresponding quantum channel, and note that
    \begin{equation*}
        \omega_{R_{1}\tilde{E}} = \sum_{r_{1}}q(r_{1})\ketbra{r_{1}}{r_{1}}_{R_{1}} \otimes \ketbra{\phi_{r_{1}}}{\phi_{r_{1}}}_{\tilde{E}} = \mcK\Bigg(\sum_{r_{1}}q(r_{1})\ketbra{r_{1}}{r_{1}}_{R_{1}} \otimes \ketbra{\psi_{r_{1}}}{\psi_{r_{1}}}_{\tilde{E}}\Bigg) =: \mcK(\bar{\omega}_{R_{1}\tilde{E}}).
    \end{equation*}
    Furthermore, since $\mcK$ acts on $\tilde{E}$ and $\mcM_{2}$ acts on $R_{1}$ they commute, i.e., $\mcK \circ \mcM_{2} = \mcM_{2} \circ \mcK$. Using these facts,
    \begin{equation*}
        H_{\alpha}^{\uparrow}(A_{2}|B_{2}\tilde{E})_{\mcM_{2}(\omega)} = H_{\alpha}^{\uparrow}(A_{2}|B_{2}\tilde{E})_{[\mcK \circ \mcM_{2}](\bar{\omega})} \geq H_{\alpha}^{\uparrow}(A_{2}|B_{2}\tilde{E})_{\mcM_{2}(\bar{\omega})}, 
    \end{equation*}
    where the inequality follows from the data processing inequality on the system $\tilde{E}$. Therefore, for every cq-state in $\omega_{R_{1}\tilde{E}}$, there exists another cq-state $\bar{\omega}_{R_{1}\tilde{E}}$ with $\rho_{\tilde{E}}^{|0}$ and $\rho_{\tilde{E}}^{|1}$ being orthogonal that has either the same or a smaller entropy, proving the claim. 
\end{proof}

A consequence of \cref{claim:counter} is that, after noting the entropy is invariant under the choice of basis for $\tilde{E}$, we can restrict the infimum in \eqref{eq:app_opt} to states of the form
\begin{equation*}
    \omega'_{R_{1}\tilde{E}} = \sum_{r_{1}} q(r_{1}) \ketbra{r_{1}}{r_{1}} \otimes \ketbra{r_{1}}{r_{1}}_{\tilde{E}}
\end{equation*}
that are fully classical, and described by the probability distribution $p_{R_{1}\tilde{E}}(r_{1},e) = p_{\tilde{E}}(e) \, p_{R_{1}|\tilde{E}}(r_{1}|e)$ where $p_{\tilde{E}}(e) =q(e)$ and $p_{R_{1}|\tilde{E}}(r_{1}|e) = \delta_{r_{1},e}$. With this in mind, the output distribution of the channel $\mcM_{2}$ is given by
\begin{equation*}
    p_{A_{2}B_{2}\tilde{E}}(a_{2},b_{2},e) = \sum_{r_{1}} p_{R_{1}\tilde{E}}(r_{1},e)\,p_{A_{2}B_{2}|R_{1}}(a_{2},b_{2}|r_{1}) = p_{\tilde{E}}(e) \, p_{B_{2}}(b_{2})  \, p_{A_{2}|R_{1}B_{2}}(a_{2}|e,b_{2}).
\end{equation*}
Passing the optimization inside the logarithm, \eqref{eq:app_opt} is equal to
\begin{equation*}
    \frac{\alpha}{1-\alpha} \log \Bigg( \sup_{p_{\tilde{E}}} \sum_{e} p_{\tilde{E}}(e)\sum_{b_{2}} p_{B_{2}}(b_{2}) \Bigg[ \sum_{a_{2}} p_{A_{2}|R_{1}B_{2}}(a_{2}|e,b_{2})^{\alpha} \Bigg]^{\frac{1}{\alpha}}\Bigg),  
\end{equation*}
where the supremum is over all probability distribution $p_{\tilde{E}}$ on the system $\tilde{E}$. Since the objective function is affine in the distribution $p_{\tilde{E}}$, the optimum must be attained at the boundary of the feasible set, resulting in the expression
\begin{equation*}
    \inf_{\omega' \in \mcD(R_{1}\tilde{E})} H_{\alpha}^{\uparrow}(A_{2}|B_{2}\tilde{E})_{\mcM_{2}(\omega')} = \frac{\alpha}{1-\alpha} \log \Bigg( \max_{e \in \{0,1\}} \sum_{b_{2}} p_{B_{2}}(b_{2}) \Bigg[ \sum_{a_{2}} p_{A_{2}|R_{1}B_{2}}(a_{2}|e,b_{2})^{\alpha} \Bigg]^{\frac{1}{\alpha}}\Bigg)
\end{equation*}
By substituting in the values of the distributions from \cref{eq:classical_ch1} and evaluating the maximum, we obtain the claimed value in \eqref{eq:app_claim}. 

\subsection{Properties of the partially optimized conditional R\'enyi entropy} \label{sec:ent_props}

\noindent \textbf{Lemma 4.2} (Variational expression)\textbf{.} \textit{Let $\rho_{ABC} = \sum_{b} p(b) \, \ketbra{b}{b}_{B} \otimes \rho_{AC}^{|b} \in \mcD(ABC)$ be a cq-state classical on $B$, and let $\alpha \in (1,\infty)$. Then the partially optimized conditional R\'enyi entropy of order $\alpha$ according to \cref{def:partial_ent} admits the following variational expression:
    \begin{equation}
        H_{\alpha}(A|B^{\uparrow}C^{\downarrow})_{\rho} = \sup_{q_{B} \in \Delta(B)} - D_{\alpha}\big( \rho_{ABC} \| \id_{A} \otimes \sigma_{BC}\big) \label{eq:var_form_app}
    \end{equation}
    where $\sigma_{BC} = \sum_{b} q_{B}(b) \, \ketbra{b}{b} \otimes \rho_{C}^{|b}$, $\rho_{C}^{|b} = \tr_{A}[\rho_{AC}^{|b}]$ and the supremum is taken over all probability distributions on the classical register $B$.}
\begin{proof}
    Let $\sigma_{BC} = \sum_{b} q_{B}(b) \, \ketbra{b}{b} \otimes \rho_{C}^{|b}$ be the marginal state constructed from a feasible point $q_{B} \in \Delta(B)$. By \cref{eq:cq_div}, the divergence is equal to
    \begin{equation}
        -D_{\alpha}\big( \rho_{ABC} \| \id_{A} \otimes \sigma_{BC}\big) = \frac{1}{1 - \alpha} \log \Bigg( \sum_{b} p_{B}(b)^{\alpha} q_{B}(b)^{1-\alpha} 2^{(\alpha-1)D_{\alpha}(\rho_{AC}^{|b} \| \id_{A} \otimes \rho_{C}^{|b})}\Bigg). \label{eq:diverg_cq}
    \end{equation}
    Furthermore, $(\alpha-1)D_{\alpha}(\rho_{AC}^{|b} \| \id_{A} \otimes \rho_{C}^{|b}) =(1-\alpha)H_{\alpha}^{\downarrow}(A|C)_{\rho^{|b}}$ since $\rho_{C}^{|b} = \tr_{A}[\rho_{AC}^{|b}]$. Let
    \begin{equation*}
        r_{b} := p_{B}(b)^{\alpha} 2^{(1-\alpha)H_{\alpha}^{\downarrow}(A|C)_{\rho^{|b}}}.
    \end{equation*}
    Passing the supremum inside the logarithm, we have the following optimization problem:
    \begin{equation*}
        \begin{aligned}
            \inf \ & \ \sum_{b}q_{B}(b)^{1-\alpha} r_{b} \\
            \text{subject to:} \ & \ \sum_{b} q_{B}(b) = 1,\\
            & \ q_{B}(b) \geq 0.
        \end{aligned}
    \end{equation*}
    It is known that this is a convex optimization problem, and the optimal value is given by $\Big( \sum_{b} r_{b}^{\frac{1}{\alpha}} \Big)^{\alpha}$
    at the point $q_{B}^*(b) = r_{b}^{\frac{1}{\alpha}} / \Big(\sum_{b} r_{b}^{\frac{1}{\alpha}}\Big)$ (see, e.g.,~\cite[Proposition 5.4]{Tomamichel_2016}). Inserting this into \eqref{eq:diverg_cq} we obtain \eqref{eq:var_form_app} as claimed.
\end{proof}

\noindent \textbf{Lemma 4.3} (Additional properties)\textbf{.} \textit{Let $\rho_{ABCD} = \sum_{b,d} p_{BD}(b,d) \, \ketbra{b}{b}_{B} \otimes \ketbra{d}{d}_{D} \otimes \rho_{AC}^{|b,d} \in \mcD(ABCD)$ be a cq-state classical on $BD$, and let $\alpha \in (1,\infty)$. Then the partially optimized conditional R\'enyi entropy of order $\alpha$ according to \cref{def:partial_ent} has the following properties:
    \begin{enumerate}[(i)]
        \item When $\rho_{ABC} = \rho_{AB} \otimes \rho_{C}$, $H_{\alpha}(A|B^{\uparrow}C^{\downarrow})_{\rho} = H_{\alpha}^{\uparrow}(A|B)_{\rho}$. Similarly, when $\rho_{ABC} = \rho_{AC} \otimes \rho_{B}$,  $H_{\alpha}(A|B^{\uparrow}C^{\downarrow})_{\rho} = H_{\alpha}^{\downarrow}(A|C)_{\rho}$ (consistency with $H_{\alpha}^{\uparrow}$ and $H_{\alpha}^{\downarrow}$).
        \item Let $\mcE: C \to \hat{C}$ be a CP map. Then $H_{\alpha}(A|B^{\uparrow}C^{\downarrow})_{\rho} \leq H_{\alpha}(A|B^{\uparrow}\hat{C}^{\downarrow})_{\mcE(\rho)}$, with equality if $\mcE$ is an isometry (data processing and isometric invariance for $C$).
        \item $H_{\alpha}(AD|B^{\uparrow}C^{\downarrow})_{\rho} \geq H_{\alpha}(A|B^{\uparrow}C^{\downarrow})_{\rho}  \geq H_{\alpha}(A|B^{\uparrow}(CD)^{\downarrow})_{\rho}$ (classical registers cannot decrease entropy and subadditivity). 
    \end{enumerate}}
\begin{proof}
    \noindent \underline{Part $(i)$}

    \vspace{0.2cm}

    \noindent Consider any cq-state classical on $B$ of the form $\rho_{ABC} = \sum_{b}p_{B}(b) \, \ketbra{b}{b}_{B} \otimes \rho_{A}^{|b} \otimes \rho_{C}$. We see that $H_{\alpha}^{\downarrow}(A|C)_{\rho^{|b}} = H_{\alpha}^{\downarrow}(A)_{\rho^{|b}} =  H_{\alpha}^{\uparrow}(A)_{\rho^{|b}}$ for all $b$. Hence
    \begin{equation*}
        H_{\alpha}(A|B^{\uparrow}C^{\downarrow})_{\rho} = \frac{\alpha}{1-\alpha} \log \Bigg( \sum_{b} p_{B}(b) 2^{\frac{1-\alpha}{\alpha}H_{\alpha}^{\uparrow}(A)_{\rho^{|b}}}\Bigg) = H_{\alpha}^{\uparrow}(A|B)_{\rho}.
    \end{equation*}
    If instead $\rho_{ABC} = \sum_{b}p_{B}(b) \, \ketbra{b}{b}_{B} \otimes \rho_{AC}$, $\rho_{AC}^{|b} = \rho_{AC}$ and we find
    \begin{equation*}
        H_{\alpha}(A|B^{\uparrow}C^{\downarrow})_{\rho} = \frac{\alpha}{1-\alpha} \log \Bigg( \sum_{b} p_{B}(b) 2^{\frac{1-\alpha}{\alpha}H_{\alpha}^{\downarrow}(A|C)_{\rho}}\Bigg) = H_{\alpha}^{\downarrow}(A|C)_{\rho}.
    \end{equation*}

    \vspace{0.2cm}

    \noindent \underline{Part $(ii)$}

    \vspace{0.2cm}

    \noindent Consider any cq-state classical on $B$ of the form $\rho_{ABC} = \sum_{b}p_{B}(b) \, \ketbra{b}{b}_{B} \otimes \rho_{AC}^{|b}$, and note that for any $\mcE : C \to \hat{C}$ we have $\mcE(\rho_{ABC}) = \sum_{b}p_{B}(b) \, \ketbra{b}{b}_{B} \otimes \mcE(\rho_{AC}^{|b})$. Furthermore, by the data processing inequality $H_{\alpha}^{\downarrow}(A|C)_{\rho^{|b}} \leq H_{\alpha}^{\downarrow}(A|\hat{C})_{\mcE(\rho^{|b})}$. Putting this together, 
    \begin{align*}
        H_{\alpha}(A|B^{\uparrow}\hat{C}^{\downarrow})_{\mcE(\rho)} &= \frac{\alpha}{1-\alpha} \log \Bigg( \sum_{b} p_{B}(b) 2^{\frac{1-\alpha}{\alpha}H_{\alpha}^{\downarrow}(A|\hat{C})_{\mcE(\rho^{|b})}}\Bigg) \\ 
        &\geq \frac{\alpha}{1-\alpha} \log \Bigg( \sum_{b} p_{B}(b) 2^{\frac{1-\alpha}{\alpha}H_{\alpha}^{\downarrow}(A|C)_{\rho^{|b}}}\Bigg) \\
        &= H_{\alpha}(A|B^{\uparrow}C^{\downarrow})_{\rho}
    \end{align*}
    as claimed. Note that if $\mcE$ is an isometry, the R\'enyi entropy satisfies $H_{\alpha}^{\downarrow}(A|C)_{\rho^{|b}} = H_{\alpha}^{\downarrow}(A|\hat{C})_{\mcE(\rho^{|b})}$ and we have equality above.

    \vspace{0.2cm}

    \noindent \underline{Part $(iii)$}

    \vspace{0.2cm}

    Consider the cq-state $\rho_{ABCD} = \sum_{b}p_{B}(b) \ketbra{b}{b}_{B} \otimes \rho_{ACD}^{|b}$ where $\rho_{ACD}^{|b} = \sum_{d}p_{D|B}(d|b) \ketbra{d}{d} \otimes \rho_{AC}^{|b,d}$. Since $D$ is classical, $H_{\alpha}^{\downarrow}(AD|C)_{\rho^{|b}} \geq H_{\alpha}^{\downarrow}(A|C)_{\rho^{|b}}$ \cite[Lemma 5.14]{Tomamichel_2016}. This implies
    \begin{align*}
        H_{\alpha}(A|B^{\uparrow}C^{\downarrow})_{\rho} &= \frac{\alpha}{1-\alpha} \log \Bigg( \sum_{b} p_{B}(b) 2^{\frac{1-\alpha}{\alpha}H_{\alpha}^{\downarrow}(A|C)_{\rho^{|b}}}\Bigg) \\ 
        &\leq \frac{\alpha}{1-\alpha} \log \Bigg( \sum_{b} p_{B}(b) 2^{\frac{1-\alpha}{\alpha}H_{\alpha}^{\downarrow}(AD|C)_{\rho^{|b}}}\Bigg) \\
        &= H_{\alpha}(AD|B^{\uparrow}C^{\downarrow})_{\rho},
    \end{align*}
    establishing the first inequality. 

    For the second, consider the variational expression from \cref{lem:var_ent}:
    \begin{multline}
        H_{\alpha}(A|B^{\uparrow}C^{\downarrow})_{\rho} = \sup \Big\{ -D_{\alpha}(\rho_{ABC} \| \id_{A} \otimes \sigma_{BC}) \ : \ \sigma_{BC} = \sum_{b} q_{B}(b) \ketbra{b}{b} \otimes \rho_{C}^{|b}, \ q_{B} \in \Delta(B)\Big\} \\
        \geq \sup \Big\{ -D_{\alpha}(\rho_{ABCD} \| \id_{A} \otimes \sigma_{BCD}) \ : \ \tr_{D}[\sigma_{BCD}] = \sum_{b} q_{B}(b) \ketbra{b}{b} \otimes \rho_{C}^{|b}, \ q_{B} \in \Delta(B)\Big\}, \label{eq:p_tr}
    \end{multline}
    where we used the data processing for $D_{\alpha}$ under the partial trace over register $D$. Recall that $\rho_{C}^{|b} = \sum_{d}p_{D|B}(d|b)\rho_{C}^{|b,d}$, and therefore every state in the set
    \begin{equation*}
        \Big\{ \sigma_{BCD} \in \mcD(BCD) \ : \ \sigma_{BCD} = \sum_{b,d}q_{B}(b)\,p_{D|B}(d|b) \ketbra{b}{b}_{B} \otimes \ketbra{d}{d}_{D} \otimes \rho_{C}^{|b,d}, \ q_{B} \in \Delta(B)\Big\}
    \end{equation*}
    satisfies the partial trace constraint in \eqref{eq:p_tr}. Restricting the supremum in \eqref{eq:p_tr} to this set shows that $H_{\alpha}(A|B^{\uparrow}C^{\downarrow})_{\rho}$ is lower bounded by
    \begin{multline*}
        \sup \Big\{ -D_{\alpha}(\rho_{ABCD} \| \id_{A} \otimes \sigma_{BCD}) \ : \ \sigma_{BCD} = \sum_{b,d}q_{B}(b)\,p_{D|B}(d|b) \ketbra{b}{b}_{B} \otimes \ketbra{d}{d}_{D} \otimes \rho_{C}^{|b,d}, \ q_{B} \in \Delta(B) \Big\}
    \end{multline*}
    which is equal to $H_{\alpha}(A|B^{\uparrow}(CD)^{\downarrow})_{\rho}$, concluding the proof.
\end{proof}

\subsection{Proofs for the new chain rule} \label{sec:chain_prood}

To prove our new chain rule, we follow closely the approach taken in~\cite[Section 3]{DFR} to establish \cref{eq:EAT_chain_1}. In particular, we require the following lemma.
\begin{lemma}[\cite{DFR} Lemma 3.1]
    Let $\rho_{A_{1}A_{2}B} \in \mcD(A_{1}A_{2}B)$ and $\sigma_{B} \in \mcD(B)$ be states, $\alpha \in (0,\infty)$ and define $\alpha' =  (\alpha - 1)/\alpha$. Then the following equality holds:
    \begin{equation}
        -D_{\alpha}( \rho_{A_{1}A_{2}B}  \| \id_{A_{1}A_{2}} \otimes \sigma_{B}) =  - D_{\alpha}(\rho_{A_{1}B} \| \id_{A_{1}} \otimes \sigma_{B}) + H_{\alpha}^{\downarrow}(A_{2}|A_{1}B)_{\nu},
    \end{equation}
    where
    \begin{equation}
        \nu_{A_{1}B} = \frac{\Big(\rho_{A_{1}B}^{\frac{1}{2}} \sigma_{B}^{-\alpha'} \rho_{A_{1}B}^{\frac{1}{2}} \Big)^{\alpha}}{\tr\Big[ \Big(\rho_{A_{1}B}^{\frac{1}{2}} \sigma_{B}^{-\alpha'} \rho_{A_{1}B}^{\frac{1}{2}} \Big)^{\alpha}\Big]}, \ \ \ \text{and} \ \ \ \nu_{A_{1}A_{2}B} = \nu_{A_{1}B}^{\frac{1}{2}} \, \rho_{A_{2}|A_{1}B} \, \nu_{A_{1}B}^{\frac{1}{2}}.
    \end{equation}
    \label{lem:DFR_chain}
\end{lemma}
\noindent We will also need the following modifications of~\cite[Claim 3.4]{DFR} and~\cite[Corollary 3.5]{DFR}, respectively.
\begin{claim}
    Let $\rho_{A_{1}B_{1}A_{2}B_{2}} = \sum_{b_{2}}p_{B_{2}}(b_{2}) \ketbra{b_{2}}{b_{2}} \otimes \rho_{A_{1}B_{1}A_{2}}^{|b_{2}} \in \mcD(A_{1}B_{1}A_{2}B_{2})$ be a cq-state classical on $B_{2}$ such that $\rho_{A_{1}B_{1}B_{2}} = \rho_{A_{1}B_{1}} \otimes \rho_{B_{2}}$, and let $\sigma_{B_{1}} \in \mcD(B_{1})$ be any state that satisfies $\mathrm{supp}(\rho_{B_{1}}) \subseteq \mathrm{supp}(\sigma_{B_{1}})$. Then for any $\alpha \in (1,\infty)$, define the cq-state
    \begin{equation}
        \nu_{A_{1}B_{1}A_{2}B_{2}} := \sum_{b_{2}}p_{B_{2}}(b_{2}) \ketbra{b_{2}}{b_{2}}_{B_{2}} \otimes \nu^{|b_{2}}_{A_{1}A_{2}B_{1}}
    \end{equation}
    where 
    \begin{equation}
        \nu_{A_{1}B_{1}} := \frac{\Big(\rho_{A_{1}B_{1}}^{\frac{1}{2}} \sigma_{B_{1}}^{-\alpha'} \rho_{A_{1}B_{1}}^{\frac{1}{2}} \Big)^{\alpha}}{\tr \Big[\Big(\rho_{A_{1}B_{1}}^{\frac{1}{2}} \sigma_{B_{1}}^{-\alpha'} \rho_{A_{1}B_{1}}^{\frac{1}{2}}\Big)^{\alpha}\Big]}, \ \ \ \text{and} \ \ \ \nu_{A_{1}A_{2}B_{1}}^{|b_{2}} := \nu_{A_{1}B}^{\frac{1}{2}} \, \rho_{A_{1}B_{1}}^{-\frac{1}{2}} \, \rho_{A_{1}A_{2}B_{1}}^{|b_{2}} \, \rho_{A_{1}B_{1}}^{-\frac{1}{2}} \, \nu_{A_{1}B}^{\frac{1}{2}}. \label{eq:nu_state}
    \end{equation}
    Then $\nu_{A_{1}B_{1}A_{2}B_{2}}$ satisfies $\nu_{A_{2}B_{2}|A_{1}B_{1}} = \rho_{A_{2}B_{2}|A_{1}B_{1}}$.
    \label{lem:nu_prop}
\end{claim}
\begin{proof}
    First note that the condition $\mathrm{supp}(\rho_{B_{1}}) \subseteq \mathrm{supp}(\sigma_{B_{1}})$ implies that $\nu_{A_{1}B_{1}}$ is a well defined state. By direct calculation,
    \begin{equation*}
        \begin{aligned}
            \nu_{A_{2}B_{2}|A_{1}B_{1}} &= \nu_{A_{1}B_{1}}^{-\frac{1}{2}} \nu_{A_{1}B_{1}A_{2}B_{2}}\nu_{A_{1}B_{1}}^{-\frac{1}{2}} \\
            &= \sum_{b_{2}}p_{B_{2}}(b_{2})\ketbra{b_{2}}{b_{2}}_{B_{2}} \otimes \nu_{A_{1}B_{1}}^{-\frac{1}{2}} \nu_{A_{1}B_{1}A_{2}}^{|b_{2}}\nu_{A_{1}B_{1}}^{-\frac{1}{2}} \\
            &= \sum_{b_{2}}p_{B_{2}}(b_{2})\ketbra{b_{2}}{b_{2}}_{B_{2}} \otimes \Pi(\nu_{A_{1}B_{1}})\rho_{A_{1}B_{1}}^{-\frac{1}{2}} \rho_{A_{1}B_{1}A_{2}}^{|b_{2}}\rho_{A_{1}B_{1}}^{-\frac{1}{2}}\Pi(\nu_{A_{1}B_{1}}) \\
            &= \Pi(\nu_{A_{1}B_{1}})\rho_{A_{1}B_{1}}^{-\frac{1}{2}} \rho_{A_{1}B_{1}A_{2}B_{2}} \rho_{A_{1}B_{1}}^{-\frac{1}{2}}\Pi(\nu_{A_{1}B_{1}}).
        \end{aligned}
    \end{equation*}
    Note that $\Pi(\nu_{A_{1}B_{1}}) = \Pi\Big( \Big(\rho_{A_{1}B_{1}}^{\frac{1}{2}} \sigma_{B_{1}}^{-\alpha'} \rho_{A_{1}B_{1}}^{\frac{1}{2}} \Big)^{\alpha} \Big) = \Pi(\rho_{A_{1}B_{1}})$, which implies 
    \begin{equation*}
        \Pi(\nu_{A_{1}B_{1}})\rho_{A_{1}B_{1}}^{-\frac{1}{2}} \rho_{A_{1}B_{1}A_{2}B_{2}} \rho_{A_{1}B_{1}}^{-\frac{1}{2}}\Pi(\nu_{A_{1}B_{1}}) = \rho_{A_{2}B_{2}|A_{1}B_{1}},     
    \end{equation*}
    completing the proof.
\end{proof}

\begin{claim}
    Let $\mcM : R \to A_{2}B_{2}$ be a quantum channel such that $B_{2}$ is classical and for all input states $\omega_{RR'} \in \mcD(RR')$ where $R'$ is an arbitrary register, the output state $\rho_{A_{2}B_{2}R'} = \mcM(\omega_{RR'})$ satisfies $\rho_{B_{2}R'} = \rho_{B_{2}} \otimes \omega_{R'}$. Let $\omega_{A_{1}B_{1}R} \in \mcD(A_{1}B_{1}R)$ be a state and $\rho_{A_{1}B_{1}A_{2}B_{2}} = \mcM(\omega_{A_{1}B_{1}R})$. Let $\nu_{A_{1}B_{1}A_{2}B_{2}} \in \mcD(A_{1}B_{1}A_{2}B_{2})$ be any state that satisfies 
    \begin{equation}
        \nu_{A_{2}B_{2}|A_{1}B_{1}} = \rho_{A_{2}B_{2}|A_{1}B_{1}}
    \end{equation}
    and $\mathrm{supp}(\nu_{A_{1}B_{1}}) = \mathrm{supp}(\rho_{A_{1}B_{1}})$. Then there exists a state $\omega'_{A_{1}B_{1}R} \in \mcD(A_{1}B_{1}R)$ that satisfies
    \begin{equation}
        \mcM(\omega'_{RA_{1}B_{1}}) = \nu_{A_{1}B_{1}A_{2}B_{2}}.
    \end{equation} \label{lem:chan_1}
\end{claim}
\begin{proof}
    Note that $\nu_{A_{2}B_{2}|A_{1}B_{1}} = \rho_{A_{2}B_{2}|A_{1}B_{1}}$ is equivalent to 
    \begin{equation*}
        \nu_{A_{1}B_{1}A_{2}B_{2}} = \nu_{A_{1}B_{1}}^{\frac{1}{2}} \, \rho_{A_{2}B_{2}|A_{1}B_{1}} \, \nu_{A_{1}B_{1}}^{\frac{1}{2}}.
    \end{equation*}
    Furthermore, the assumption on $\mcM$ implies 
    \begin{equation}
        \rho_{A_{1}B_{1}} = \tr_{A_{2}B_{2}}[\mcM(\omega_{A_{1}B_{1}R})] = \tr_{B_{2}}[\rho_{B_{2}} \otimes \omega_{A_{1}B_{1}}] = \omega_{A_{1}B_{1}}. \label{eq:eq_marg}
    \end{equation}
    Now, consider the following positive semi-definite operator:
    \begin{equation*}
        \omega_{RA_{1}B_{1}}' := \nu_{A_{1}B_{1}}^{\frac{1}{2}} \, \rho_{A_{1}B_{1}}^{-\frac{1}{2}} \, \omega_{A_{1}B_{1}R} \, \rho_{A_{1}B_{1}}^{-\frac{1}{2}} \, \nu_{A_{1}B_{1}}^{\frac{1}{2}}.
    \end{equation*}
    By combining \cref{eq:eq_marg} with the fact that $\mathrm{supp}(\nu_{A_{1}B_{1}}) = \mathrm{supp}(\rho_{A_{1}B_{1}})$, or equivalently $\Pi(\nu_{A_{1}B_{1}}) = \Pi(\rho_{A_{1}B_{1}})$, we find $\tr[\omega_{RA_{1}B_{1}}']=1$, i.e., $\omega_{RA_{1}B_{1}}'$ is a normalized quantum state. Then, by direct calculation,
    \begin{equation*}
        \begin{aligned}
            \mcM(\omega_{RA_{1}B_{1}}') &= \nu_{A_{1}B_{1}}^{\frac{1}{2}} \, \rho_{A_{1}B_{1}}^{-\frac{1}{2}} \, \mcM(\omega_{A_{1}B_{1}R}) \, \rho_{A_{1}B_{1}}^{-\frac{1}{2}} \, \nu_{A_{1}B_{1}}^{\frac{1}{2}}  \\
            &= \nu_{A_{1}B_{1}}^{\frac{1}{2}} \, \rho_{A_{1}B_{1}}^{-\frac{1}{2}} \, \rho_{A_{1}B_{1}A_{2}B_{2}} \, \rho_{A_{1}B_{1}}^{-\frac{1}{2}} \, \nu_{A_{1}B_{1}}^{\frac{1}{2}} \\
            &= \nu_{A_{1}B_{1}A_{2}B_{2}}
        \end{aligned}
    \end{equation*}
    as claimed. 
\end{proof}

Having established these facts, we are now ready to prove a chain rule for $H_{\alpha}^{\uparrow}$ in terms of the partially optimized conditional entropy for cq-states (\cref{def:partial_ent}).

\begin{lemma}
    Let $\rho_{A_{1}B_{1}A_{2}B_{2}} = \sum_{b_{2}}p_{B_{2}}(b_{2}) \ketbra{b_{2}}{b_{2}} \otimes \rho_{A_{1}B_{1}A_{2}}^{|b_{2}} \in \mcD(A_{1}B_{1}A_{2}B_{2})$ be a cq-state classical on $B_{2}$ such that $\rho_{A_{1}B_{1}B_{2}} = \rho_{A_{1}B_{1}} \otimes \rho_{B_{2}}$, and let $\alpha \in (0,\infty)$. Then
    \begin{equation}
        H_{\alpha}^{\uparrow}(A_{1}A_{2}|B_{1}B_{2})_{\rho} \geq H_{\alpha}^{\uparrow}(A_{1}|B_{1})_{\rho} + \inf_{\nu} H_{\alpha}(A_{2}|B_{2}^{\uparrow}(A_{1}B_{1})^{\downarrow})_{\nu}, \label{eq:new_chain1}
    \end{equation}
    where the infimum ranges over all cq-states classical on $B_{2}$ such that $\nu_{A_{2}B_{2}|A_{1}B_{1}} = \rho_{A_{2}B_{2}|A_{1}B_{2}}$ and $\mathrm{supp}(\nu_{A_{1}B_{1}}) = \mathrm{supp}(\rho_{A_{1}B_{1}})$. \label{lem:newchain}
\end{lemma}
\begin{proof}
    We begin by noting that the condition $\rho_{A_{1}B_{1}B_{2}} = \rho_{A_{1}B_{1}} \otimes \rho_{B_{2}}$ implies
    \begin{equation*}
        \sum_{b_{2}}p_{B_{2}}(b_{2})\ketbra{b_{2}}{b_{2}}_{B_{2}} \otimes \rho_{A_{1}B_{1}}^{|b_{2}} = \sum_{b_{2}}p_{B_{2}}(b_{2})\ketbra{b_{2}}{b_{2}}_{B_{2}} \otimes \rho_{A_{1}B_{1}}
    \end{equation*}
    where $\rho_{A_{1}B_{1}}^{|b_{2}} = \tr_{A_{2}}[\rho_{A_{1}B_{1}A_{2}}^{|b_{2}}]$ and $\rho_{A_{1}B_{1}} = \sum_{b_{2}}p_{B_{2}}(b_{2})\rho_{A_{1}B_{1}}^{|b_{2}}$. Therefore, $\rho_{A_{1}B_{1}}^{|b_{2}} = \rho_{A_{1}B_{1}}$ for all $b_{2}$. Using the fact that $B_{2}$ is classical (see \cref{eq:c_ent}) we can write 
    \begin{equation}
        H_{\alpha}^{\uparrow}(A_{1}A_{2}|B_{1}B_{2})_{\rho} = \frac{\alpha}{1-\alpha} \log \Bigg(\sum_{b_{2}}p_{B_{2}}(b_{2})\, 2^{\frac{1-\alpha}{\alpha} H_{\alpha}^{\uparrow}(A_{1}A_{2}|B_{1})_{\rho^{|b_{2}}} }\Bigg). \label{eq:sum_ent}
    \end{equation}
    Now applying \cref{lem:DFR_chain} to each entropy in the summation,
    \begin{equation}
    \begin{aligned}
        H_{\alpha}^{\uparrow}(A_{1}A_{2}|B_{1})_{\rho^{|b_{2}}} &= \sup_{\sigma \in \mcD(B_{1})} - D_{\alpha}(\rho^{|b_{2}}_{A_{1}A_{2}B_{1}} \| \id_{A_{1}A_{2}} \otimes \sigma_{B_{1}}) \\ &= \sup_{\sigma \in \mcD(B_{1})} \Big(- D_{\alpha}(\rho_{A_{1}B_{1}} \| \id_{A_{1}} \otimes \sigma_{B_{1}}) + H_{\alpha}^{\downarrow}(A_{2}|A_{1}B_{1})_{\nu^{|b_{2}}} \Big), \label{eq:sup_chain}
    \end{aligned}
    \end{equation}
    where we note the supremum can be restricted to states $\sigma_{B_{1}}$ that satisfy $\Pi(\sigma_{B_{1}}) \geq \Pi(\rho_{B_{1}})$~\cite{ML_2013}, and according to \eqref{eq:nu_state}
    \begin{equation*}
        \nu_{A_{1}B_{1}} = \frac{\Big(\rho_{A_{1}B_{1}}^{\frac{1}{2}} \sigma_{B_{1}}^{-\alpha'} \rho_{A_{1}B_{1}}^{\frac{1}{2}} \Big)^{\alpha}}{\tr \Big[\Big(\rho_{A_{1}B_{1}}^{\frac{1}{2}} \sigma_{B_{1}}^{-\alpha'} \rho_{A_{1}B_{1}}^{\frac{1}{2}}\Big)^{\alpha}\Big]}, \ \ \ \text{and} \ \ \ \nu_{A_{1}A_{2}B_{1}}^{|b_{2}} = \nu_{A_{1}B_{1}}^{\frac{1}{2}} \,  \rho_{A_{2}|A_{1}B_{1}}^{|b_{2}} \, \nu_{A_{1}B_{1}}^{\frac{1}{2}}.
    \end{equation*}
    Let $\tilde{\sigma}_{B_{1}} \in \mcD(B_{1})$ be a state that satisfies
    \begin{equation*}
         - D_{\alpha}(\rho_{A_{1}B_{1}} \| \id_{A_{1}} \otimes \tilde{\sigma}_{B_{1}}) = H_{\alpha}^{\uparrow}(A_{1}|B_{1})_{\rho}.
    \end{equation*}
    Note that such a state is guaranteed to exist, and satisfies $\Pi(\tilde{\sigma}_{B_1}) \geq \Pi(\rho_{B_1})$~\cite{ML_2013}. Choosing this state in \eqref{eq:sup_chain} results in the lower bound
    \begin{equation*}
        H_{\alpha}^{\uparrow}(A_{1}A_{2}|B_{1})_{\rho^{|b_{2}}} \geq H_{\alpha}^{\uparrow}(A_{1}|B_{1})_{\rho} + H_{\alpha}^{\downarrow}(A_{2}|A_{1}B_{1})_{\nu^{|b_{2}}}, 
    \end{equation*}
    where 
    \begin{equation*}
        \nu_{A_{1}B_{1}} = \frac{\Big(\rho_{A_{1}B_{1}}^{\frac{1}{2}} \tilde{\sigma}_{B_{1}}^{-\alpha'} \rho_{A_{1}B_{1}}^{\frac{1}{2}} \Big)^{\alpha}}{\tr \Big[\Big(\rho_{A_{1}B_{1}}^{\frac{1}{2}} \tilde{\sigma}_{B_{1}}^{-\alpha'} \rho_{A_{1}B_{1}}^{\frac{1}{2}}\Big)^{\alpha}\Big]}, \ \ \ \text{and} \ \ \ \nu_{A_{1}A_{2}B_{1}}^{|b_{2}} = \nu_{A_{1}B_{1}}^{\frac{1}{2}} \,  \rho_{A_{2}|A_{1}B_{1}}^{|b_{2}} \, \nu_{A_{1}B_{1}}^{\frac{1}{2}}.
    \end{equation*}
    Substituting this lower bound into \eqref{eq:sum_ent} and noting that $H_{\alpha}^{\uparrow}(A_{1}|B_{1})_{\rho}$ is independent of $b_{2}$, we obtain
    \begin{equation}
         H_{\alpha}^{\uparrow}(A_{1}A_{2}|B_{1}B_{2})_{\rho} \geq  H_{\alpha}^{\uparrow}(A_{1}|B_{1})_{\rho} + \frac{\alpha}{1-\alpha} \log\Bigg( \sum_{b_{2}}p_{B_{2}}(b_{2})\, 2^{\frac{1-\alpha}{\alpha} H_{\alpha}^{\downarrow}(A_{2}|A_{1}B_{1})_{\nu^{|b_{2}}}}\Bigg). \label{eq:pen_eq}
    \end{equation}
    Next, we define the following cq-state classical on $B_{2}$:
    \begin{equation*}
        \nu_{A_{1}B_{1}A_{2}B_{2}} := \sum_{b_{2}}p_{B_{2}}(b_{2})\ketbra{b_{2}}{b_{2}}_{B_{2}} \otimes \nu_{A_{1}A_{2}B_{1}}^{|b_{2}},
    \end{equation*}
    According to \cref{def:partial_ent}, 
    \begin{equation*}
        \frac{\alpha}{1-\alpha} \log \Bigg(\sum_{b_{2}}p_{B_{2}}(b_{2})\, 2^{\frac{1-\alpha}{\alpha} H_{\alpha}^{\downarrow}(A_{2}|A_{1}B_{1})_{\nu^{|b_{2}}} }\Bigg) = H_{\alpha}(A_{2}|B_{2}^{\uparrow}(A_{1}B_{1})^{\downarrow})_{\nu}.
    \end{equation*}
    Furthermore, the state $\nu_{A_{1}B_{1}A_{2}B_{2}}$ is exactly of the form \eqref{eq:nu_state} from \cref{lem:nu_prop}, with $\tilde{\sigma}_{B_{1}}$ and $\rho_{A_{1}B_{1}A_{2}B_{2}}$ satisfying the conditions of the lemma, which implies
    \begin{equation*}
        \nu_{A_{2}B_{2}|A_{1}B_{1}} = \rho_{A_{2}B_{2}|A_{1}B_{2}}.
    \end{equation*}
    We can therefore take the infimum of \eqref{eq:pen_eq} over all cq-states $\nu$ that satisfy this condition, proving the claim. 
\end{proof}

By combining \cref{lem:newchain} with \cref{lem:chan_1}, we can lower bound the infimum in \cref{eq:new_chain1} by an infimum over all extensions of inputs to the quantum channel $\mcM$. Then using \cref{lem:pureOpt} (proven separately below) we can restrict this infimum to pure states.

\vspace{0.2cm}

\noindent \textbf{Corollary 4.4} (Chain rule)\textbf{.}  
    \textit{Let $\mcM : R \to A_{2}B_{2}$ be a quantum channel such that $B_{2}$ is classical and for all input states $\omega_{RR'} \in \mcD(RR')$ where $R'$ is an arbitrary register, the output state $\rho_{A_{2}B_{2}R'} = \mcM(\omega_{RR'})$ satisfies $\rho_{B_{2}R'} = \rho_{B_{2}} \otimes \omega_{R'}$. Let $\omega_{A_{1}B_{1}R} \in \mcD(A_{1}B_{1}R)$ be a quantum state. Then for any $\alpha \in (1,\infty)$,
    \begin{equation}
        H_{\alpha}^{\uparrow}(A_{1}A_{2}|B_{1}B_{2})_{\mcM(\omega)} \geq H_{\alpha}^{\uparrow}(A_{1}|B_{1})_{\omega} + \inf_{\omega'} H_{\alpha}(A_{2}|B_{2}^{\uparrow}\tilde{E}^{\downarrow})_{\mcM(\omega')}, \label{eq:app_chain_new}
    \end{equation}
    where the infimum is taken over all pure states and the system $\tilde{E}$ has the same dimension as $R$.} 

\begin{proof}
   Consider the state $\rho_{A_{1}B_{1}A_{2}B_{2}} = \mcM(\omega_{A_{1}B_{1}R})$. Due to the assumption on $\mcM$, 
   \begin{equation*}
       \rho_{A_{1}B_{1}B_{2}} = \omega_{A_{1}B_{1}} \otimes \rho_{B_{2}}. 
   \end{equation*}
   We therefore see $\rho_{A_{1}B_{1}A_{2}B_{2}}$ satisfies the conditions needed to apply \cref{lem:newchain}:
    \begin{equation*}
        H_{\alpha}^{\uparrow}(A_{1}A_{2}|B_{1}B_{2})_{\rho} \geq H_{\alpha}^{\uparrow}(A_{1}|B_{1})_{\rho} + \inf_{\nu} H_{\alpha}(A_{2}|B_{2}^{\uparrow}(A_{1}B_{1})^{\downarrow})_{\nu},
    \end{equation*}
    where the infimum ranges over all cq-states $\nu_{A_{1}B_{1}A_{2}B_{2}}$ classical on $B$ such that $\nu_{A_{2}B_{2}|A_{1}B_{1}} = \rho_{A_{2}B_{2}|A_{1}B_{2}}$ and $\Pi(\nu_{A_{1}B_{1}}) = \Pi(\rho_{A_{1}B_{1}}) = \Pi(\omega_{A_{1}B_{1}})$. Every feasible point $\nu_{A_{1}B_{1}A_{2}B_{2}}$ satisfies the conditions of \cref{lem:chan_1}, i.e., can be written as $\mcM(\omega'_{A_{1}B_{1}R})$ for some state $\omega'_{A_{1}B_{1}R}$. This allows us to lower bound the infimum:
    \begin{equation*}
         \inf_{\nu} H_{\alpha}(A_{2}|A_{1}^{\downarrow}B_{1}^{\downarrow}B_{2}^{\uparrow})_{\nu} \geq \inf_{\omega' \in \mcD(A_{1}B_{1}R)} H_{\alpha}(A_{2}|B_{2}^{\uparrow}(A_{1}B_{1})^{\downarrow})_{\mcM(\omega')} \geq \inf_{E} \inf_{\omega' \in \mcD(RE)} H_{\alpha}(A_{2}|B_{2}^{\uparrow}E^{\downarrow})_{\mcM(\omega')}.
    \end{equation*}
    By applying \cref{lem:pureOpt} we arrive at the desired statement.
\end{proof}

\begin{claim}
    Let $\mcM : R \to AC$ be a quantum channel. Then for any $\alpha \in (1,\infty)$ the infimum 
    \begin{equation*}
        \inf_{C}\inf_{\omega \in \mcD(RC)} H_{\alpha}(A|B^{\uparrow}C^{\downarrow})_{\mcM(\omega)} 
    \end{equation*}
    is achieved by a system $\tilde{C}$ with dimension equal to $R$ and a pure state on $R\tilde{C}$. \label{lem:pureOpt}
\end{claim}
\begin{proof}
    Let $(C,\omega_{RC})$ be any feasible point, and let $\omega_{RC\tilde{C}} \in \mcD(RCC')$ be any purification of $\omega_{RC}$ for a purifying register $C'$. By the subadditivity property $(iii)$ in \cref{lem:ent_prop},
    \begin{equation*}
        H_{\alpha}(A|B^{\uparrow}C^{\downarrow})_{\mcM(\omega)} \geq H_{\alpha}(A|B^{\uparrow}(CC')^{\downarrow})_{\mcM(\omega)},
    \end{equation*}
    hence for every feasible point $(C,\omega_{RC})$, there exists another feasible point $(CC',\omega_{RCC'})$ with non-increasing entropy for which the state is pure. Let $\tilde{C}$ be a register that has the same dimension as $R$, and let $\tilde{\omega}_{R\tilde{C}}$ be a purification of $\omega_{R} = \tr_{CC'}[\omega_{RCC'}]$ for the purifying system $\tilde{C}$. All purifications are equal up to an isometry on the purifying system, hence there exists an isometry $V$ from $\tilde{C}$ to $CC'$ such that $V \tilde{\omega}_{R\tilde{C}} V^{\dagger} =  \omega_{RCC'}$. Using the isometric invariance property $(ii)$ in \cref{lem:ent_prop} and the fact that $V$ commutes with $\mcM$,
    \begin{equation*}
         H_{\alpha}(A|B^{\uparrow}\tilde{C}^{\downarrow})_{\mcM(\tilde{\omega})} = H_{\alpha}(A|B^{\uparrow}(CC')^{\downarrow})_{V \mcM(\tilde{\omega}) V^{\dagger}} = H_{\alpha}(A|B^{\uparrow}(CC')^{\downarrow})_{ \mcM(V\tilde{\omega}V^{\dagger})} = H_{\alpha}(A|B^{\uparrow}(CC')^{\downarrow})_{ \mcM(\omega)}.
    \end{equation*}
    Therefore, for every feasible point $(C,\omega_{RC})$ there exists another feasible point $(\tilde{C},\tilde{\omega}_{R\tilde{C}})$ with non-increasing entropy where the state is pure and $\tilde{C}$ has the same dimension as $R$, completing the proof.   
\end{proof}

\section{Proofs for \cref{sec:REAT}: entropy accumulation} \label{app:REAT}

In this appendix, we provide a detailed proof of our new entropy accumulation theorem. Specifically, we consider the ``R\'enyi EAT'' (REAT) as presented by Arqand, Hahn and Tan~\cite{arqand2024}, and follow their proof structure closely. If desired, the REAT can be converted into a statement akin to the original EAT derived by Dupuis Fawzi and Renner~\cite{DFR}, relating the $n$-round smooth min-entropy to $n$ times the single round von Neumann entropy. However, by avoiding this conversion and working with R\'enyi entropies directly, the REAT has been shown to improve the finite size performance of DI protocols in particular~\cite{arqand2024,hahn2024,hahn2025}. 

\subsection{Definitions}

We begin with a modified definition of the ``$f$-weighted'' R\'enyi entropy~\cite{VHB25,arqand2024} needed for our work. 

\begin{definition}[$f$-weighted R\'enyi entropy]
    Let $C$ be a classical system that takes values in a finite alphabet $\mcC$, $\rho_{ABC} = \sum_{c\in \mcC}p_{C}(c)\ketbra{c}{c}_{C} \otimes \rho_{AB}^{|c}$ be a cq-state and $\sigma_{B} \in \mcD(B)$ be a state. Let $f \in \mbR^{|\mcC|}$, $f = [f_{c}]_{c \in \mcC}$ denote a real vector and $\alpha \in (1,\infty)$. Then the $f$-weighted R\'enyi entropy is defined as
    \begin{equation}
        H_{\alpha}^{f}\big(AC|B\big)_{\rho|\sigma} := \frac{1}{1 - \alpha} \log \Bigg(\sum_{c\in \mcC} p_{C}(c)^{\alpha} \, 2^{(\alpha-1)\big[ f_{c} + D_{\alpha}(\rho_{AB}^{|c}\|\id_{A} \otimes \sigma_{B}) \big]} \Bigg)
    \end{equation}
    if $\text{supp}(\rho_{B}) \subseteq \text{supp}(\sigma_{B})$ and $+\infty$ otherwise.\label{def:f_div} 
\end{definition}
\noindent If $\sigma_{B} = \tr_{AC}[\rho_{ABC}]$, \cref{def:f_div} reduces to a special case of~\cite[Definition 4.1]{arqand2024}. The notation $\rho|\sigma$ is chosen to reflect similar definitions in the literature, $H_{\alpha}(A|B)_{\rho|\sigma} := -D_{\alpha}(\rho_{AB}\|\id_{A}\otimes \sigma_{B})$ if $\text{supp}(\rho_{B}) \subseteq \text{supp}(\sigma_{B})$ and $+\infty$ otherwise (see, e.g.,~\cite[Definition 3]{TCR}), which is what \cref{def:f_div} reduces to if $f = 0$. The $f$-weighted R\'enyi entropy satisfies a number of important properties, many of which have been proving in previous works~\cite{VHB25,arqand2024}. We prove analogous versions of these properties for our modified definition in \cref{sec:f_prop}. 

We will also need to define a ``read-and-prepare'' channel, and state some useful properties.

\begin{definition}[read-and-prepare channel]
    Let $C$ denote a classical register that takes values in a finite alphabet $\mcC$, and $\rho_{CQ} = \sum_{c} p_{C}(c) \ketbra{c}{c}_{C} \otimes \rho_{Q}^{|c}$ be a cq-state. Then a channel $\mcR : C \to CD$ is a read-and-prepare channel if it is of the form
    \begin{equation*}
        \mcR(\rho_{CQ}) = \sum_{c \in \mcC} p_{C}(c) \ketbra{c}{c}_{C} \otimes \rho_{Q}^{|c} \otimes \tau(c)_{D},
    \end{equation*}
    where $\{\tau(c)_{D}\}_{c \in \mcC} \subset \mcD(D)$.  \label{def:randp}
\end{definition}

\begin{lemma}[\cite{arqand2024} Lemmas 4.1 and 4.2]
    Let $\mcC$ be a finite alphabet and $f \in \mbR^{|\mcC|}$ be any real vector, and let $M > 0$ be a constant such that $M - f_{c} > M/2 > 0$ for all $c \in \mcC$. Then there always exists a read-and-prepare channel $\mcR:C \to CD$ according to \cref{def:randp} satisfying
    \begin{equation*}
        H_{\alpha}(D)_{\tau(c)} \in \big[ M - f_{c} , M - f_{c} + 2^{-\frac{M}{2}}\log(e) \big]
    \end{equation*}
    for all $\alpha \in [0,\infty]$ and all $c \in \mcC$. Furthermore, for any fixed $\alpha \in [0,\infty]$, there always exists a read-and-prepare channel $\mcR:C \to CD$ according to \cref{def:randp} satisfying
    \begin{equation*}
        H_{\alpha}(D)_{\tau(c)} = M - f_{c} 
    \end{equation*}
    for all $c \in \mcC$. Moreover, every state $\tau(c)_{D}$ can be classical. \label{lem:prep}
\end{lemma}
\noindent Note that a read-and-prepare channel $\mcR$ according to \cref{def:randp} does not disturb the state, i.e., $\tr_{D} \circ \mcR = \mcI$, where $\mcI$ is the identity channel. 

\subsection{Reduction to single round quantities}
In the first step of the proof, we reduce the $n$ round R\'enyi entropy of the output state $\tr_{R_{n}} \circ \mcM_{n} \circ \cdots \circ \mcM_{1}[\omega_{R_{0}E}]$, conditioned on an event $\Omega$ on $C^n$, to a sum of entropy contributions from each round. This step relies on an $n$ fold application of the chain rule in \cref{lem:newchain}. Importantly, we take into account a trade-off function $f$ (that can be chosen later) by applying a read-and-prepare channel to the output state. For each round $i\in \{1,...,n\}$, this channel reads the value in the classical register $C_i=c_i$, and encodes the number $f_{c_i}$ as the entropy of a state prepared in a new register $D_{i}$. This is the main departure from the EAT proof sketched in \cref{sec:markov}, where testing was not taken into account. 

The proof of \cref{lem:single} borrows ideas from~\cite[Proposition 4.5]{DFR}  and~\cite[Corollary 5.1]{arqand2024}. Furthermore, a direct comparison can be made to~\cite[Corollary 5.1]{arqand2024}, where the main difference is the additional supremum over $q_{B_{i}} \in \Delta(B_{i})$ in \cref{lem:single} inherited from our new chain rule. 

\begin{lemma}[Reduction to single round quantities]
     Let $\omega_{R_{0}E} \in \mcD(R_{0}E)$ be a state, $\{\mcM_{i}\}_{i=1}^{n}$ be a sequence of DI channels according to \cref{def:DIchan} and $\rho_{A^{n}C^{n}B^{n}E} = [\tr_{R_{n}} \circ \mcM_{n} \circ \cdots \mcM_{1}](\omega_{R_{0}E})$. Let $\Omega \subset \mcC^{n}$ be an event on $C^{n}$, $f = [f_{c}]_{c \in \mcC} \in \mbR^{|\mcC|}$ be any real vector and $\Delta_{\Omega} \subset \Delta(C)$ be any compact convex subset such that $\mathsf{freq}(c^{n}) \in \Delta_{\Omega}$ for all $c^{n} \in \Omega$. Then for all $\alpha \in (1,\infty)$,
     \begin{multline}
         H_{\alpha}^{\uparrow}(A^{n}C^{n}|B^{n}E)_{\rho^{|\Omega}} \geq  \sum_{i=1}^{n} \inf_{\omega \in \mcD(R_{i-1}\tilde{E}_{i-1})} \inf_{v_{C} \in \Delta_{\Omega}} \sup_{q_{B_{i}} \in \Delta(B_{i})} \Bigg( v_{C} \cdot f + H_{\alpha}^{f}(A_{i}C_{i}|B_{i}\tilde{E}_{i-1})_{\mcM_{i}(\omega)|\sigma} \Bigg) \\ - \frac{\alpha}{\alpha-1} \log \frac{1}{p_{\Omega}}, \label{eq:single_round}
     \end{multline}
     where $\sigma_{B_{i}\tilde{E}_{i-1}} = \sum_{b}q_{B_{i}}(b)\ketbra{b}{b}_{B_{i}} \otimes \rho_{\tilde{E}_{i-1}}^{|b}$ and $\{\rho_{\tilde{E}_{i-1}}^{|b}\}_{b}$ is the set of conditional states defined via  $\tr_{A_{i}C_{i}}[\mcM_{i}(\omega_{R_{i-1}\tilde{E}_{i-1}})] = \sum_{b}p_{B_{i}}(b)\ketbra{b}{b}_{B_{i}} \otimes \rho_{\tilde{E}_{i-1}}^{|b}$.
     \label{lem:single}  
\end{lemma}
\begin{proof}
    Using the formula for the R\'enyi divergence between cq-states (cf. \cref{eq:cq_div}), we can write 
    \begin{equation}
        H_{\alpha}^{\uparrow}(A^{n}C^{n}|B^{n}E)_{\rho^{|\Omega}} = \sup_{\sigma \in \mcD(B^{n}E)} \frac{1}{1-\alpha} \log \Bigg(\sum_{C^{n} \in \Omega} \Big(\frac{p_{C^{n}}(c^{n})}{p_{\Omega}} \Big)^{\alpha} 2^{(\alpha - 1)D_{\alpha}\big( \rho_{A^{n}B^{n}E}^{|c^{n}} \| \id_{A^{n}} \otimes \sigma_{B^{n}E}\big)} \Bigg). \label{eq:decEnt}
    \end{equation}
    Consider any vector $f \in \mbR^{|\mcC|}$, and let $M > 0$ be such that $M - f_{c} > 0$ for all $c \in \mcC$. For every $c \in \mcC$, let $\tau(c)_{D}$ be any state such that $H_{\alpha}(D)_{\tau(c)} = M - f_{c}$. Such states are guaranteed to exist according to \cref{lem:prep}. Let $\mcR_{i}:C_{i} \to C_{i}D_{i}$ be a read-and-prepare channel for the states $\{\tau(c)\}_{c\in \mcC}$ according to \cref{def:randp}, and let $\bar{\rho}^{|\Omega} = [\mcR_{n} \circ \cdots \circ \mcR_{1}](\rho^{|\Omega})$. Using the additivity of the sandwiched R\'enyi divergence under tensor products~\cite{Tomamichel_2016} and the fact that $\bar{\rho}^{|c^{n}}_{A^{n}D^{n}B^{n}E} = \rho^{|c^{n}}_{A^{n}B^{n}E} \otimes \tau(c^{n})_{D^{n}}$, we can break up every exponent in the summation of \eqref{eq:decEnt},
    \begin{equation}
    \begin{aligned}
        - D_{\alpha}\big(\bar{\rho}^{|c^{n}}_{A^{n}D^{n}B^{n}E} \, \| \, \id_{A^{n}D^{n}} \otimes \sigma_{B^{n}E}\big) &= -D_{\alpha}\big(\rho^{|c^{n}}_{A^{n}B^{n}E} \| \id_{A^{n}} \otimes \sigma_{B^{n}E}\big) - D_{\alpha}(\tau(c^{n}) \| \id_{D^{n}}) \\
        &= -D_{\alpha}\big(\rho^{|c^{n}}_{A^{n}B^{n}E} \| \id_{A^{n}} \otimes \sigma_{B^{n}E}\big) + H_{\alpha}(D^{n})_{\tau(c^{n})}. \label{eq:div_bnd1}
    \end{aligned}
    \end{equation}
    Furthermore,
    \begin{equation*}
        \begin{aligned}
            H_{\alpha}(D_{1}^{n})_{\tau(c_{1}^{n})}
            &= \sum_{i=1}^{n}H_{\alpha}(D_{i})_{\tau(c_{i})} \\
            &= \sum_{i=1}^{n}\big( M  - f_{c_{i}}\big) \\
            &= n\,M - n\sum_{c}\frac{\{i \ : \ c_{i}=x\}}{n} f_{c} \\
            &= n\,M  - n\sum_{c}\mathsf{freq}(c_{1}^{n})(c) f_{c} .
        \end{aligned}
    \end{equation*}
    Inserting this bound into \eqref{eq:div_bnd1} we arrive at
    \begin{equation}
         -D_{\alpha}\big(\rho^{|c^{n}}_{A^{n}B^{n}E} \| \id_{A^{n}} \otimes \sigma_{B^{n}E}\big) = - D_{\alpha}\big(\rho^{|c^{n}}_{A^{n}D^{n}B^{n}E} \, \| \, \id_{A^{n}D^{n}} \otimes \sigma_{B^{n}E}\big) - n \, M  + n\sum_{c}\mathsf{freq}(c_{1}^{n})(c) f_{x} \label{eq:lb4}
    \end{equation}
    Thus, by combining \eqref{eq:lb4} and \eqref{eq:decEnt} we obtain the following inequalities: 
    \begin{equation}
    \begin{aligned}
        H_{\alpha}^{\uparrow}(A^{n}C^{n}|B^{n}E)_{\rho^{|\Omega}} &= \sup_{\sigma \in \mcD(B^{n}E)} \frac{1}{1-\alpha} \log \Bigg(\sum_{c^{n} \in \Omega} \Big(\frac{p_{C^{n}}(c^{n})}{p_{\Omega}} \Big)^{\alpha} \\ & \hspace{1.5cm} \cdot 2^{(1 - \alpha)\big[ -D_{\alpha}\big(\bar{\rho}^{|c^{n}}_{A^{n}D^{n}B^{n}E} \, \| \, \id_{A^{n}D^{n}} \otimes \sigma_{B^{n}E}\big) - nM + n\sum_{c}\mathsf{freq}(c^{n})(c) f_{c} \big]}\Bigg) \\
        &\geq - n \, M  + n \min_{c^{n} \in \Omega} \sum_{c\in \mcC}\mathsf{freq}(c^{n})(c) \, f_{c} \\ 
        & \hspace{1.5cm} + \sup_{\sigma \in \mcD(B^{n}E)} \frac{1}{1-\alpha} \log \Bigg(\sum_{c^{n} \in \Omega} \Big(\frac{p_{C^{n}}(c^{n})}{p_{\Omega}} \Big)^{\alpha} 2^{(\alpha - 1)D_{\alpha}\big(\bar{\rho}^{|c^{n}}_{A^{n}D^{n}B^{n}E} \, \| \, \id_{A^{n}D^{n}} \otimes \sigma_{B^{n}E}\big)} \Bigg) \\
        &= - n \, M + n\min_{c^{n} \in \Omega} \sum_{c\in \mcC}\mathsf{freq}(c^{n})(c) \, f_{c} + H_{\alpha}^{\uparrow}(A^{n}C^{n}D^{n}|B^{n}E)_{\bar{\rho}^{|\Omega}} \\
        & \geq - n \, M  + n\inf_{v_{C} \in \Delta_{\Omega}} v_{C} \cdot f + H_{\alpha}^{\uparrow}(A^{n}C^{n}D^{n}|B^{n}E)_{\bar{\rho}^{|\Omega}}.
    \end{aligned} \label{eq:lb1}
    \end{equation}
    For the second line, we took the minimum over $c^{n} \in \Omega$ of the exponent to separate it from the sum, and for the final line we used the fact that for every $c^{n}$, the associated frequency distribution lies in $\mathcal{P}_{\Omega}$. We next apply \cite[Lemma B.5]{DFR} to the state $\bar{\rho} = (1-p_{\Omega})\bar{\rho}^{|\Omega_{\perp}} + p_{\Omega} \bar{\rho}^{|\Omega}$, where $\Omega_{\perp} \cup \Omega = \mcC^{n}$, to lower bound
    \begin{equation}
    \begin{aligned}
        H_{\alpha}^{\uparrow}(A^{n}C^{n}D^{n}|B^{n}E)_{\bar{\rho}^{|\Omega}} &\geq H_{\alpha}^{\uparrow}(A^{n}C^{n}D^{n}|B^{n}E)_{\bar{\rho}} - \frac{\alpha}{\alpha-1} \log \frac{1}{p_{\Omega}}.
    \end{aligned} \label{eq:no_cond}
    \end{equation}

    We are now ready to apply the chain rule in \cref{cor:new_chain_1}. For $i$ decreasing from $n$ to $2$ consider the following substitutions
    \begin{equation*}
        \begin{gathered}
            \mcM \mapsto \mcR_{i} \circ \mcM_{i}, \ 
            R \mapsto R_{i-1}, \ 
            A_{1} \mapsto A^{i-1}C^{i-1}D^{i-1}, \ 
            B_{1} \mapsto B^{i-1}E, \ 
            A_{2} \mapsto A_{i}C_{i}D_{i}, \ 
            B_{2} \mapsto B_{i} \\
            \omega \mapsto [ \mcR_{i-1} \circ \cdots \circ \mcR_{1} \circ \mcM_{i-1} \circ \cdots \circ \mcM_{1}](\omega_{R_{0}E}) =: \bar{\rho}^{(i-1)}_{A^{i-1}C^{i-1}B^{i-1}R_{i-1}E}.
        \end{gathered}
    \end{equation*}
    With this substitution, \cref{eq:app_chain_new} reads
    \begin{multline*}
        H_{\alpha}^{\uparrow}(A^{i}C^{i}D^{i}|B^{i}E)_{[\mcR_{i} \circ \mcM_{i}](\bar{\rho}^{(i-1)})} \geq H_{\alpha}^{\uparrow}(A^{i-1}C^{i-1}D^{i-1}|B^{i-1}E)_{\bar{\rho}^{(i-1)}} \\+ \inf_{\omega\in \mcD(R_{i-1}\tilde{E}_{i-1})} H_{\alpha}(A_{i}C_{i}D_{i}|B_{i}^{\uparrow}\tilde{E}^{\downarrow}_{i-1})_{[\mcR_{i} \circ \mcM_{i}](\omega')}, \label{eq:applychain}
    \end{multline*}
    where the infimum is taken over all pure states $\omega'_{R_{i-1}\tilde{E}_{i-1}}$ with $\tilde{E}_{i-1}$ and $R_{i-1}$ having the same dimension. For the above relation to hold, we we need to verify that the channels $\bar{\mcM}_{i} = \mcR_{i} \circ \mcM_{i}$ form a set of DI channels according to \cref{def:DIchan}, namely, that for every $i \in \{1,...,n\}$ and any input state $\omega_{R_{i-1}R'}$, the output state $\rho_{A_{i}C_{i}D_{i}B_{i}R'} = \bar{\mcM}_{i}(\omega_{R_{i-1}R'})$ satisfies $\rho_{B_{i}R'} = \rho_{B_{i}} \otimes \omega_{R'}$. This can be verified from the fact that $\tr_{D_{i}} \circ \mcR_{i} = \mcI$ and $\mcM_{i}$ form a sequence of DI channels by assumption:
    \begin{equation*}
        \rho_{B_{i}R'} = [\tr_{A_{i}C_{i}D_{i}} \circ \mcR_{i} \circ \mcM_{i}](\omega_{RR'}) =  [\tr_{A_{i}C_{i}}  \circ \mcM_{i}](\omega_{RR'}) = \rho_{B_{i}} \otimes \omega_{R'}.
    \end{equation*}
    We can therefore apply the chain rule $n-1$ times to the right hand side of \eqref{eq:no_cond} to obtain
    \begin{equation*}
        H_{\alpha}^{\uparrow}(A^{n}C^{n}D^{n}|B^{n}E)_{\bar{\rho}} \geq \sum_{i=1}^{n} \inf_{\omega \in \mcD(R_{i-1}\tilde{E}_{i-1})} H_{\alpha}(A_{i}C_{i}D_{i}|B_{i}^{\uparrow}\tilde{E}^{\downarrow}_{i-1})_{[\mcR_{i} \circ \mcM_{i}](\omega)}.
    \end{equation*}
    Inserting this back into \eqref{eq:lb1}, we arrive at the lower bound
    \begin{multline}
        H_{\alpha}^{\uparrow}(A^{n}C^{n}|B^{n}E)_{\rho^{|\Omega}} \geq - n  \, M + n\inf_{v_{C} \in \Delta_{\Omega}} v_{C} \cdot f \\ + \sum_{i=1}^{n} \inf_{\omega \in \mcD(R_{i-1}\tilde{E}_{i-1})} H_{\alpha}(A_{i}C_{i}D_{i}|B_{i}^{\uparrow}\tilde{E}^{\downarrow}_{i-1})_{[\mcR_{i} \circ \mcM_{i}](\omega)} - \frac{\alpha}{\alpha-1} \log \frac{1}{p_{\Omega}}. \label{eq:lb2}
    \end{multline}
    
    We now turn our attention to removing the registers $D_{i}$. Using \cref{lem:var_ent}, each entropy in the summation admits a variational expression  
    \begin{equation}
    \begin{aligned}
        H_{\alpha}(A_{i}C_{i}D_{i}|B_{i}^{\uparrow}\tilde{E}^{\downarrow}_{i-1})_{[\mcR_{i} \circ \mcM_{i}](\omega)} = \sup_{q_{B_{i}} \in \Delta(B_{i})} -D_{\alpha}\Big( [\mcR_{i} \circ \mcM_{i}](\omega_{R_{i-1}\tilde{E}_{i-1}}) \Big \| \id_{A_{i}C_{i}D_{i}} \otimes \sigma_{B_{i}\tilde{E}} \Big)
    \end{aligned} \label{eq:varDiv}
    \end{equation}
    where the states inside the divergence are of the form 
    \begin{equation*}
         \begin{aligned}
             [\mcR_{i} \circ \mcM_{i}](\omega_{R_{i-1}\tilde{E}_{i-1}}) &= \sum_{b} p_{B_{i}}(b) \, \ketbra{b}{b}_{B_{i}} \otimes \mcR_{i}[\rho_{A_{i}C_{i}\tilde{E}_{i-1}}^{|b}],  \ \ \ \text{and} \\
             \sigma_{B_{i}\tilde{E}} &= \sum_{b}q_{B_{i}}(b) \, \ketbra{b}{b}_{B_{i}} \otimes \rho_{\tilde{E}_{i-1}}^{|b}, 
         \end{aligned}
     \end{equation*}
     with $\rho_{\tilde{E}_{i-1}}^{|b} = [\tr_{A_{i}C_{i}D_{i}} \circ \mcR_{i}](\rho_{A_{i}C_{i}\tilde{E}_{i-1}}^{|b}) = \tr_{A_{i}C_{i}}[\rho_{A_{i}C_{i}\tilde{E}_{i-1}}^{|b}]$. Since the system $C_{i}$ is classical, we can also write
     \begin{equation*}
         [\mcR_{i} \circ \mcM_{i}](\omega_{R_{i-1}\tilde{E}_{i-1}}) = \sum_{c} p_{C_{i}}(c) \, \mcR_{i}[\ketbra{c}{c}_{C_{i}}] \otimes \rho_{A_{i}B_{i}\tilde{E}_{i-1}}^{|c} = \sum_{c} p_{C_{i}}(c) \, \ketbra{c}{c}_{C_{i}} \otimes \rho_{A_{i}B_{i}\tilde{E}_{i-1}}^{|c} \otimes \tau(c)_{D_{i}}.
     \end{equation*}
     Applying the formula for cq-states to the divergence in \cref{eq:varDiv}, we have
     \begin{equation}
     \begin{aligned}
         D_{\alpha}\Big( [\mcR_{i} \circ \mcM_{i}]&(\omega_{R_{i-1}\tilde{E}_{i-1}}) \Big \| \id_{A_{i}C_{i}D_{i}} \otimes \sigma_{B_{i}\tilde{E}} \Big) \\ &= \frac{1}{\alpha - 1} \log \Bigg( \sum_{c} p_{C_{i}}(c)^{\alpha} 2^{(\alpha - 1)D_{\alpha}\big( \rho_{A_{i}B_{i}\tilde{E}_{i-1}}^{|c} \otimes \tau(c)_{D_{i}} \| \id_{A_{i}D_{i}} \otimes \sigma_{B_{i}\tilde{E}_{i-1}}\big)}\Bigg) \\
         &= \frac{1}{\alpha - 1} \log \Bigg( \sum_{c} p_{C_{i}}(c)^{\alpha} 2^{(\alpha - 1)\big[D_{\alpha}\big( \rho_{A_{i}B_{i}\tilde{E}_{i-1}}^{|c} \| \id_{A_{i}} \otimes \sigma_{B_{i}\tilde{E}_{i-1}}\big) -M + f_{c}\big]}\Bigg)\\
         &= -H_{\alpha}^{f}(A_{i}C_{i}|B_{i}\tilde{E}_{i-1})_{\mcM_{i}(\omega)|\sigma} + M.
    \end{aligned} \label{eq:rec_fdiv}
\end{equation}
 By inserting the upper bound \eqref{eq:rec_fdiv} into the variational expression \eqref{eq:varDiv}, and using the resulting expression to lower bound \eqref{eq:lb2}, we obtain the final lower bound
\begin{multline*}
        H_{\alpha}^{\uparrow}(A^{n}C^{n}|B^{n}E)_{\rho^{|\Omega}}  \geq - n \, M + n\inf_{v_{C} \in \Delta_{\Omega}} v_{C} \cdot f \\ + \sum_{i=1}^{n} \inf_{\omega \in \mcD(R_{i-1}\tilde{E}_{i-1})} \sup_{q_{B_{i}} \in \Delta(B_{i})}  H_{\alpha}^{f}(A_{i}C_{i}|B_{i}\tilde{E}_{i-1})_{\mcM_{i}(\omega)|\sigma} + n \,M - \frac{\alpha}{\alpha-1} \log \frac{1}{p_{\Omega}},
\end{multline*}
completing the proof. 
\end{proof}

\subsection{Optimizing over trade-off functions}
In the next step of the proof, we choose concrete values of the trade-off function $f$ for the event $\Omega$. Before doing so, we review the approach taken in the original EAT proof. In~\cite{DFR}, $f$ is taken to be a min-tradeoff function. Specifically, each entry $f_{c}$ is chosen such that for all $i\in \{1,...,n\}$, and all input states $\omega_{R_{i-1}\tilde{E}}$, 
\begin{equation*}
    H(A_iC_i|B_i\tilde{E})_{\mcM_{i}(\omega)} \geq \sum_{c} f_{c}\cdot p_{C}(c),
\end{equation*}
where $\mcM_{i}(\omega)_{C_i} = \sum_{c}p_{C}(c)\ketbra{c}{c}$ is the distribution on $C_{i}$ induced by $\mcM_{i}(\omega)$. Such a function can be constructed from the tangent of any convex lower bound on the function 
\begin{equation*}
    p_{C} \mapsto \min_{i\in\{1,...,n\}}\inf_{\substack{\omega \in \mcD(R_{i-1}\tilde{E}) \\
    \text{s.t.} \, \mcM_{i}(\omega)_{C_i} = \sum_{c}p_{C}(c)\ketbra{c}{c}}} H(A_iC_i|B_i\tilde{E})_{\mcM_{i}(\omega)},
\end{equation*}
where $p_{C}$ is shorthand for a probability distribution $\{p_{C}(c)\}_{c\in \mcC}$. With this choice, we can take analogues steps to those in the proof of~\cite[Proposition 4.5.]{DFR}, and obtain bound of roughly $n$ times the worst can von Neumann entropy compatible with the abort condition (albeit with a worse second order term). For illustration purposes, we outline this now. 

Rather than employing the lower bound in \cref{eq:rec_fdiv}, we proceed with the entropy $H_{\alpha}(ACD|B\tilde{E})_{\mcM(\omega)}$ in \cref{eq:varDiv} (taking $q_{B} = p_{B}$, where $p_{B}$ is the distribution on $B$ induced by $\mcM(\omega)$). The proof then roughly proceeds as follows (omitting the index $i$ to ease notation):
\begin{equation}
    \begin{aligned}
        H_{\alpha}(ACD|B\tilde{E})_{\mcM(\omega)} &\geq H(ACD|B\tilde{E})_{\mcM(\omega)} - \zeta_{\alpha} \\
        &= H(C|B\tilde{E})_{\mcM(\omega)} + \sum_{c}p_{C}(c)H(AD|B,C=c,\tilde{E})_{\mcM(\omega)} - \zeta_{\alpha} \\
        &= H(C|B\tilde{E})_{\mcM(\omega)} + \sum_{c}p_{C}(c) \big( H(A|B,C=c,\tilde{E})_{\mcM(\omega)} -f_{c} + M\big) - \zeta_{\alpha}  \\
        &= H(AC|B\tilde{E})_{\mcM(\omega)} - \sum_{c}f_{c}\cdot p_{C}(c) + M - \zeta_{\alpha}\\
        &\geq M - \zeta_{\alpha},
    \end{aligned} \label{eq:vN_calc}
\end{equation}
where $\zeta_{\alpha}=(\alpha-1)\big(\log(1+2\text{dim}[ACD])\big)^2$. For the first inequality, we related R\'enyi entropy to the von Neumann entropy via the continuity bound \cite[Lemma B.9]{DFR}. For the first equality, we used the chain rule for the von Neumann entropy $H(AB|C)_{\rho} = H(B|C)_{\rho} + H(A|BC)_{\rho}$, followed by the fact that the system $C$ is classical. We then applied \cref{lem:prep} and the chain rule for the second and third equalities. Finally, we used the fact that $f$ is a min trade-off function for the last inequality. Now, let $\alpha = 1 + 1/\sqrt{n}$, $M = \max_{c}|f_{c}|$ and $h_{\text{vN}}(\Omega) \geq 0$ be such that $v_{C} \cdot f \geq h_{\text{vN}}(\Omega)$ for all $v_{C} \in \Delta_{\Omega}$. Then inserting \cref{eq:vN_calc} back into \cref{eq:lb2} we obtain 
\begin{equation*}
        H_{\alpha}^{\uparrow}(A^{n}C^{n}|B^{n}E)_{\rho^{|\Omega}} \geq  n\, h_{\text{vN}}(\Omega) - \sqrt{n}\big(\log(1+2\,\text{dim}[A] \cdot |\mcC| \cdot 2^{\ceil{M}})\big)^2 - (1+\sqrt{n}) \log \frac{1}{p_{\Omega}}. 
\end{equation*}
To obtain the above bound, we noted $\text{dim}[D] \leq 2^{\ceil{M}}$. This sketch illustrates how the original EAT proof can be recovered; importantly, the error term scales as $O(\sqrt{n})$, which becomes insignificant compared to $n\, h_{\text{vN}}(\Omega)$ as $n$ grows. A lower bound on the min-entropy with the same scaling properties can then be recovered via its relationship to $H_{\alpha}^{\uparrow}$ (see~\cite[Lemma B.10]{DFR}).    

A significant difference between this and the proof approach that we follow, namely that of~\cite{arqand2024}, is that rather than taking $f$ to be a min trade-off function, we take the supremum of \cref{eq:single_round} over all $f \in \mathbb{R}^{|\mcC|}$. Towards this end, in \cref{lem:swap} we exchange the order of the supremum over the trade-off function and infimum in \cref{eq:single_round} using Sion's minimax theorem~\cite{Sion1958}. We then solve the supremum in \cref{lem:singleProve}. These steps follow the the ideas presented in~\cite[Section 5.3]{arqand2024}; in particular, the discussion of Sion's minimax theorem around~\cite[Eq. 113]{arqand2024}, and the proof of~\cite[Lemma 5.3]{arqand2024}.   

\begin{lemma}[Single round convex problem]
    Let $\{\mcM_{i}\}_{i=1}^{n}$ be a sequence of DI channels according to \cref{def:DIchan}, $\Omega \subset \mcC^{n}$ be an event on $C^{n}$ and $\Delta_{\Omega} \subset \Delta(C)$ be any compact convex subset such that $\mathsf{freq}(c^{n}) \in \Delta_{\Omega}$ for all $c^{n} \in \Omega$.  Then for all $\alpha \in (1,\infty)$,
    \begin{multline*}
        \sup_{f\in \mathbb{R}^{|\mcC|}} \sum_{i=1}^{n}\inf_{\omega \in \mcD(R_{i-1}\tilde{E}_{i-1})} \inf_{v_{C} \in \Delta_{\Omega}} \sup_{q_{B_{i}} \in \Delta(B_{i})} \Bigg( v_{C} \cdot f + H_{\alpha}^{f}(A_{i}C_{i}|B_{i}\tilde{E}_{i-1})_{\mcM_{i}(\omega)|\sigma} \Bigg) \\ \geq  n \sup_{M>0}\inf_{\omega \in \mcD(R\tilde{E})}\inf_{v_{C} \in \Delta_{\Omega}} \sup_{q_{B} \in \Delta(B)} \sup_{f\in \mcI_{M}} \Bigg( v_{C} \cdot f + H_{\alpha}^{f}(AC|B\tilde{E})_{\mcM(\omega)|\sigma} \Bigg)
    \end{multline*}
    where $\mcI_{M} = \{ f \in \mbR^{|\mcC|} \ : \ -M \leq f_{c} \leq M \ \forall c \in \mcC\}$, and $\mcM : R \to ACB$ is any channel that satisfies
    \begin{equation*}
        \inf_{\omega \in \mcD(R_{i-1}\tilde{E}_{i-1})}\sup_{q_{B_{i}} \in \Delta(B_{i})}  H_{\alpha}^{f}(A_{i}C_{i}|B_{i}\tilde{E}_{i-1})_{\mcM_{i}(\omega)|\sigma} \geq \inf_{\omega \in \mcD(R\tilde{E})}\sup_{q_{B} \in \Delta(B)}  H_{\alpha}^{f}(AC|B\tilde{E})_{\mcM(\omega)|\sigma}
    \end{equation*}
    for all $i \in \{ 1,...,n\}$ and all $f \in \mbR^{|\mcC|}$. In the above, both infimums over $\omega_{R_{i-1}\tilde{E}_{i-1}}$ and $\omega_{R\tilde{E}}$ are restricted to pure states with $R_{i-1} \cong \tilde{E}_{i-1}$ and $R \cong \tilde{E}$, respectively, and the state $\sigma_{B\tilde{E}}$ ($\sigma_{B_{i}\tilde{E}_{i-1}}$) is defined in terms of $q_{B}$ ($q_{B_{i}}$) and $\mcM(\omega)$ ($\mcM_{i}(\omega)$) analogously to \cref{lem:single}. \label{lem:swap}
\end{lemma}
\begin{proof}
    Using the definition of the channel $\mcM$, 
    \begin{multline*}
        \sup_{f\in \mathbb{R}^{|\mcC|}} \sum_{i=1}^{n}\inf_{\omega \in \mcD(R_{i-1}\tilde{E}_{i-1})} \inf_{v_{C} \in \Delta_{\Omega}} \sup_{q_{B_{i}} \in \Delta(B_{i})} \Bigg( v_{C} \cdot f + H_{\alpha}^{f}(A_{i}C_{i}|B_{i}\tilde{E}_{i-1})_{\mcM_{i}(\omega)|\sigma} \Bigg) \\ \geq n \sup_{f\in \mathbb{R}^{|\mcC|}} \inf_{\omega \in \mcD(R\tilde{E})} \inf_{v_{C} \in \Delta_{\Omega}} \sup_{q_{B} \in \Delta(B)} \Bigg( v_{C} \cdot f +  H_{\alpha}^{f}(AC|B\tilde{E})_{\mcM(\omega)|\sigma} \Bigg).
    \end{multline*}
    We now need to exchange the outer supremum over $f$ with the joint infimum over $\omega_{R\tilde{E}}$ and $v_{C}$. The infimum over $\omega_{R\tilde{E}}$ is restricted to pure states with $R$ and $\tilde{E}$ having the same dimension; this can be equivalently stated as an infimum over $\omega_{R} \in \mcD(R)$ where the objective function is evaluated on any purification of $\omega_{R}$, since the choice of purification will not affect the objective value. The domain $\mcD(R)$ is compact and convex. We also note that $\Delta_{\Omega}$ is compact and convex by assumption. Let $M > 0$, and $\mcI_{M} = \{ f \in \mbR^{|\mcC|} \ : \ -M \leq f_{c} \leq M \ \forall c \in \mcC \}$. The set $\mcI_{M}$ is compact and convex for any $M$, and we can write\footnote{In the following, we will apply Sion's minimax theorem~\cite[Corollary 3.3]{Sion1958} to exchange the supremum over $f\in \mathcal{X}$ with the joint infimum, where $\mathcal{X} \subseteq \mathbb{R}^{|\mathcal{C}|}$. This does not require compactness of $\mathcal{X}$, implying we could take $\mathcal{X} = \mathbb{R}^{|\mathcal{C}|}$. However, we still split the supremum because it improves the clarity of later proofs.} 
    \begin{multline*}
        \sup_{f\in \mathbb{R}^{|\mcC|}} \inf_{\omega \in \mcD(R\tilde{E})} \inf_{v_{C} \in \Delta_{\Omega}} \sup_{q_{B} \in \Delta(B)} \Bigg( v_{C} \cdot f +  H_{\alpha}^{f}(AC|B\tilde{E})_{\mcM(\omega)|\sigma} \Bigg) \\= \sup_{M > 0} \sup_{f\in \mcI_{M}} \inf_{\omega \in \mcD(R\tilde{E})} \inf_{v_{C} \in \Delta_{\Omega}} \sup_{q_{B} \in \Delta(B)} \Bigg( v_{C} \cdot f +  H_{\alpha}^{f}(AC|B\tilde{E})_{\mcM(\omega)|\sigma} \Bigg).
    \end{multline*}
    
    We now recall the following convexity properties of the objective function proven in \cref{lem:concave,lem:convex}. For any fixed $\omega_{R\tilde{E}} \in \mcD(R\tilde{E})$ and $v_{C} \in \mbR^{|\mcC|}$, the function
    \begin{equation}
        f \mapsto \sup_{q_{B} \in \Delta(B)} \Bigg( v_{C} \cdot f +  H_{\alpha}^{f}(AC|B\tilde{E})_{\mcM(\omega)|\sigma} \Bigg)
    \end{equation}
    is concave in $f \in \mcI_{M}$. For any fixed $f \in \mbR^{|\mcC|}$, the function
    \begin{equation}
        (\omega_{R},v_{C}) \mapsto \sup_{q_{B} \in \Delta(B)} \Bigg( v_{C} \cdot f +  H_{\alpha}^{f}(AC|B\tilde{E})_{\mcM(\omega)|\sigma} \Bigg),
    \end{equation}
    where $\omega_{R\tilde{E}}$ is the canonical purification of $\omega_{R}$, is jointly convex in $(\omega_{R},v_{C}) \in \mcD(R) \times \Delta_{\Omega}$. 
    
    We can therefore apply Sion's minimax theorem~\cite[Corollary 3.3]{Sion1958} to exchange the supremum over $f \in \mcI_{M}$ with the joint infimum. We also note that $\Delta(B)$ is a compact convex set and the objective function is continuous in $q_{B}$. The objective function is also continuous in $f$, hence each supremum is achieved by a point inside the domain. This allows us to exchange their order, proving the claim.  
\end{proof}

\begin{lemma}[Optimizing over trade-off functions]
     Let $\{\mcM_{i}\}_{i=1}^{n}$ be a sequence of DI channels according to \cref{def:DIchan}, $\Omega \subset \mcC^{n}$ be an event on $C^{n}$ and $\Delta_{\Omega} \subset \Delta(C)$ be any compact convex subset such that $\mathsf{freq}(c^{n}) \in \Delta_{\Omega}$ for all $c^{n} \in \Omega$.  Let $\mcM:R\to ACB$ be as defined in \cref{lem:swap}, $\omega_{R\tilde{E}} \in \mcD(R\tilde{E})$, $v_{C} \in \Delta_{\Omega}$, $\alpha \in (1,\infty)$ and let $\rho_{AB\tilde{E}}^{|c}$ be the conditional states for which
     \begin{equation*}
         \mcM(\omega) = \sum_{c\in \mathrm{supp}(p_{C})}p_{C}(c) \ketbra{c}{c}_{C} \otimes \rho_{AB\tilde{E}}^{|c}.
     \end{equation*} 
     Then there exists an $0 < M < \infty$ such that
     \begin{multline}
        \inf_{\omega \in \mcD(R\tilde{E})}\inf_{v_{C} \in \Delta_{\Omega}} \sup_{q_{B} \in \Delta(B)} \sup_{f\in \mcI_{M}} \Bigg( v_{C} \cdot f + H_{\alpha}^{f}(AC|B\tilde{E})_{\mcM(\omega)|\sigma} \Bigg)\\ \geq \inf_{\omega \in \mcD(R\tilde{E})}\inf_{v_{C} \in \Delta_{\Omega}} \sup_{q_{B} \in \Delta(B)}\Bigg(\frac{1}{\alpha - 1}\sum_{c \in \mathrm{supp}(p_{C})} v_{C}(c) \log \Big(\frac{v_{C}(c)}{p_{C}(c)^{\alpha}} \Big)
        \\- \sum_{c \in \mathrm{supp}(p_{C})} v_{C}(c)\,D_{\alpha}(\rho_{AB\tilde{E}}^{|c} \| \id_{A} \otimes \sigma_{B\tilde{E}})\Bigg). \label{eq:2prove}
    \end{multline} \label{lem:singleProve}
\end{lemma}
\begin{proof}
    By \cref{lem:concave}, for a fixed $\omega_{R\tilde{E}}$ and $v_{C}$, the objective function is concave in $f$. It therefore suffices to evaluate its stationary point over $\mbR^{|\mcC|}$, and choose an $M > 0$ such that $\mcI_{M}$ contains that point. 

    We first modify the objective function by replacing $p_{C}$ and $v_{C}$ with nearby distributions $\tilde{p}_{C}$ and $\tilde{v}_{C}$ that have full support. Let 
    \begin{equation}
        \tilde{\rho}_{ACB\tilde{E}} = (1-\epsilon)\mcM(\omega_{R\tilde{E}}) + \epsilon \Big(\frac{\id_{A}}{\text{dim}(A)}\otimes \frac{\id_{C}}{|\mcC|} \otimes \sigma_{B\tilde{E}} \Big) = \sum_{c} \tilde{p}_{C}(c) \ketbra{c}{c}_{C} \otimes \tilde{\rho}_{AB\tilde{E}}^{|c},
    \end{equation}
    where $\tilde{p}_{C}(c) > 0$ and $\tilde{\rho}_{AB\tilde{E}}^{|c} \in \mcD(AB\tilde{E})$ for all $c \in \mcC$. Let $k = \| \mcM(\omega_{R\tilde{E}}) - \frac{\id_{A}}{\text{dim}(A)}\otimes \frac{\id_{C}}{|\mcC|} \otimes \sigma_{B\tilde{E}}\|_{1}$ and therefore $\frac{1}{2}\| \mcM(\omega) - \tilde{\rho}_{ACB\tilde{E}}\|_{1} = \epsilon \, k / 2$. Furthermore, let $\tilde{v}_{C}(c) = (1-\epsilon) v_{C}(c) + \epsilon/|\mcC|$. Then
    \begin{equation}
        v_{C} \cdot f = \sum_{c} \big(v_{C}(c) + \tilde{v}_{C}(c) - \tilde{v}_{C}(c)\big)f_{c} = \tilde{v_{C}} \cdot f - \epsilon\Big(1 - \frac{1}{|\mcC|}\Big)\sum_{c}f_{c} \geq \tilde{v_{C}} \cdot f - \epsilon\big(|\mcC| - 1\big)M. \label{eq:vec_cont}
    \end{equation}
    Furthermore, note that the optimization over $q_{B}\in \Delta(B)$ in \cref{eq:2prove} can be restricted to $q_{B}$ such that $\sigma_{B\tilde{E}}$ satisfies $\text{supp}(\rho_{B\tilde{E}}) \subseteq \text{supp}(\sigma_{B\tilde{E}})$, and because $\tilde{\rho}_{B\tilde{E}} = (1-\epsilon)\rho_{B\tilde{E}} + \epsilon \, \sigma_{B\tilde{E}}$ we consequently have $\text{supp}(\tilde{\rho}_{B\tilde{E}}) \subseteq \text{supp}(\sigma_{B\tilde{E}})$. The tuple of states $(\rho_{ACB\tilde{E}},\tilde{\rho}_{ACB\tilde{E}},\sigma_{B\tilde{E}})$ therefore fulfills the conditions needed to apply the continuity bound \cref{lem:f_div_cont}. Combining this with \cref{eq:vec_cont},
    \begin{equation}
        \sup_{f\in \mcI_{M}} \Bigg( v_{C} \cdot f + H_{\alpha}^{f}(AC|B\tilde{E})_{\mcM(\omega)|\sigma} \Bigg) \geq \sup_{f\in \mcI_{M}} \Bigg( \tilde{v_{C}} \cdot f + H_{\alpha}^{f}(AC|B\tilde{E})_{\tilde{\rho}|\sigma} \Bigg) - g(\epsilon,M) \label{eq:shift}
    \end{equation}
    where
    \begin{equation*}
        g(\epsilon,M) = \epsilon\big(|\mcC| - 1\big)M + \frac{\alpha}{\alpha-1}\log\Bigg(1 + \frac{\epsilon \, k}{2} \Bigg(\frac{\mathrm{dim}(A) \cdot |\mcC| \cdot 2^{\ceil{2M}}}{m_{\sigma}}\Bigg)^{\frac{\alpha-1}{\alpha}}\Bigg)
    \end{equation*}
    where $m_{\sigma}$ is the smallest non-zero eigenvalue of $\sigma_{B\tilde{E}}$.
    
    We now solve the optimization problem in \cref{eq:shift}. Let $c \in \mcC$. We can compute
    \begin{multline}
        \frac{\partial}{\partial f_{c}} \Bigg( \tilde{v_{C}} \cdot f + H_{\alpha}^{f}(AC|B\tilde{E})_{\tilde{\rho}|\sigma} \Bigg) \\ = \tilde{v}_{C}(c) - \frac{\exp\big\{ \ln(\tilde{p}_{C}(c)^{\alpha}) + \ln(2)(\alpha - 1) \big[D_{\alpha}(\tilde{\rho}_{AB\tilde{E}}^{|c} \| \id_{A} \otimes \sigma_{B\tilde{E}}) + f_{c}  \big]\big\}}{\sum_{c'} \exp\big\{ \ln(\tilde{p}_{C}(c')^{\alpha}) + \ln(2)(\alpha - 1) \big[D_{\alpha}(\tilde{\rho}_{AB\tilde{E}}^{|c'} \| \id_{A} \otimes \sigma_{B\tilde{E}}) + f_{c'}  \big]\big\}}. \label{eq:deriv}
    \end{multline}
    Let
    \begin{equation*}
        f_{c}^* = \frac{1}{\alpha -1}\log \Big(\frac{\tilde{v}_{C}(c)}{\tilde{p}_{C}(c)^{\alpha}} \Big)- D_{\alpha}(\tilde{\rho}_{AB\tilde{E}}^{|c} \| \id_{A} \otimes \sigma_{B\tilde{E}}). 
    \end{equation*}
    We need consider two cases.
    
    \vspace{0.2cm}

    \noindent \underline{Case 1: $D_{\alpha}(\tilde{\rho}_{AB\tilde{E}}^{|c} \| \id_{A} \otimes \sigma_{B\tilde{E}})$ is finite for all $c \in \mcC$.}

    \vspace{0.2cm}
    
    \noindent In this case, $f_{c}^*$ is well defined and finite for all $c \in \mcC$ because $\tilde{p}_{C} > 0$ and $\tilde{v}_{C} > 0$. Moreover, it satisfies 
    \begin{equation*}
        \exp\big\{ \ln(\tilde{p}_{C}(c)^{\alpha}) + \ln(2)(\alpha - 1) \big[D_{\alpha}(\tilde{\rho}_{AB\tilde{E}}^{|c} \| \id_{A} \otimes \sigma_{B\tilde{E}}) + f_{c}^{*}  \big]\big\} = \tilde{v}_{C}(c).
    \end{equation*}
    Using the fact that $\sum_{c} \tilde{v}_{C}(c) = 1$, the vector $f^{*} = [f_{c}^{*}]_{c\in \mcC}$ then satisfies \cref{eq:deriv}. Moreover, because $f_{c}^*$ is finite for all $c$, there must always exist a large enough $M$ such that $f^{*} \in \mcI_{M}$, implying $f^*$ is a maximizer of \cref{eq:shift}. Inserting the optimal solution $f^*$ into \cref{eq:shift} and noting that $H_{\alpha}^{f^*}(AC|B\tilde{E})_{\tilde{\rho}|\sigma} = 0$, we obtain 
    \begin{equation}
        \tilde{v_{C}} \cdot f^{*} + H_{\alpha}^{f^*}(AC|B\tilde{E})_{\tilde{\rho}|\sigma} = \frac{1}{\alpha - 1}\sum_{c} \tilde{v}_{C}(c) \log \Big(\frac{\tilde{v}_{C}(c)}{\tilde{p}_{C}(c)^{\alpha}} \Big)- \sum_{c} \tilde{v}_{C}(c)\,D_{\alpha}(\tilde{\rho}_{AB\tilde{E}}^{|c} \| \id_{A} \otimes \sigma_{B\tilde{E}}). \label{eq:sol1}
    \end{equation}
    Next, we split the summation over $c$ into a sum over $\text{supp}(p_{C})$ and a sum over the complement set $\text{supp}(p_{C})^{\perp}$. Note that for any $c \in \text{supp}(p_{C})^{\perp}$, $\tilde{p}_{C}(c) = \epsilon/|\mcC|$ and $\tilde{\rho}^{|c}_{AB\tilde{E}} = \id_{A}/\text{dim}(A) \otimes \sigma_{B\tilde{E}}$, hence
    \begin{equation*}
    \begin{gathered}
        \tilde{v}_{C}(c) \log \Big(\frac{\tilde{v}_{C}(c)}{\tilde{p}_{C}(c)^{\alpha}} \Big) = \tilde{v}_{C}(c)\Big( \log(\tilde{v}_{C}(c)) - \alpha \log \Big(\frac{\epsilon}{|\mcC|}\Big)\Big) \geq \tilde{v}_{C}(c)\Big( \log(\epsilon) - \alpha \log \Big(\frac{\epsilon}{|\mcC|}\Big)\Big) \geq 0 \ \ \ \text{and} \\
        -\tilde{v}_{C}(c)\,D_{\alpha}(\rho^{|c}_{AB\tilde{E}}\|\id_{A} \otimes \sigma_{B\tilde{E}}) = -\tilde{v}_{C}(c)\,D_{\alpha}\Bigg(\frac{\id_{A}}{\text{dim}(A)}\otimes \sigma_{B\tilde{E}}\Big\|\id_{A} \otimes \sigma_{B\tilde{E}}\Bigg) = \tilde{v}_{C}(c)\,\text{dim}(A) \geq 0.
    \end{gathered}
    \end{equation*}
    This allows us to lower bound \cref{eq:sol1} by restricting the summation over $\mcC$ to a summation over $\text{supp}(p_{C})$. 
    
    \vspace{0.2cm}

    \noindent \underline{Case 2: $D_{\alpha}(\tilde{\rho}_{AB\tilde{E}}^{|c} \| \id_{A} \otimes \sigma_{B\tilde{E}}) = +\infty$ for some $c \in \mcC$.}

    \vspace{0.2cm}
    
    \noindent If there exists a $c\in \mcC$ such that $D_{\alpha}(\tilde{\rho}_{AB\tilde{E}}^{|c} \| \id_{A} \otimes \sigma_{B\tilde{E}}) = + \infty$, $f^{*} \notin \mcI_{M}$ for a finite $M$. In fact, for any finite $M$ and all $f \in \mcI_{M}$, we have $H_{\alpha}^{f}(AC|B\tilde{E})_{\mcM(\omega)|\sigma} = - \infty$, and the objective function in \cref{eq:shift} diverges to $-\infty$. This is consistent with \cref{eq:sol1} if any of the divergences are not finite. Thus, \cref{eq:sol1} (when $c$ is restricted to $\text{supp}(p_{C})$) lower bounds \cref{eq:shift} without assuming the divergences are finite.  

    \vspace{0.2cm}

    \noindent We have now obtained for any $\epsilon > 0$ and a sufficiently large $M$,
    \begin{multline*}
        \inf_{\omega \in \mcD(R\tilde{E})}\inf_{v_{C} \in \Delta_{\Omega}} \sup_{q_{B} \in \Delta(B)} \sup_{f\in \mcI_{M}} \Bigg( v_{C} \cdot f + H_{\alpha}^{f}(AC|B\tilde{E})_{\mcM(\omega)|\sigma} \Bigg)\\ \geq \inf_{\omega \in \mcD(R\tilde{E})}\inf_{v_{C} \in \Delta_{\Omega}} \sup_{q_{B} \in \Delta(B)}\Bigg(\frac{1}{\alpha - 1}\sum_{c \in \text{supp}(p_{C})} \tilde{v}_{C}(c) \log \Big(\frac{\tilde{v}_{C}(c)}{\tilde{p}_{C}(c)^{\alpha}} \Big)
        \\- \sum_{c \in \text{supp}(p_{C})} \tilde{v}_{C}(c)\,D_{\alpha}(\tilde{\rho}_{AB\tilde{E}}^{|c} \| \id_{A} \otimes \sigma_{B\tilde{E}})\Bigg) - g(\epsilon,M).
    \end{multline*}
    The above expression is continuous in $\epsilon$; this follows from the continuity of the logarithm and the continuity of the sandwiched R\'enyi divergence in its first argument~\cite[Theorem 4.17]{Bluhm_2024}. It is also well defined at $\epsilon = 0$ because the summation is restricted to $c \in \text{supp}(p_{C})$. We can therefore take the limit as $\epsilon \to 0$, noting that $g(\epsilon,M) \to 0$, $\tilde{v}_{C} \to v_{C}$, $\tilde{p}_{C} \to p_{C}$ and $\tilde{\rho}_{AB\tilde{E}}^{|c} \to \rho_{AB\tilde{E}}^{|c}$, to complete the proof.   
\end{proof}

\subsection{Entropy accumulation theorem} \label{sec:big_EAT}

We now present the most general form of our modified R\'enyi EAT, which combines the single round reduction in \cref{lem:single} with the analysis of the single round quantity in \cref{lem:swap,lem:singleProve}.

\begin{theorem}[Entropy accumulation theorem]
    Let $\omega_{R_{0}E} \in \mcD(R_{0}E)$ be a state, $\{\mcM_{i}\}_{i=1}^{n}$ be a sequence of DI channels according to \cref{def:DIchan} and $\rho_{A^{n}C^{n}B^{n}E} = [\tr_{R_{n}} \circ \mcM_{n} \circ \cdots \mcM_{1}](\omega_{R_{0}E})$. Let $\Omega \subset \mcC^{n}$ be an event on $C^{n}$ and $\Delta_{\Omega} \subset \Delta(C)$ be any compact convex subset such that $\mathsf{freq}(c^{n}) \in \Delta_{\Omega}$ for all $c^{n} \in \Omega$. Let $\alpha \in (1,\infty)$ and $\mcM : R \to ACB$ be any quantum channel that satisfies for all $i \in \{ 1,...,n\}$ and all $f \in \mbR^{|\mcC|}$,
    \begin{equation*}
        \inf_{\omega \in \mcD(R_{i-1}\tilde{E}_{i-1})}\sup_{q_{B_{i}} \in \Delta(B_{i})}  H_{\alpha}^{f}(A_{i}C_{i}|B_{i}\tilde{E}_{i-1})_{\mcM_{i}(\omega)|\sigma} \geq \inf_{\omega \in \mcD(R\tilde{E})}\sup_{q_{B} \in \Delta(B)}  H_{\alpha}^{f}(AC|B\tilde{E})_{\mcM(\omega)|\sigma},
    \end{equation*}
    where $\sigma$ is constructed as described below. Then
    \begin{equation*}
        H_{\alpha}^{\uparrow}(A^{n}C^{n}|B^{n}E)_{\rho^{|\Omega}} \geq n \, h_{\alpha}(\Omega) - \frac{\alpha}{\alpha - 1} \log \Big(\frac{1}{p_{\Omega}}\Big),
    \end{equation*}
    where $p_{\Omega}$ is the probability of observing $\Omega$ in $\rho_{C^{n}}$ and
    \begin{multline*}
        h_{\alpha}(\Omega) = \inf_{\omega \in \mcD(R\tilde{E})}\inf_{v_{C} \in \Delta_{\Omega}} \sup_{q_{B} \in \Delta(B)}\Bigg(\frac{1}{\alpha - 1}\sum_{c \in \mathrm{supp}(p_{C})} v_{C}(c) \log \Big(\frac{v_{C}(c)}{p_{C}(c)^{\alpha}} \Big)
        \\- \sum_{c \in \mathrm{supp}(p_{C})} v_{C}(c)\,D_{\alpha}(\rho_{AB\tilde{E}}^{|c} \| \id_{A} \otimes \sigma_{B\tilde{E}})\Bigg),
    \end{multline*}
    where $\mcM(\omega_{R\tilde{E}}) = \sum_{c} p_{C}(c)\ketbra{c}{c}_{C} \otimes \rho_{AB\tilde{E}}^{|c}$ and $\sigma_{B\tilde{E}} = \sum_{b}q_{B}(b)\ketbra{b}{b}_{B} \otimes \rho_{\tilde{E}}^{|b}$ for the collection of conditional states $\{\rho_{\tilde{E}}^{|b}\}_{b}$ defined via $\tr_{AC}[\mcM(\omega_{R\tilde{E}})] = \sum_{b}p_{B}(b)\ketbra{b}{b}_{B}\otimes \rho_{\tilde{E}}^{|b}$. \label{thm:big_eat}
\end{theorem}
\begin{proof}
    By applying \cref{lem:single} to the state $\rho$,
    \begin{multline*}
         H_{\alpha}^{\uparrow}(A^{n}C^{n}|B^{n}E)_{\rho^{|\Omega}} \geq  \sum_{i=1}^{n} \inf_{\omega \in \mcD(R_{i-1}\tilde{E}_{i-1})} \inf_{v_{C} \in \Delta_{\Omega}} \sup_{q_{B_{i}} \in \Delta(B_{i})} \Bigg( v_{C} \cdot f + H_{\alpha}^{f}(A_{i}C_{i}|B_{i}\tilde{E}_{i-1})_{\mcM_{i}(\omega)|\sigma} \Bigg) \\- \frac{\alpha}{\alpha-1} \log \frac{1}{p_{\Omega}},
     \end{multline*}    
     where $\sigma_{B_{i}\tilde{E}_{i-1}} = \sum_{b}q_{B_{i}}(b)\ketbra{b}{b}_{B_{i}} \otimes \rho_{\tilde{E}_{i-1}}^{|b}$ and the set of conditional states $\{\rho_{\tilde{E}_{i-1}}^{|b}\}_{b}$ is defined via $\tr_{A_{i}C_{i}}[\mcM_{i}(\omega_{R_{i-1}\tilde{E}_{i-1}})]=\sum_{b}p_{B_{i}}(b)\ketbra{b}{b}_{B_{i}} \otimes \rho_{\tilde{E}_{i-1}}^{|b}$. Next we apply \cref{lem:swap} to lower bound the summation by
     \begin{equation*}
         n \sup_{M>0}\inf_{\omega \in \mcD(R\tilde{E})}\inf_{v_{C} \in \Delta_{\Omega}} \sup_{q_{B} \in \Delta(B)} \sup_{f\in \mcI_{M}} \Bigg( v_{C} \cdot f + H_{\alpha}^{f}(AC|B\tilde{E})_{\mcM(\omega)|\sigma} \Bigg)
     \end{equation*}
     where $\mcM:R \to ACB$ and $\sigma$ are as defined in the theorem statement. The proof is completed by lower bounding the above expression using \cref{lem:singleProve}.
\end{proof}

\subsection{Properties of the $f$-weighted R\'enyi entropy} \label{sec:f_prop}

In this subsection, we prove the various properties of the $f$-weighted R\'enyi entropy from \cref{def:f_div} that were required in the proof of \cref{thm:big_eat}. A summary of the different properties can be found below:
\begin{enumerate}
    \item \cref{lem:fexp}: a closed form expression for $\sup_{q_{B} \in \Delta(B)} H_{\alpha}^{f}(AC|BE)_{\rho|\sigma}$.
    \item \cref{lem:altexpr}: an expression for $H_{\alpha}^{f}(AC|BE)_{\rho|\sigma}$ as a shifted R\'enyi divergence.
    \item \cref{lem:f_div_cont}: continuity of $H_{\alpha}^{f}(AC|BE)_{\rho|\sigma}$ in $\rho$.
    \item \cref{lem:DPI}: data processing inequality for $\sup_{q_{B} \in \Delta(B)} H_{\alpha}^{f}(AC|BE)_{\rho|\sigma}$.
    \item \cref{lem:concave}: concavity of $\sup_{q_{B} \in \Delta(B)} H_{\alpha}^{f}(AC|BE)_{\mcM(\omega)|\sigma}$ in $f$.
    \item \cref{lem:convex}: convexity of $\sup_{q_{B} \in \Delta(B)} H_{\alpha}^{f}(AC|BE)_{\mcM(\omega)|\sigma}$ in $\omega$.
\end{enumerate}

\begin{lemma}[Closed form expression]
    Let $\rho_{ACBE} = \sum_{c,b}p_{CB}(c,b)\ketbra{cb}{cb}_{CB} \otimes \rho_{AE}^{|c,b}$ be a cq-state where $B$ and $C$ are classical, and $C$ takes values in a finite alphabet $\mcC$. Then for any vector $f\in \mathbb{R}^{|\mcC|}$ and any $\alpha \in (1,\infty)$ the following equality is true:
    \begin{equation*}
        \sup_{q_{B} \in \Delta(B)} H_{\alpha}^{f}(AC|BE)_{\rho|\sigma} = \frac{\alpha}{1-\alpha} \log \Bigg(\sum_{b } p_{B}(b) \Bigg( \sum_{c } p_{C|B}(c|b)^{\alpha} 2^{(\alpha -1)\big[ D_{\alpha}(\rho_{AE}^{|c,b} \| \id_{A} \otimes \rho_{E}^{|b}) + f_{c} \big]}\Bigg)^{\frac{1}{\alpha}}\Bigg),
    \end{equation*}
    where $\sigma_{BE} = \sum_{b}q_{B}(b)\ketbra{b}{b}_{B} \otimes \rho_{E}^{|b}$ and $\rho_{E}^{|b} = \sum_{c} p_{C|B}(c|b)\rho_{E}^{|c,b}$.
    \label{lem:fexp}
\end{lemma}
\begin{proof}
    Using \cref{def:f_div}, for any feasible $\sigma_{BE}$
    \begin{equation*}
        H_{\alpha}^{f}(AC|BE)_{\rho|\sigma} = \frac{1}{1-\alpha} \log \Bigg(\sum_{c} p_{C}(c)^{\alpha} 2^{(\alpha - 1)\big[ D_{\alpha}(\rho_{ABE}^{|c} \| \id_{A} \otimes \sigma_{BE}) + f_{c}\big]} \Bigg)
    \end{equation*}
    where $\rho_{ABE}^{|c} = \sum_{b} p_{B}(b|c) \ketbra{b}{b}_{B} \otimes \rho_{AE}^{|c,b}$. Since $B$ is classical,
    \begin{equation*}
        2^{(\alpha - 1)\big[ D_{\alpha}(\rho_{ABE}^{|c} \| \id_{A} \otimes \sigma_{BE}) + f_{c}\big]} = \sum_{b} p_{B}(b|c)^{\alpha} q_{B}(b)^{1-\alpha} 2^{(\alpha - 1)\big[ D_{\alpha}(\rho_{AE}^{|c,b} \| \id_{A} \otimes \rho_{E}^{|b}) + f_{c} \big]}.
    \end{equation*}
    Inserting this, 
    \begin{equation*}
    \begin{aligned}
        \sup_{q_{B} \in \Delta(B)}H_{\alpha}^{f}(AC|BE)_{\rho|\sigma} &= \sup_{q_{B} \in \Delta(B)} \frac{1}{1-\alpha} \log \Bigg(\sum_{c} p_{C}(c)^{\alpha} \sum_{b} p_{B|C}(b|c)^{\alpha} q_{B}(b)^{1-\alpha} 2^{(\alpha - 1)\big[ D_{\alpha}(\rho_{AE}^{|c,b} \| \id_{A} \otimes \rho_{E}^{|b}) + f_{c} \big]} \Bigg)\\
        &=  \sup_{q_{B} \in \Delta(B)}\frac{1}{1-\alpha} \log \Bigg(\sum_{b}q_{B}(b)^{1-\alpha} \sum_{c} p_{CB}(c,b)^{\alpha} 2^{(\alpha - 1)\big[ D_{\alpha}(\rho_{AE}^{|c,b} \| \id_{A} \otimes \rho_{E}^{|b}) + f_{c} \big]}\Bigg) \\
        &=  \frac{1}{1-\alpha} \log \Bigg(\inf_{q_{B} \in \Delta(B)}\sum_{b}q_{B}(b)^{1-\alpha} r_{b}\Bigg),
    \end{aligned}
    \end{equation*}
    where 
    \begin{equation*}
        r_{b} = \sum_{c} p_{CB}(c,b)^{\alpha} 2^{(\alpha - 1)\big[ D_{\alpha}(\rho_{AE}^{|c,b} \| \id_{A} \otimes \rho_{E}^{|b}) + f_{c} \big]}.
    \end{equation*} 
    We now have the following optimization problem:
    \begin{equation*}
        \begin{aligned}
            \inf \ & \ \sum_{b}q_{B}(b)^{1-\alpha} r_{b} \\
            \text{subject to:} \ & \ \sum_{b} q_{B}(b) = 1,\\
            & \ q_{B}(b) \geq 0.
        \end{aligned}
    \end{equation*}
    It is know that this is a convex optimization problem, and the optimal value is given by $\Big( \sum_{b} r_{b}^{\frac{1}{\alpha}} \Big)^{\alpha}$
    at the point $q_{B}^*(b) = r_{b}^{\frac{1}{\alpha}} / \Big(\sum_{b} r_{b}^{\frac{1}{\alpha}}\Big)$ (see, e.g.,~\cite[Proposition 5.4]{Tomamichel_2016}). We therefore have
    \begin{equation*}
        \sup_{q_{B} \in \Delta(B)}H_{\alpha}^{f}(AC|BE)_{\rho|\sigma} = \frac{\alpha}{1-\alpha}\log \Bigg(\sum_{b} p_{B}(b) \Bigg(\sum_{c} p_{C|B}(c|b)^{\alpha} 2^{(\alpha - 1)\big[ D_{\alpha}(\rho_{AE}^{|c,b} \| \id_{A} \otimes \rho_{E}^{|b}) + f_{c} \big]}\Bigg)^{\frac{1}{\alpha}}\Bigg)
    \end{equation*}
    as claimed.
\end{proof}

The following is a direct analogue of~\cite[Lemma 4.1.]{arqand2024}. 

\begin{lemma}[Divergence expression]
    Let $C$ be a classical system which takes values in a finite alphabet $\mcC$, $\rho_{ABC} = \sum_{c\in \mcC}p_{C}(c)\ketbra{c}{c}_{C} \otimes \rho_{AB}^{|c}$ be a cq-state and $\sigma_{B} \in \mcD(B)$ be a state. Let $f \in \mbR^{|\mcC|}$, $f = [f_{c}]_{c \in \mcC}$ denote a real vector and $M > 0$ be such that $M/2 < M - f_{c}$ for all $c \in \mcC$. Then there exists a read-and-prepare channel $\mcR : C \to CD$ according to \cref{def:randp} such that
    \begin{equation*}
        H_{\alpha}^{f}\big(AC|B\big)_{\rho|\sigma} = -D_{\alpha}\big(\bar{\rho}_{ABCD} \|\id_{ACD} \otimes \sigma_{B}\big) - M
    \end{equation*}
    where $\bar{\rho}_{ABCD} = \mcR(\rho_{ABC})$. \label{lem:altexpr}
\end{lemma}
\begin{proof}
    By \cref{lem:prep}, there always exists a read-and-prepare channel $\mcR:C \to CD$ satisfying
    \begin{equation*}
        H_{\alpha}(D)_{\tau(c)} = H_{\alpha}^{\uparrow}(D)_{\tau(c)} = M - f_{c} 
    \end{equation*}
    for a set of states $\tau(c)_{D}$ all $c \in \mcC$. According to \cref{def:randp}, the state $\bar{\rho}_{ABCD}$ is of the form $\bar{\rho}_{ABCD} = \sum_{c\in \mcC}p_{C}(c)\ketbra{c}{c}_{C} \otimes \rho_{AB}^{|c} \otimes \tau(c)_{D}$. We therefore have
    \begin{equation*}
        \begin{aligned}
            D_{\alpha}\big(\bar{\rho}_{ABCD} \|\id_{ACD} \otimes \sigma_{B}\big) &= \frac{1}{\alpha - 1}\log\Bigg( \sum_{c\in \mcC} p_{C}(c)^{\alpha} 2^{(\alpha-1)D_{\alpha}(\rho_{AB}^{|c}\otimes \tau(c)_{D}\|\id_{AD} \otimes \sigma_{B})}\Bigg)\\
            &= \frac{1}{\alpha - 1}\log\Bigg( \sum_{c\in \mcC} p_{C}(c)^{\alpha} 2^{(\alpha-1)\big[D_{\alpha}(\rho_{AB}^{|c}\|\id_{A} \otimes \sigma_{B}) - H_{\alpha}(D)_{\tau(c)}\big]}\Bigg) \\
            &= \frac{1}{\alpha - 1}\log\Bigg( \sum_{c\in \mcC} p_{C}(c)^{\alpha} 2^{(\alpha-1)\big[D_{\alpha}(\rho_{AB}^{|c}\|\id_{A} \otimes \sigma_{B}) + f_{c} \big]}\Bigg) - M \\
            &= - H_{\alpha}^{f}\big(AC|B\big)_{\rho|\sigma} - M
        \end{aligned}
    \end{equation*}
    as claimed.
\end{proof}

In the following lemma, we show how continuity of the sandwiched R\'enyi divergence in the first argument, as obtained by~\cite{Bluhm_2024}, implies a similar continuity bound for $f$-weighted R\'enyi entropies. 

\begin{lemma}[Continuity]
    Let $C$ be a classical system that takes values in a finite alphabet $\mcC$, $\rho_{ABC} = \sum_{c\in \mcC}p_{C}(c)\ketbra{c}{c}_{C} \otimes \rho_{AB}^{|c}$ and $\tau_{ABC} = \sum_{c\in \mcC}q_{C}(c)\ketbra{c}{c}_{C} \otimes \tau_{AB}^{|c}$ be cq-states and $\sigma_{B} \in \mcD(B)$ be any state such that $\mathrm{supp}(\rho_{B}) \subseteq \mathrm{supp}(\sigma_{B})$ and $\mathrm{supp}(\tau_{B}) \subseteq \mathrm{supp}(\sigma_{B})$. Let $f \in \mbR^{|\mcC|}$, $f = [f_{c}]_{c \in \mcC}$ denote a real vector and $M > 0$ be such that $M/2 < M - f_{c}$ and $-M < f_{c}$ for all $c \in \mcC$. Suppose $\frac{1}{2}\|\rho_{ABC} - \tau_{ABC}\|_{1} \leq \epsilon$ for some $\epsilon \in [0,1]$. Then for any $\alpha \in (1,\infty)$, the $f$-weighted R\'enyi entropy satisfies
    \begin{equation*}
        \Big| H_{\alpha}^{f}\big(AC|B\big)_{\rho|\sigma} -  H_{\alpha}^{f}\big(AC|B\big)_{\tau|\sigma}\Big| \leq \frac{\alpha}{\alpha-1}\log\Bigg(1 + \epsilon \Bigg(\frac{\mathrm{dim}(A) \cdot |\mcC| \cdot 2^{\ceil{2M}}}{m_{\sigma}}\Bigg)^{\frac{\alpha-1}{\alpha}}\Bigg),
    \end{equation*}
    where $m_{\sigma}$ is the smallest non-zero eigenvalue of $\sigma_{B}$. \label{lem:f_div_cont}
\end{lemma}
\begin{proof}
    We first apply \cref{lem:altexpr}, which states that there exists a read-and-prepare channel $\mcR$ that satisfies     
    \begin{equation*}
        \begin{aligned}
            H_{\alpha}^{f}\big(AC|B\big)_{\rho|\sigma} &= -D_{\alpha}\big(\bar{\rho}_{ABCD} \|\id_{ACD} \otimes \sigma_{B}\big) - M \ \ \text{and}\\
            H_{\alpha}^{f}\big(AC|B\big)_{\tau|\sigma} &= -D_{\alpha}\big(\bar{\tau}_{ABCD} \|\id_{ACD} \otimes \sigma_{B}\big) - M
        \end{aligned}
    \end{equation*}
    where $\bar{\rho}_{ABCD} = \mcR(\rho_{ABC})$ and $\bar{\tau}_{ABCD} = \mcR(\tau_{ABC})$. We therefore have
    \begin{equation}
        \Big| H_{\alpha}^{f}\big(AC|B\big)_{\rho|\sigma} -  H_{\alpha}^{f}\big(AC|B\big)_{\tau|\sigma}\Big| = \Big| D_{\alpha}\big(\bar{\rho}_{ABCD} \|\id_{ACD} \otimes \sigma_{B}\big) -  D_{\alpha}\big(\bar{\tau}_{ABCD} \|\id_{ACD} \otimes \sigma_{B}\big) \Big|. \label{eq:c_bound_1}
    \end{equation}
    Note that the data-processing inequality for the trace norm implies $\frac{1}{2}\| \bar{\rho}_{ABCD} - \bar{\tau}_{ABCD}\|_{1}\leq \epsilon$. We can therefore apply a continuity bound for the sandwiched R\'enyi divergence~\cite[Theorem 4.17]{Bluhm_2024}\footnote{We could also apply~\cite[Theorem 4.4]{Bluhm_2024} or ~\cite[Theorem 4.20]{Bluhm_2024} for this purpose. We only choose~\cite[Theorem 4.17]{Bluhm_2024} because it has the simplest expression.}: for any compact convex set $\mcS \subseteq \mcD(\mcH)$ with at least one positive definite state,
    \begin{equation}
        \Big| \inf_{\nu \in \mcS} D_{\alpha}(\rho\|\nu) - \inf_{\nu \in \mcS} D_{\alpha}(\tau\|\nu)\Big | \leq \frac{\alpha}{\alpha - 1}\log\Big( 1 + \epsilon \, \kappa^{\frac{\alpha-1}{\alpha}}\Big) \label{eq:c_bound_2}
    \end{equation}
    for any pair of density operators $\rho,\tau \in \mathcal{D}(\mathcal{H})$ that satisfy $\frac{1}{2}\| \rho - \tau\|_{1}\leq \epsilon$, and where $\kappa$ satisfies
    \begin{equation*}
        \sup_{\rho \in \mcD(\mcH)} \inf_{\nu \in \mcS} D_{\alpha}(\rho\|\nu) \leq \log(\kappa) < \infty.
    \end{equation*}
    For our case, we first restrict the Hilbert space $\mcH_{B}$ to the support of $\sigma_{B}$, and denote it by $\overline{\mcH}_{B}$. The operator $\id_{ACD} \otimes \sigma_{B}$ on $\mcH_{A}\otimes \overline{\mcH}_{B} \otimes \mcH_{C} \otimes \mcH_{D}$ is now positive definite. Note that $\text{supp}(\rho_{B}) \subseteq \text{supp}(\sigma_{B})$ implies $\rho_{B} \in \mcD(\overline{\mcH})$, and similarly $\tau_B \in \mcD(\overline{\mcH})$. We then choose $\mcS = \{ \frac{1}{\text{dim}(ACD)}\id_{ACD} \otimes \sigma_{B}\}$, and note that $\mcS$ is a compact, convex subset of $\mcD(\mcH_{A}\otimes \overline{\mcH}_{B} \otimes \mcH_{C} \otimes \mcH_{D})$ containing a single positive definite state. Furthermore, following~\cite[Remark 4.5]{Bluhm_2024}, we can bound
    \begin{multline*}
        \sup_{\rho \in \mcD(ABCD)}  D_{\alpha}\Big(\rho_{ABCD}\big\|\frac{1}{\text{dim}(ACD)}\id_{ACD} \otimes \sigma_{B}\Big) \leq \\ \log\Big(\Big\|\Big(\frac{1}{\text{dim}(ACD)}\id_{ACD} \otimes \sigma_{B}\Big)^{-1}\Big\|_{\infty}\Big) = \log\Big( \frac{\text{dim}(ACD)}{m_{\sigma}}\Big)
    \end{multline*}
    where $m_{\sigma}$ is the smallest eigenvalue of $\sigma_{B}$, which is non-zero because $\sigma_{B}$ is positive definite. Inserting this into \cref{eq:c_bound_2,eq:c_bound_1} we obtain
    \begin{equation*}
        \Big| H_{\alpha}^{f}\big(AC|B\big)_{\rho|\sigma} -  H_{\alpha}^{f}\big(AC|B\big)_{\tau|\sigma}\Big|  \leq \frac{\alpha}{\alpha-1}\log\Bigg(1 + \epsilon \Bigg(\frac{\text{dim}(ACD)}{m_{\sigma}}\Bigg)^{\frac{\alpha-1}{\alpha}}\Bigg). 
    \end{equation*}
    To complete the proof, we recall that the read-and-prepare channel $\mcR$ is such that $H_{\alpha}(D)_{\tau(c)} = M - f_{c} \leq 2M$, and therefore we can always choose the system $D$ such that $\text{dim}(D)\leq 2^{\ceil{2M}}$.  
\end{proof}

The next lemma is an extension of~\cite[Lemma 4.3.]{arqand2024}. 

\begin{lemma}[Data processing]
    Let $\rho_{ACBE} = \sum_{c,b}p_{CB}(c,b)\ketbra{cb}{cb}_{CB} \otimes \rho_{AE}^{|c,b}$ be a cq-state where $B$ and $C$ are classical, and $C$ takes values in a finite alphabet $\mcC$. Let $f\in \mathbb{R}^{|\mcC|}$, and $\mcN : E \to E'$ be a quantum channel. Then for any $\alpha \in (1,\infty)$ the following equality is true:
    \begin{equation*}
        \sup_{q_{B} \in \Delta(B)} H_{\alpha}^{f}(AC|BE)_{\rho|\sigma} \geq \sup_{q_{B} \in \Delta(B)} H_{\alpha}^{f}(AC|BE')_{\mcN(\rho)|\sigma}.
    \end{equation*} \label{lem:DPI}
\end{lemma}
\begin{proof}
    We use the closed form expression from \cref{lem:fexp}, and note that by the data processing inequality for the sandwiched R\'enyi divergence
    \begin{equation*}
        D_{\alpha}(\rho_{AE}^{|c,b} \| \id_{A} \otimes \rho_{E}^{|b}) \geq D_{\alpha}(\mcN(\rho_{AE}^{|c,b}) \|  \id_{A} \otimes \mcN(\rho_{E}^{|b})).
    \end{equation*}
    This implies
    \begin{equation*}
        \sup_{q_{B} \in \Delta(B)} H_{\alpha}^{f}(AC|BE)_{\rho|\sigma} \geq \frac{\alpha}{1-\alpha} \log \Bigg(\sum_{b } p_{B}(b) \Bigg( \sum_{c } p_{C|B}(c|b)^{\alpha} 2^{(\alpha -1)\big[ D_{\alpha}(\mcN(\rho_{AE}^{|c,b}) \| \id_{A} \otimes \mcN(\rho_{E}^{|b})) + f_{c} \big]}\Bigg)^{\frac{1}{\alpha}}\Bigg).
    \end{equation*}
    Note that $\mcN(\rho_{ACBE}) = \sum_{c,b}p_{CB}(c,b)\ketbra{cb}{cb}_{CB} \otimes \mcN(\rho_{AE}^{|c,b})$, and by \cref{lem:fexp} 
    \begin{equation*}
        \frac{1}{1-\alpha} \log \Bigg(\sum_{b } p_{B}(b) \Bigg( \sum_{c } p_{C|B}(c|b)^{\alpha} 2^{(\alpha -1)\big[ D_{\alpha}(\mcN(\rho_{AE}^{|c,b}) \| \id_{A} \otimes \mcN(\rho_{E}^{|b})) + f_{c} \big]}\Bigg)^{\frac{1}{\alpha}}\Bigg) = \sup_{q_{B} \in \Delta(B)} H_{\alpha}^{f}(AC|BE')_{\mcN(\rho)|\sigma},
    \end{equation*}
    completing the proof.
\end{proof}

To conclude this subsection, we analyze the convexity of our extended definition of the $f$-weighted R\'enyi entropy. The standard case (i.e., when $q_{B} = p_{B}$ and there is no additional supremum) was established in \cite[Lemma 4.8.]{arqand2024}, and we follow this proof approach, modifying it to accommodate the additional supremum. 

\begin{lemma}[Concavity in $f$]
   Let $\rho_{ACBE} = \sum_{c,b}p_{CB}(c,b)\ketbra{cb}{cb}_{CB} \otimes \rho_{AE}^{|c,b}$ be a cq-state where $B$ and $C$ are classical, and $C$ takes values in a finite alphabet $\mcC$. For any $\alpha \in (1,\infty)$, the function 
    \begin{equation*}
        f \mapsto \sup_{q_{B} \in \Delta(B)} H_{\alpha}^{f}(AC|BE)_{\rho|\sigma}
    \end{equation*}
    is concave in $f \in \mathbb{R}^{|\mcC|}$. \label{lem:concave}
\end{lemma}
\begin{proof}
    We begin by using \cref{lem:fexp}, which allows us to write
    \begin{equation*}
       \sup_{q_{B} \in \Delta(B)} H_{\alpha}^{f}(AC|BE)_{\rho|\sigma} = \frac{\alpha}{1-\alpha} \log \Bigg(\sum_{b } p_{B}(b) \Bigg( \sum_{c } p_{C|B}(c|b)^{\alpha} 2^{(\alpha -1)\big[ D_{\alpha}(\rho_{AE}^{|c,b} \| \id_{A} \otimes \rho_{E}^{|b}) + f_{c} \big]}\Bigg)^{\frac{1}{\alpha}}\Bigg).
    \end{equation*}
    We can rewrite the inner summation as 
    \begin{equation*}
    \begin{aligned}
        \Bigg(\sum_{c} p_{C|B}(c|b)^{\alpha} 2^{(\alpha-1)\big[ D_{\alpha}(\rho_{AE}^{|c,b} \| \id_{A} \otimes \rho_{E}^{|b}) + f_{c}\big]} \Bigg)^{\frac{1}{\alpha}} &= \Bigg( \sum_{c} \exp \Big\{ \alpha \ln(p_{C|B}(c|b)) \\ & \hspace{2cm} + \ln(2)(\alpha-1)\big[ D_{\alpha}(\rho_{AE}^{|c,b} \| \id_{A} \otimes \rho_{E}^{|b}) + f_{c}\big] \Big\} \Bigg)^{\frac{1}{\alpha}} \\
        &= \exp \Big\{\frac{1}{\alpha} \text{lse}\Big [ \alpha \ln(p_{b}) + \ln(2)(\alpha - 1)\big[ f + D_{b}\big] \Big] \Big\} \\
        &= \exp \big\{ g_{b}(f) \}
    \end{aligned}
    \end{equation*}
    where for a vector $t = [t_{c}]_{c \in \mcC}$, $\text{lse} [t] = \ln \sum_{c} e^{t_{c}}$ is the log-sum-exponential function, we defined a family of real valued functions
    \begin{equation*}
        g_{b}(f) := \frac{1}{\alpha} \text{lse}\Big [ \alpha \ln(p_{b}) + \ln(2)(\alpha - 1)\big[ f + D_{b}\big] \Big],
    \end{equation*}
    and we defined the vectors $D_{b} = [D_{\alpha}(\rho_{AE}^{|c,b} \| \id_{A} \otimes \rho_{E}^{|b})]_{c}$ and $\ln(p_{b}) = [\ln(p_{C|B}(c|b))]_{c}$ in $\mbR^{|\mcC|}$. Due to the convexity of the log-sum-exponential function~\cite{BV04}, for every $b$ the function $f \mapsto g_{b}(f)$ is convex. We then have
    \begin{equation*}
    \begin{aligned}
        \sup_{q_{B} \in \Delta(B)} H_{\alpha}^{f}(AC|BE)_{\rho|\sigma} &= \frac{\alpha}{1-\alpha}\frac{1}{\ln(2)} \ln \Bigg(\sum_{b} \exp \Big\{ \ln(p_{B}(b)) + g_{b}(f)\Big\}\Bigg) \\
        &= \frac{\alpha}{1-\alpha}\frac{1}{\ln(2)} \text{lse} \Big[ \ln(p) + g(f) \Big],
    \end{aligned}
    \end{equation*}
    where we defined the vectors $\ln(p) = [\ln(p_{B}(b))]_{b}$ and $g(f)  = [g_{b}(f)]_{b}$ in $\mbR^{\text{dim}(B)}$. Denoting the left hand side of the above equation by the function $h(f)$, we find for any $\lambda \in [0,1]$ and any $f_{1}, \, f_{2} \in \mbR^{|\mcC|}$,
    \begin{equation*}
        \begin{aligned}
            h(\lambda f_{1} + (1-\lambda)f_{2}) &=  \frac{\alpha}{1-\alpha}\frac{1}{\ln(2)} \text{lse} \Big[ \ln(p) + g(\lambda f_{1} + (1-\lambda)f_{2}) \Big] \\
            &\geq  \frac{\alpha}{1-\alpha}\frac{1}{\ln(2)} \text{lse} \Big[  \lambda[ \ln(p) + g( f_{1})] + (1-\lambda)[\ln(p) + g(f_{2})] \Big] \\
            &\geq \lambda \frac{\alpha}{1-\alpha}\frac{1}{\ln(2)} \text{lse} \Big[ \ln(p) + g( f_{1}) \Big] + (1-\lambda) \frac{\alpha}{1-\alpha}\frac{1}{\ln(2)} \text{lse} \Big[ \ln(p) + g(f_{2}) \Big]\\
            &= \lambda h(f_{1}) + (1-\lambda)h(f_{2}),
        \end{aligned}
    \end{equation*}
    where for the first inequality we used the fact that $g(f)$ is component-wise convex in $f$ and the log-sum-exponential function is non-decreasing (i.e., for two vectors $t$ and $u$, $\text{lse}[t] \leq \text{lse}[u]$ if $t_{c} \leq u_{c}$ for all $c$), and for the final inequality we used the convexity of the log-sum-exponential function. This completes the proof.    
\end{proof}

\begin{lemma}[Convexity in $\omega_{R}$]
    Let $\mcM : R \to ACB$ be a quantum channel, where $B$ and $C$ are classical, $C$ takes values in a finite alphabet $\mcC$, and for any input state $\omega_{RR'}$ the output state $\rho_{ACBR'} = \mcM(\omega_{RR'})$ satisfies $\rho_{BR'} = \rho_{B} \otimes \omega_{R'}$. Let $f \in \mathbb{R}^{|\mcC|}$ be a vector and $\alpha \in (1,\infty)$. Then the function
    \begin{equation*}
        \omega_{R} \mapsto \sup_{q_{B} \in \Delta(B)} H_{\alpha}^{f}(AC|B\tilde{E})_{\mcM(\omega)|\sigma}
    \end{equation*}
    where the channel $\mcM$ is applied to any purification $\omega_{R\tilde{E}}$ of $\omega_{R}$, is concave in $\omega_{R}$. \label{lem:convex}
\end{lemma}
\begin{proof}
    Consider two states $\omega^{|0}_{R}$ and $\omega^{|1}_{R}$, and define $\bar{\omega}_{R} := \sum_{i=0}^{1} \lambda_{i}\, \omega^{|i}_{R}$ where $\{\lambda_{i}\}_{i=0}^{1}$ is a probability distribution. Consider the following purification $\ket{\psi}_{REF}$ of $\bar{\omega}_{R}$,
    \begin{equation*}
        \ket{\psi}_{REF} := \sum_{i=0}^{1} \sqrt{\lambda_{i}} \ket{\psi_{i}}_{RE} \otimes \ket{i}_{F},
    \end{equation*}
    where $\ket{\psi_{i}}_{RE}$ is a purification of $\omega_{R}^{|i}$. The state $\bar{\omega}_{R}$ is therefore mapped in the following way:
    \begin{equation*}
        \bar{\omega}_{R} \mapsto \sup_{q_{B} \in \Delta(B)} H_{\alpha}^{f}(AC|BEF)_{\mcM(\ketbra{\psi}{\psi})|\sigma}.
    \end{equation*}
    Let $\mcZ_{F}(\sigma) = \sum_{i=0}^{1}\ketbra{i}{i}_{F} \sigma \ketbra{i}{i}_{F}$ be the pinching channel on $F$. Using the data processing inequality in \cref{lem:DPI} and the fact that $\mcZ_{F} \circ \mcM = \mcM \circ \mcZ_{F}$,  
    \begin{equation}
        \begin{aligned}
             \bar{\omega}_{R} \mapsto \sup_{q_{B} \in \Delta(B)} H_{\alpha}^{f}(AC|BEF)_{\mcM(\ketbra{\psi}{\psi})|\sigma} &\leq \sup_{q_{B} \in \Delta(B)} H_{\alpha}^{f}(AC|BEF)_{[\mcZ_{F} \circ \mcM](\ketbra{\psi}{\psi})|\sigma} \\
            &= \sup_{q_{B} \in \Delta(B)} H_{\alpha}^{f}(AC|BEF)_{[\mcM \circ \mcZ_{F}](\ketbra{\psi}{\psi})|\sigma}.
        \end{aligned} \label{eq:dpi_app_con}
    \end{equation}
    
    We now analyze the state $[\mcM \circ \mcZ_{F}](\ketbra{\psi}{\psi})$. Note that without loss of generality
    \begin{equation}
        \begin{aligned}
            [\mcM \circ \mcZ_{F}](\ketbra{\psi}{\psi}_{REF}) &= \sum_{i=0}^{1}\lambda_{i} \mcM(\ketbra{\psi_{i}}{\psi_{i}}_{RE}) \otimes \ketbra{i}{i}_{F} \\
            &= \sum_{i=0}^{1}\lambda_{i} \sum_{c,b}p_{B|F}(b|i)\,p_{C|BF}(c|b,i)\,\ketbra{cb}{cb}_{CB} \otimes \rho_{AE}^{|c,b,i} \otimes \ketbra{i}{i}_{F}
        \end{aligned} \label{eq:pinch_cq}
    \end{equation}
    where we used the fact that both $B$ and $C$ are classical. Tracing out $AC$, recall that the channel $\mcM$ satisfies $\rho_{BEF} = \rho_{B} \otimes [\tr_{R}\circ\mcZ_{F}](\ketbra{\psi}{\psi})$ where $\rho_{BEF}$ and $\rho_{B}$ are marginals of the output state $\rho_{ACBEF} = [\mcM \circ \mcZ_{F}](\ketbra{\psi}{\psi})$. This constraint evaluates to
    \begin{multline*}
        \sum_{i=0}^{1}\lambda_{i} \sum_{b}p_{B|F}(b|i)\,\ketbra{b}{b}_{B} \otimes \sum_{c}p_{C|BF}(c|b,i)\rho_{E}^{|c,b,i} \otimes \ketbra{i}{i}_{F} \\ = \Bigg(\sum_{i=0}^{1}\lambda_{i} \sum_{b}p_{B|F}(b|i)\,\ketbra{b}{b}_{B}\Bigg) \otimes \Bigg(\sum_{j=0}^{1}\lambda_{j} \tr_{R}[\ketbra{\psi_{j}}{\psi_{j}}] \otimes \ketbra{i}{i}_{F}\Bigg).
    \end{multline*}
    Multiplying by $\bra{b}_{B}\otimes \id_{E} \otimes \bra{i}_{F}$ on the left and by $\ket{b}_{B}\otimes \id_{E} \otimes \ket{i}_{F}$ on the right we find
    \begin{equation*}
        \lambda_{i}\,p_{B|F}(b|i)\,\rho_{E}^{|i,b} = \lambda_{i} \, p_{B}(b) \, \tr_{R}[\ketbra{\psi_{i}}{\psi_{i}}]
    \end{equation*}
    where $p_{B}(b) = \sum_{i=0}^{1}\lambda_{i} \, p_{B|F}(b|i)$. Tracing out $E$ further implies $\lambda_{i}\,p_{B|F}(b|i) = \lambda_{i}\,p_{B}(b)$, i.e., the distribution on $B$ is independent of $F$. It therefore follows that \eqref{eq:pinch_cq} can be rewritten as
    \begin{equation}
    \begin{aligned}
        [\mcM \circ \mcZ_{F}](\ketbra{\psi}{\psi}) &= \sum_{c} p_{C}(c) \ketbra{c}{c}_{C} \otimes \sum_{b,i} \frac{\lambda_{i}\,p_{B}(b)\,p_{C|BF}(c|b,i)}{p_{C}(c)} \ketbra{b}{b}_{B} \otimes \rho_{AE}^{|c,b,i} \otimes \ketbra{i}{i}_{F}\\
        &= \sum_{c} p_{C}(c) \ketbra{c}{c}_{C} \otimes \sum_{i} p_{F|C}(i|c)  \ketbra{i}{i}_{F} \otimes \sum_{b} \frac{\lambda_{i}\,p_{B}(b)\, p_{C|BF}(c|b,i)}{p_{C}(c)\, p_{F|C}(i|c)} \ketbra{b}{b}_{B} \otimes \rho_{AE}^{|c,b,i}\\
        &= \sum_{c} p_{C}(c) \ketbra{c}{c}_{C} \otimes \sum_{i} p_{F|C}(i|c)  \ketbra{i}{i}_{F} \otimes \rho_{ABE}^{|c,i}
    \end{aligned}
    \label{eq:cq_express}
    \end{equation}
    where 
    \begin{equation}
        \begin{aligned}
            p_{C}(c) &= \sum_{b,i}\lambda_{i}\,p_{B}(b)\,p_{C|BF}(c|b,i),\\
            p_{F|C}(i|c) &= \frac{\lambda_{i}\, p_{C|F}(c|i)}{p_{C}(c)}  = \frac{\lambda_{i}}{p_{C}(c)} \sum_{b}p_{B}(b)p_{C|FB}(c|i,b), \ \ \ \text{and}\\
            \rho_{ABE}^{|c,i} &= \sum_{b} \frac{\lambda_{i}\,p_{B}(b)\, p_{C|BF}(c|b,i)}{p_{C}(c)\, p_{F|C}(i|c)} \ketbra{b}{b}_{B} \otimes \rho_{AE}^{|c,b,i}.
        \end{aligned} \label{eq:p_dists}
    \end{equation}
    We next turn our attention to any feasible $q_{B} \in \Delta(B)$ in \cref{eq:dpi_app_con}, which defines a state 
    \begin{equation}
        \sigma_{BEF} = \sum_{b}q_{B}(b)\ketbra{b}{b}_{B} \otimes \rho_{EF}^{|b}. \label{eq:sigma_st}
    \end{equation}
    Using \cref{eq:cq_express}, the conditional states take the form $\rho_{EF}^{|b} = \sum_{i}\lambda_{i}\sum_{c}p_{C|BF}(c|b,i)\rho_{E}^{|c,b,i} \otimes \ketbra{i}{i}_{F}$, and by inserting this into \cref{eq:sigma_st} and re-ordering the summation we obtain
    \begin{equation}
        \sigma_{BEF} = \sum_{i}\lambda_{i} \sigma_{BE}^{|i} \otimes \ketbra{i}{i}_{F}, \label{eq:sigma_expr}
    \end{equation}
    where we defined $\sigma_{BE}^{|i} = \sum_{B}q_{B}(b)\ketbra{b}{b}_{B} \otimes \sum_{c} p_{C|BF}(c|b,i)\rho_{E}^{|c,b,i}$. 
    
    We are now ready to evaluate the expression on the right hand side of \cref{eq:dpi_app_con}. By definition,
    \begin{equation}
        H_{\alpha}^{f}(AC|BEF)_{[\mcM \circ \mcZ_{F}](\ketbra{\psi}{\psi})|\sigma} = \frac{1}{1-\alpha} \log \Bigg( \sum_{c}p_{C}(c)^{\alpha} 2^{(\alpha-1)\big[ D_{\alpha}(\rho_{ABEF}^{|c} \| \id_{A} \otimes \sigma_{BEF}) + f_{c}]} \Bigg). \label{eq:rho1}
    \end{equation}
    Using the fact that $F$ is classical and the form of $\rho_{ABEF}^{|c}$ and $\sigma_{BEF}$ given by \eqref{eq:cq_express} and \eqref{eq:sigma_expr}, respectively,
    \begin{equation*}
    \begin{aligned}
        2^{(\alpha-1)\big[ D_{\alpha}(\rho_{ABEF}^{|c} \| \id_{A} \otimes \sigma_{BEF}) + f_{c}]} &= \sum_{i=0}^{1}p_{F|C}(i|c)^{\alpha}\lambda_{i}^{1-\alpha} 2^{(\alpha-1)\big[D_{\alpha}(\rho_{ABE}^{|c,i}\|\id_{A}\otimes \sigma_{BE}^{|i}) + f_{c} \big]}\\
        &=\sum_{i=0}^{1}\lambda_{i}\Big(\frac{ p_{C|F}(c|i)}{p_{C}(c)}\Big)^{\alpha} 2^{(\alpha-1)\big[D_{\alpha}(\rho_{ABE}^{|c,i}\|\id_{A}\otimes \sigma_{BE}^{|i}) + f_{c} \big]},
    \end{aligned}
    \end{equation*}
    where for the second equality we inserted the definition of $p_{F|C}(i|c)$ from \eqref{eq:p_dists}. Inserting this back into \cref{eq:rho1}, together with \cref{eq:dpi_app_con} we obtain the following series of inequalities,
    \begin{equation}
        \begin{aligned}
            \bar{\omega}_{R} \mapsto \sup_{q_{B} \in \Delta(B)} H_{\alpha}^{f}(AC|BEF)_{\mcM(\ketbra{\psi}{\psi})|\sigma} &\leq \sup_{q_{B} \in \Delta(B)} H_{\alpha}^{f}(AC|BEF)_{[\mcM \circ \mcZ_{F}](\ketbra{\psi}{\psi})|\sigma} \\
            &= \sup_{q_{B} \in \Delta(B)} \frac{1}{1-\alpha} \log \Bigg( \sum_{i=0}^{1}\lambda_{i} \sum_{c} p_{C|F}(c|i)^{\alpha}2^{(\alpha-1)\big[D_{\alpha}(\rho_{ABE}^{|c,i}\|\id_{A}\otimes \sigma_{BE}^{|i}) + f_{c} \big]} \Bigg) \\
            &\leq \sup_{q_{B} \in \Delta(B)} \sum_{i} \lambda_{i}\frac{1}{1-\alpha} \log \Bigg( \sum_{c} p_{C|F}(c|i)^{\alpha}2^{(\alpha-1)\big[D_{\alpha}(\rho_{ABE}^{|c,i}\|\id_{A}\otimes \sigma_{BE}^{|i}) + f_{c} \big]} \Bigg)\\
            &\leq  \sum_{i} \lambda_{i}  \sup_{q_{B} \in \Delta(B)} \frac{1}{1-\alpha} \log \Bigg( \sum_{c} p_{C|F}(c|i)^{\alpha}2^{(\alpha-1)\big[D_{\alpha}(\rho_{ABE}^{|c,i}\|\id_{A}\otimes \sigma_{BE}^{|i}) + f_{c} \big]} \Bigg)\\
            &=  \sum_{i} \lambda_{i}  \sup_{q_{B} \in \Delta(B)} H_{\alpha}^{f}(AC|BEF)_{\mcM(\ketbra{\psi_{i}}{\psi_{i}})|\sigma}.
        \end{aligned} \label{eq:final_convex}
    \end{equation}
    For the second inequality, we used the concavity of the logarithm, and for the second we passed the supremum inside the sum. The final equality follows from the fact that
    \begin{equation*}
    \begin{aligned}
        \mcM(\ketbra{\psi_{i}}{\psi_{i}}) &= \sum_{c} \ketbra{c}{c}_{C} \otimes \sum_{b} p_{B}(b)\,p_{C|BF}(c|b,i) \ketbra{b}{b}_{B} \otimes \rho_{AE}^{|c,b,i} \\
        &= \sum_{c} p_{C|R}(c|i)\ketbra{c}{c}_{C} \otimes \sum_{b} \frac{p_{B}(b)\,p_{C|BF}(c|b,i)}{p_{C|R}(c|i)} \ketbra{b}{b}_{B} \otimes \rho_{AE}^{|c,b,i}\\
        &= \sum_{c} p_{C|R}(c|i)\ketbra{c}{c}_{C} \otimes \rho_{ABE}^{|c,i},
    \end{aligned}
    \end{equation*}
   where we used \cref{eq:p_dists} to write 
    \begin{equation*}
        \frac{p_{B}(b)\,p_{C|BF}(c|b,i)}{p_{C|R}(c|i)} = \frac{\lambda_{i}\,p_{B}(b)\,p_{C|BF}(c|b,i)}{\lambda_{i}\,p_{C|R}(c|i)} = \frac{\lambda_{i}\,p_{B}(b)\,p_{C|BF}(c|b,i)}{p_{C}(i)\,p_{F|C}(i|c)}
    \end{equation*}
    followed by the definition of $\rho_{ABE}^{|c,i}$. Furthermore, 
    \begin{equation*}
        [\tr_{AC} \circ \mcM](\ketbra{\psi_{i}}{\psi_{i}}) = \sum_{b} p_{B}(b) \ketbra{b}{b}_{B} \otimes \sum_{c}p_{C|BF}(c|b,i)\rho_{E}^{|c,b,i},
    \end{equation*}
    which after making the substitution $p_{B} \mapsto q_{B}$ is exactly the form of $\sigma_{BE}^{|i}$ as defined below \cref{eq:sigma_expr}.  

    To complete the proof, we recall that $\ket{\psi_{i}}_{RE}$ is a purification of $\omega^{|i}_{E}$, hence $\omega^{|i}_{E}$ is mapped in the following way:
    \begin{equation*}
        \omega^{|i}_{E} \mapsto \sup_{q_{B} \in \Delta(B)} H_{\alpha}^{f}(AC|BEF)_{\mcM(\ketbra{\psi_{i}}{\psi_{i}})|\sigma}.
    \end{equation*}
    Thus, \cref{eq:final_convex} proves the convexity of the function stated in the lemma. 
\end{proof}

\subsection{Application to infrequent sampling channels} \label{app:infreq_sample}

While \cref{thm:big_eat} is the most general accumulation statement proven in this manuscript, it is cumbersome to work with directly. For specific applications, certain simplifications can be made to make the analysis easier, whilst still maintaining an improvement over previous approaches~\cite{arqand2024,hahn2024,hahn2025}. We now do this by applying \cref{thm:big_eat} to a specific type of protocol commonly encountered in DI cryptography: spot-checking protocols. We begin by recalling \cref{def:sample} from the main text. 

\vspace{0.2cm}

\noindent \textbf{Definition 5.2} (Infrequent sampling channels)\textbf{.}
    A sequence of DI channels $\{\mcM_{i}\}_{i=1}^{n}$ (according to \cref{def:DIchan}) is infrequent sampling if the following holds:
    \begin{enumerate}
        \item Each $C_{i}$ takes values in $\{0,1\}^{d} \cup \{\perp\}$ for a positive integer $d$.
        \item Each $B_{i}$ is of the form $T_{i}B_{i}$ where $T_{i}$ is a classical bit and each $B_{i}$ is isomorphic to a single register $B$ that takes values in a finite alphabet $\mcB$.
        \item Each $A_{i}$ is classical and is isomorphic to a single register $A$ that takes values in a finite alphabet $\mcA$.
        \item There exists a $\gamma \in [0,1]$, distributions $p_{B}^{\text{gen}}, \, p_{B}^{\text{test}} \in \Delta(B)$, a deterministic function $f : \mcA \times \mcB \to \{0,1\}^{d}$ and for every $i$, a collection of CP maps $\mcM^{a|b}_{i}:R_{i-1} \to R_{i}$ satisfying $\sum_{a}\tr[\mcM_{i}^{a|b}(\omega)] = 1$ for all input states $\omega$ and all $b\in \mcB$, such that
        \begin{equation}
        \mcM_{i} = (1-\gamma) \ketbra{0}{0}_{T_{i}} \otimes \mcM^{\text{gen}}_{i}  
        + \gamma \, \ketbra{1}{1}_{T_{i}} \otimes \mcM^{\text{test}}_{i}, \label{eq:sample_ch}
        \end{equation}
        where
        \begin{equation*}
            \begin{aligned}
                \mcM^{\text{gen}}_{i} = \sum_{a\in \mcA,b\in \mcB} p_{B}^{\text{gen}}(b)\ketbra{a}{a}_{A_{i}} \otimes \ketbra{\perp}{\perp}_{C_{i}} \otimes \ketbra{b}{b}_{B_{i}} \otimes \mcM_{i}^{a|b} \ \ \text{and} \\
                \mcM^{\text{test}}_{i} = \sum_{a\in \mcA,b\in \mcB} p_{B}^{\text{test}}(b)\ketbra{a}{a}_{A_{i}} \otimes \ketbra{f(a,b)}{f(a,b)}_{C_{i}} \otimes \ketbra{b}{b}_{B_{i}} \otimes \mcM_{i}^{a|b}.
            \end{aligned}
        \end{equation*}
    \end{enumerate} 

\vspace{0.2cm}

\noindent For a sequence of infrequent sampling channels, we can provide a simplified version \cref{thm:big_eat} that does not involve divergences, and only includes the entropy obtained from generation rounds. While not tight in general, the single round optimization can be handled in a similar manner to~\cite{hahn2024,hahn2025}. The following corollary is an analogue of~\cite[Lemma 5.1.]{arqand2024} and~\cite[Eq. 93]{arqand2024}.

\vspace{0.2cm}

\noindent \textbf{Corollary 5.3} (EAT with infrequent sampling)\textbf{.}
    \textit{Let $\omega_{R_{0}E} \in \mcD(R_{0}E)$ be a state, $\{\mcM_{i}\}_{i=1}^{n}$ be a sequence of DI channels that are infrequent sampling according to \cref{def:sample} and $\rho_{A^{n}C^{n}B^{n}E} = [\tr_{R_{n}} \circ \mcM_{n} \circ \cdots \mcM_{1}](\omega_{R_{0}E})$. Let $\Omega \subset \mcC^{n}$ be an event on $C^{n}$, $\rho_{A^{n}C^{n}B^{n}E}^{|\Omega}$ be the output state conditioned on $\Omega$ and $\Delta_{\Omega} \subset \Delta(C)$ be any compact convex subset such that $\mathsf{freq}(c^{n}) \in \Delta_{\Omega}$ for all $c^{n} \in \Omega$. Then for any $\alpha \in (1,\infty)$, 
    \begin{equation*}
        H_{\alpha}^{\uparrow}(A^{n}C^{n}|B^{n}E)_{\rho^{|\Omega}} \geq n \, h_{\alpha}(\Omega) - \frac{\alpha}{\alpha - 1} \log \Big(\frac{1}{p_{\Omega}}\Big)
    \end{equation*}
    where $p_{\Omega}$ is the probability of observing $\Omega$ in $\rho_{C^{n}}$ and
    \begin{equation*}
        h_{\alpha}(\Omega) \geq \inf_{(\mcM,R,\tilde{E})} \inf_{\omega \in \mcD(R\tilde{E})}\inf_{v_{C} \in \Delta_{\Omega}} \Bigg(\frac{1}{\alpha - 1}D\big(v_{C}\|p_{C}\big)
        + v_{C}(\perp)\,H_{\alpha}(A|B^{\uparrow}\tilde{E}^{\downarrow})_{\mcM^{\mathrm{gen}}(\omega)}\Bigg),
    \end{equation*}
    where the outer infimum is over all quantum systems $R$ and $\tilde{E}$ and all quantum channels $\mcM:R\to ACB$ of the form \eqref{eq:sample_ch}, $[\tr_{AB\tilde{E}} \circ \mcM](\omega) = \sum_{c}p_{C}(c) \ketbra{c}{c}$ and $p_{C} = [p_{C}(c)]_{c\in \mcC}$.}
\begin{proof}
    We begin by applying \cref{thm:big_eat}. To do so, we need to specify a channel $\mcM: R \to ACB$ that satisfies
    \begin{equation}
        \inf_{\omega \in \mcD(R_{i-1}\tilde{E}_{i-1})}\sup_{q_{B_{i}T_{i}} \in \mcP_{B_{i}T_{i}}}  H_{\alpha}^{f}(A_{i}C_{i}|B_{i}T_{i}\tilde{E}_{i-1})_{\mcM_{i}(\omega)|\sigma} \geq \inf_{\omega \in \mcD(R\tilde{E})}\sup_{q_{BT} \in \mcP_{BT}}  H_{\alpha}^{f}(AC|BT\tilde{E})_{\mcM(\omega)|\sigma} \label{eq:chan_cond}
    \end{equation}
    for all $f$ and each infrequent sampling channel $\mcM_{i}$. We follow the same approach as that outlined in~\cite[Appendix D]{arqand2024}. Let $\mcH_{\tilde{R}}$ be such that $\text{dim}(R_{i}) \leq \text{dim}(\tilde{R})$ for all $i$, and let $F$ denote a quantum system with $\text{dim}(F) = n$. Let $\omega \in \mcD(FR)$ be any quantum state, and define the quantum channel $\mcM : R \to FACBT$ (recall every $A_{i}$ is isomorphic to a single register $A$, and similarly for $C_{i}$, $B_{i}$ and $T_{i}$),
    \begin{equation*}
        \mcM(\omega) = \tr_{F} \Bigg[\sum_{i=1}^{n} \ketbra{i}{i}_{F} \otimes \mcM_{i}\Big( (\bra{i} \otimes \id_{\tilde{R}})\omega_{F\tilde{R}} (\ket{i} \otimes \id_{\tilde{R}}) \Big) \Bigg] = \sum_{i=1}^{n} \mcM_{i}\Big(\omega_{\tilde{R}}^{i}\Big),
    \end{equation*}
    where $\omega^{i}_{\tilde{R}} = (\bra{i} \otimes \id_{\tilde{R}})\omega_{F\tilde{R}} (\ket{i} \otimes \id_{\tilde{R}})$ is the $i^{\text{th}}$ diagonal block of $\omega_{F\tilde{R}}$. Notice that for every $i$ and every feasible state on the left hand side of \eqref{eq:chan_cond}, $\omega_{R_{i-1}\tilde{E}_{i-1}}$, there exists a state $\nu_{R\tilde{E}}^{i}$ such that $\mcM(\nu^{i}) = \mcM_{i}(\omega)$, namely, $\nu^{i} = \ketbra{i}{i} \otimes \omega_{\tilde{R}\tilde{E}}$, where we replaced $\tilde{R}_{i-1}\tilde{E}_{i-1}$ with $\tilde{R}\tilde{E}$ since they are isomorphic. We therefore find that \cref{eq:chan_cond} is satisfied for all $f$, since for every $i$, the set of output states of $\mcM_{i}$ is contained in the set of output states of $\mcM$. 
    
    Applying \cref{thm:big_eat}, we obtain  
    \begin{equation*}
        H_{\alpha}^{\uparrow}(A^{n}C^{n}|B^{n}E)_{\rho^{|\Omega}} \geq n \, h_{\alpha} - \frac{\alpha}{\alpha - 1} \log \Big(\frac{1}{p_{\Omega}}\Big)
    \end{equation*}
    where 
    \begin{equation}
        h_{\alpha} = \inf_{\omega \in \mcD(R\tilde{E})}\inf_{v_{C} \in \Delta_{\Omega}} \sup_{q_{BT} \in \mcP_{BT}}\Bigg(\frac{1}{\alpha - 1}D\big(v_{C}\|p^{\alpha}\big)
        - \sum_{c \in \mathrm{supp}(p_{C})} v_{C}(c)\,D_{\alpha}(\rho_{ABT\tilde{E}}^{|c} \| \id_{A} \otimes \sigma_{BT\tilde{E}})\Bigg) \label{eq:h_term}
    \end{equation}
    where $\sum_{c} p_{C}(c)\ketbra{c}{c}_{C} \otimes \rho_{ABT\tilde{E}}^{|c} = \mcM(\omega_{R\tilde{E}})$ and $p^{\alpha} = [p_{C}(c)^{\alpha}]_{c}$. By writing $\omega_{\tilde{R}\tilde{E}}^i = p_{F}(i)\hat{\omega}_{\tilde{R}\tilde{E}}^i$ for normalized states $\hat{\omega}$, the output state $\mcM(\omega_{R\tilde{E}})$ is of the form
    \begin{equation*}
    \begin{aligned}
        \mcM(\omega_{R\tilde{E}}) &= \sum_{i=1}^{n}p_{F}(i)\,\mcM_{i}\Big(\hat{\omega}_{\tilde{R}\tilde{E}}^i\Big) \\
        &= (1-\gamma) \ketbra{0}{0}_{T} \otimes \sum_{i=1}^{n}p_{F}(i)\,\mcM^{\text{gen}}_{i}(\hat{\omega}_{\tilde{R}\tilde{E}}^{i})  
        + \gamma \, \ketbra{1}{1}_{T} \otimes \sum_{i=1}^{n}p_{F}(i)\,\mcM^{\text{test}}_{i}(\hat{\omega}_{\tilde{R}\tilde{E}}^{i}) \\
        &= (1-\gamma) \ketbra{0}{0}_{T} \otimes \sum_{ab} p_{B}^{\text{gen}}(b)\ketbra{a}{a}_{A} \otimes \ketbra{\perp}{\perp}_{C} \otimes \ketbra{b}{b}_{B} \otimes \bar{\mcM}^{a|b}(\omega_{F\tilde{R}E}) \\ 
        & \ \ \ \ \ \ \ \ + \gamma \, \ketbra{1}{1}_{T} \otimes \sum_{ab} p_{B}^{\text{test}}(b)\ketbra{a}{a}_{A} \otimes \ketbra{f(a,b)}{f(a,b)}_{C} \otimes \ketbra{b}{b}_{B} \otimes \bar{\mcM}^{a|b}(\omega_{F\tilde{R}E}),
    \end{aligned}
    \end{equation*}
    where we defined the CP maps $\bar{\mcM}^{a|b}:F\tilde{R}\to \tilde{R}$, \begin{equation*}
        \bar{\mcM}^{a|b}(\omega) = \tr_{F} \Bigg[\sum_{i=1}^{n} \ketbra{i}{i}_{F} \otimes \mcM_{i}^{a|b}\Big( (\bra{i} \otimes \id_{\tilde{R}})\omega_{F\tilde{R}} (\ket{i} \otimes \id_{\tilde{R}}) \Big) \Bigg] = \sum_{i=1}^{n}p_{F}(i)\,\mcM_{i}^{a|b}\Big(\hat{\omega}_{\tilde{R}\tilde{E}}^i\Big).
    \end{equation*}
    Since for every $b$, $\sum_{a}\tr[\bar{\mcM}^{a|b}(\omega)] = 1$, we see that $\mcM$ is an infrequent sampling channel according to \cref{eq:sample_ch}. Thus, rather than computing $h_{\alpha}$ for the particular choice $\mcM$ in \cref{eq:h_term}, we can take the worst case over all channels of the form \eqref{eq:sample_ch}. In the following, we define  $\rho_{\tilde{E}}^{a|b} = [\tr_{R'}\circ \mcM^{a|b}](\omega_{R\tilde{E}})$ and $p_{A|B}(a|b) = \tr[\rho_{\tilde{E}}^{a|b}]$.

    Next we turn our attention to the state $\sigma_{BT\tilde{E}}$. According to \cref{thm:big_eat}, it is of the form
    \begin{equation*}
        \sigma_{BT\tilde{E}} = \sum_{bt}q_{BT}(b,t) \ketbra{b}{b}_B \otimes \ketbra{t}{t}_{T} \otimes \rho_{\tilde{E}}^{|bt},
    \end{equation*}
     where the states $\rho_{\tilde{E}}^{|bt}$ are obtained by projecting the state $\mcM(\omega)$ onto $\ketbra{b}{b}_{B} \otimes \ketbra{t}{t}_{T}$, re-normalizing and tracing out $ACR'$. Applying this procedure to any state of the form \cref{eq:sample_ch}, we find $\rho_{\tilde{E}}^{|bt} = \sum_{a}[\tr_{R} \circ \mcM^{a|b}](\omega_{R\tilde{E}}) = \tr_{R}[\omega_{R\tilde{E}}] = \rho_{E}$. We will next make the choice 
     \begin{equation}
         q_{BT}(b,0) = q_{B}(b)\,(1-\gamma) \ \ \ \text{and} \ \ \ q_{BT}(b,1) = p_{B}^{\text{test}}(b) \, \gamma, \label{eq:qBT_choice}
     \end{equation}
    where $q_{B}(b) \in \Delta(B)$ is a variable free to optimize. Restricting the supremum in \cref{eq:h_term} to distributions of this form results in a lower bound. 
    
    Consider now the conditional states $\rho_{ABT\tilde{E}}^{|c}$. For $c = \perp$, \cref{eq:sample_ch} implies 
     \begin{equation*}
         \rho_{ABT\tilde{E}}^{|\perp} = \sum_{ab} p_{B}^{\text{gen}}(b)\ketbra{a}{a}_{A} \otimes \ketbra{b}{b}_{B} \otimes \ketbra{0}{0}_{T} \otimes \rho_{\tilde{E}}^{a|b} = \ketbra{0}{0}_{T} \otimes [\tr_{R'C} \circ \mcM^{\text{gen}}](\omega_{R\tilde{E}}).
     \end{equation*}
     Then, noting that for the restricted form of $q_{BT}$ in \cref{eq:qBT_choice}, the remaining supremum over $q_{B}$ only applies to the divergence with $c = \perp$ in \cref{eq:h_term},
     \begin{equation}
         \begin{aligned}
             \sup_{q_{B} \in \Delta(B)}-D_{\alpha}(\rho_{ABT\tilde{E}}^{|\perp} \| \id_{A} \otimes \sigma_{BT\tilde{E}}) &= \sup_{q_{B} \in \Delta(B)}\frac{1}{1 - \alpha}\log \Bigg( \sum_{b} p_{B}^{\text{gen}}(b)^{\alpha}q_{BT}(b,0)^{1-\alpha} 2^{(\alpha - 1)D_{\alpha}(\rho_{A\tilde{E}}^{|b}\|\id_{A} \otimes \rho_{E}) }\Bigg)\\
             &= \log(1-\gamma) + \sup_{q_{B} \in \Delta(B)}\frac{1}{1 - \alpha}\log \Bigg( \sum_{b} p_{B}^{\text{gen}}(b)^{\alpha}q_{B}(b)^{1-\alpha} 2^{(1-\alpha)H_{\alpha}(A|\tilde{E})_{\rho^{|b}}}\Bigg) \\
             &= \log(1-\gamma) + H_{\alpha}(A|B^{\uparrow}\tilde{E}^{\downarrow})_{\mcM^{\text{gen}}(\omega)}, \label{eq:perp}
         \end{aligned}
     \end{equation}
     where we defined $\rho^{|b}_{A\tilde{E}} = \sum_{a}\ketbra{a}{a}_{A} \otimes \rho_{\tilde{E}}^{a|b}$ and for the last equality applied \cref{lem:var_ent}. For the case $c \in \{0,1\}^{d}$, 
    \begin{equation*}
        p_{C}(c) \, \rho_{ABT\tilde{E}}^{|c} = \gamma \, \ketbra{1}{1}_{T} \otimes \sum_{ab} \delta_{f(a,b),c} \, p_{B}^{\text{test}}(b) \ketbra{a}{a}_{A} \otimes \ketbra{b}{b}_{B} \otimes \rho_{\tilde{E}}^{a|b},
    \end{equation*}
    where $p_{C}(c) = \gamma \sum_{ab} p_{B}^{\text{test}}(b)\, \delta_{f(a,b),c}\, p_{AB}(a|b)$. If we write $\rho_{BC} = [\tr_{AT\tilde{E}}\circ \mcM](\omega_{R\tilde{E}}) = \sum_{bc} p_{CB}(c,b)\ketbra{c}{c}_{C}\otimes \ketbra{b}{b}$, we find for $c \in \{0,1\}^{d}$
    \begin{equation*}
        p_{CB}(c,b) = \gamma \, p_{B}^{\text{test}}(b) \sum_{a}\delta_{f(a,b),c} \, p_{AB}(a|b) = p_{B}^{\text{test}}(b) \, p_{C|B}(c|b).
    \end{equation*}
    We can also define $p_{B|C}(b|c) = p_{CB}(c,b)/p_{C}(c)$. Hence, using Bayes' rule, $\rho_{ABT\tilde{E}}^{|c}$ is equal to the following cq-state:
    \begin{equation*}
        \rho_{ABT\tilde{E}}^{|c} = \ketbra{1}{1}_{T} \otimes \sum_{b} p_{B|C}(b|c) \ketbra{b}{b}_{B} \otimes \rho_{A\tilde{E}}^{|c,b}, \ \ \ \rho_{A\tilde{E}}^{|c,b} = \frac{\gamma}{p_{C|B}(c|b)}\sum_{a} \delta_{f(a,b),c} \ketbra{a}{a}_{A} \otimes \rho_{\tilde{E}}^{a|b}.
    \end{equation*}
    Inserting this expression into the divergence,
     \begin{equation}
         \begin{aligned}
             -D_{\alpha}(\rho_{ABT\tilde{E}}^{|c} \| \id_{A} \otimes \sigma_{BT\tilde{E}}) &= \frac{1}{1 - \alpha}\log \Bigg( \sum_{b} p_{B|C}(b|c)^{\alpha} (\gamma \, p_{B}^{\text{test}}(b))^{1-\alpha} 2^{(\alpha - 1)D_{\alpha}(\rho_{A\tilde{E}}^{|b,c}\|\id_{A} \otimes \rho_{\tilde{E}}) }\Bigg). \label{eq:div_term}
         \end{aligned}
     \end{equation}
     Let us define the following cq-state:
     \begin{equation*}
         \rho_{AC\tilde{E}}^{|b} = \sum_{c \in \{0,1\}^{d}} \frac{p_{C|B}(c|b)}{\gamma} \ketbra{c}{c}_{C} \otimes \rho_{A\tilde{E}}^{|b,c}.
     \end{equation*}
    Note that $\rho_{AC\tilde{E}}^{|b}$ satisfies
    \begin{equation*}
        \tr_{AC}[\rho_{AC\tilde{E}}^{|b}] = \sum_{c\in\{0,1\}^{d}}\sum_{a}\delta_{f(a,b),c}\, \rho_{\tilde{E}}^{a|b} = \sum_{a}\rho_{\tilde{E}}^{a|b} = \rho_{\tilde{E}}.
    \end{equation*}
    We can therefore apply \cref{claim:div_bnd} to the state $\rho_{AC\tilde{E}}^{|b}$ to obtain the bound
    \begin{equation*}
        2^{(\alpha - 1)D_{\alpha}(\rho_{A\tilde{E}}^{|b,c}\|\id_{A} \otimes \rho_{\tilde{E}}) } \leq \Big(\frac{p_{C|B}(c|b)}{\gamma}\Big)^{1-\alpha}.
    \end{equation*}
    Inserting this into \cref{eq:div_term}, we find
    \begin{equation*}
    \begin{aligned}
        \sum_{b} p_{B|C}(b|c)^{\alpha} (\gamma \, p_{B}^{\text{test}}(b))^{1-\alpha} \Big(\frac{p_{C|B}(c|b)}{\gamma}\Big)^{1-\alpha} &= \sum_{b} p_{B|C}(b|c)^{\alpha} \big(p_{B}^{\text{test}}(b) \, p_{C|B}(c|b)\big)^{1-\alpha} \\
        &= \sum_{b} p_{CB}(c,b) \Big(\frac{p_{B|C}(b|c)}{p_{CB}(c,b)}\Big)^{\alpha}\\
        &= \sum_{b} p_{CB}(c,b) \, p_{C}(c)^{-\alpha} \\
        &= p_{C}(c)^{1-\alpha}.
    \end{aligned}
    \end{equation*}
    We therefore obtain the lower bound
    \begin{equation}
        -D_{\alpha}(\rho_{ABT\tilde{E}}^{|c} \| \id_{A} \otimes \sigma_{BT\tilde{E}}) \geq \log(p_{C}(c)). \label{eq:notperp}
    \end{equation}

     Finally, combining \cref{eq:perp,eq:notperp} with \cref{eq:h_term},
     \begin{multline*}
             \frac{1}{\alpha - 1}D\big(v_{C}\|p^{\alpha}\big)
        - \sum_{c \in \mathrm{supp}(p_{C})} v_{C}(c)\,D_{\alpha}(\rho_{ABT\tilde{E}}^{|c} \| \id_{A} \otimes \sigma_{BT\tilde{E}}) \\ \geq  \frac{1}{\alpha - 1}D\big(v_{C}\|p^{\alpha}\big) + \sum_{c \in \text{supp}(p_{C})}v_C(c) \, \log(p_{C}(c)) + v_{C}(\perp)\,H_{\alpha}(A|B^{\uparrow}\tilde{E}^{\downarrow})_{\mcM^{\text{gen}}(\omega)} \\
        = \frac{1}{\alpha - 1}D\big(v_{C}\|p_{C}\big)  + v_{C}(\perp)\,H_{\alpha}(A|B^{\uparrow}\tilde{E}^{\downarrow})_{\mcM^{\text{gen}}(\omega)},
     \end{multline*}
     where in the last line we used the fact that
     \begin{equation*}
         \frac{1}{\alpha-1}v_{C}(c)\big( \log(v_{C}(c)) - \log(p_{C}(c)^{\alpha}) \big) + v_{C}(c) \log(p_{C}(c)) = \frac{1}{\alpha - 1} v_{C}(c)\big( \log(v_{C}(c)) - \log(p_{C}(c)) \big).
     \end{equation*}
     This completes the proof. 
\end{proof}

The missing proof of \cref{claim:div_bnd} is stated below. 

\begin{claim}
    Let $\rho_{ABE} = \sum_{a} p_{A}(a)\ketbra{a}{a} \otimes \rho_{BE}^{|a}$ be a cq-state where systems $A$ and $B$ are classical, and define $\rho_{E} = \tr_{AB}[\rho_{ABE}] = \sum_{a} p_{A}(a)\rho_{E}^{|a}$ where $\rho_{E}^{|a} = \tr_{B}[\rho_{BE}^{|a}]$. Then for all $a$ such that $p(a)>0$ and all $\alpha \in [1,\infty]$
    \begin{equation*}
        2^{(\alpha-1)D_{\alpha}(\rho_{BE}^{|a} \| \id_{B} \otimes \rho_{E})} \leq p_{A}(a)^{1-\alpha}. 
    \end{equation*} \label{claim:div_bnd}
\end{claim}
\begin{proof}
    We first note that $D_{\alpha}(\rho_{BE}^{|a} \| \id_{B} \otimes \rho_{E}) \leq D_{\infty}(\rho_{BE}^{|a} \| \id_{B} \otimes \rho_{E})$, where $D_{\infty}(\rho \| \sigma)$ is the max-divergence, defined by
    \begin{equation}
        D_{\infty}(\rho\|\sigma) = \inf \{\lambda \ : \ \rho \leq 2^{\lambda} \sigma \}. \label{eq:max_div}
    \end{equation}
    Since $\rho_{BE}^{|a}$ is a cq-state, $\rho_{BE}^{|a} \leq \id_{B} \otimes \rho_{E}^{|a}$. Furthermore, $p_{A}(a)\, \rho_{E}^{|a} \leq \rho_{E}$. Putting both operator inequalities together, we find
    \begin{equation*}
        \rho_{BE}^{|a} \leq 2^{-\log(p_{A}(a))} \id_{B} \otimes \rho_{E}.
    \end{equation*}
    Hence $\lambda = -\log(p_{A}(a))$ is a feasible point to \eqref{eq:max_div}, and we find $D_{\alpha}(\rho_{BE}^{|a} \| \id_{B} \otimes \rho_{E}) \leq D_{\infty}(\rho_{BE}^{|a} \| \id_{B} \otimes \rho_{E}) \leq -\log(p_{A}(a))$. Inserting this upper bound into the expression $2^{(\alpha-1)D_{\alpha}(\rho_{BE}^{|a} \| \id_{B} \otimes \rho_{E})}$ establishes the claim. 
\end{proof}

\end{document}